\documentclass[reprint,amsmath,amssymb,aps,onecolumn,pra,nofootinbib,longbibliography,floatfix,superscriptaddress]{revtex4-2}

\usepackage{dcolumn}% Align table columns on decimal point
\usepackage{bm}% bold math

\usepackage{relsize}
\usepackage[utf8]{inputenc}
\usepackage[english]{babel}
\usepackage[T1]{fontenc}
\usepackage{graphicx}
\usepackage{tikz}
\usepackage{pgfplotstable}
\usepgfplotslibrary{colorbrewer}
\usetikzlibrary{decorations.pathreplacing}
\usetikzlibrary{cd}
\usepackage[caption=false]{subfig}
\usepackage{ragged2e} % for the \justifying macro
\DeclareCaptionJustification{justified}{\justifying}
\usepackage{nicematrix}

\usepackage{amsthm}
\usepackage{xcolor}
\usepackage[normalem]{ulem}

\usepackage{hyperref}
\hypersetup{%
  colorlinks=true,% hyperlinks will be coloured
  linkcolor=blue%,% hyperlink text will be green
}

\usepackage[linesnumbered, ruled, vlined]{algorithm2e}
\SetKwInput{KwInput}{Input}                % Set the Input
\SetKwInput{KwOutput}{Output}              % set the Output

\usepackage[capitalise]{cleveref}
 
\newtheorem{theorem}{Theorem}

\newtheorem{lemma}[theorem]{Lemma}

\newtheorem{example}[theorem]{Example}

\DeclareMathOperator{\IM}{im}

\def\thickeningHastingsa{
    \tikz[baseline={(current bounding box.center)}]{
        \node[circle, draw, fill = none, align = center] (a) {};
        \node[circle, draw, fill = none, align = center, right = 30 pt] (b) at (a) {};
        \node[circle, draw, fill = none, align = center, right = 30 pt] (c) at (b) {};
        \node[circle, draw, fill = none, align = center, right = 30 pt] (d) at (c) {};

        \node[rectangle, draw, fill = none, align = center, above = 50pt, right = 15pt] (e) at (b) {};

        \node[rectangle, draw, fill, align = center, below = 50pt, right = 15pt] (f) at (a) {};
        \node[rectangle, draw, fill, align = center, below = 50pt, right = 15pt] (g) at (b) {};

        \draw[-, thick] (a) -- (e);
        \draw[-, thick] (b) -- (e);
        \draw[-, thick] (c) -- (e);
        \draw[-, thick] (d) -- (e);
        \draw[-, thick] (f) -- (a);
        \draw[-, thick] (f) -- (b);
        \draw[-, thick] (g) -- (a);
        \draw[-, thick] (g) -- (c);
    }
}

\def\thickeningHastingsbX{
    \tikz[baseline={(current bounding box.center)}]{
        \node[rectangle, draw, fill = none, align = center] (a) {};
        \node[rectangle, draw, fill = none, align = center, right = 60pt] (b) at (a) {};
        \node[rectangle, draw, fill = none, align = center, right = 60pt] (c) at (b) {};

        \node[circle, draw, fill = none, align = center, below = 50pt, left = 10pt] (e) at (a) {};
        \node[circle, draw, fill = none, align = center, left = 10pt] (f) at (e) {};
        \node[circle, draw, fill = none, align = center, left = 10pt] (g) at (f) {};
        
        \node[circle, draw, fill = none, align = center, below = 50pt, right = 15pt] (h) at (a) {};
        \node[circle, draw, fill = none, align = center, right = 10pt] (i) at (h) {};
        \node[circle, draw, fill = none, align = center, right = 10pt] (j) at (i) {};
        
        \node[circle, draw, fill = none, align = center, below = 50pt, right = 15pt] (k) at (b) {};
        \node[circle, draw, fill = none, align = center, right = 10pt] (l) at (k) {};
        \node[circle, draw, fill = none, align = center, right = 10pt] (m) at (l) {};
        
        \node[circle, draw, fill = none, align = center, below = 50pt, right = 15pt] (n) at (c) {};
        \node[circle, draw, fill = none, align = center, right = 10pt] (o) at (n) {};
        \node[circle, draw, fill = none, align = center, right = 10pt] (p) at (o) {};

        \node[circle, draw, fill = none, align = center, above = 30pt, right = 30pt] (q) at (a) {};
        \node[circle, draw, fill = none, align = center, above = 30pt, right = 30pt] (s) at (b) {};

        \draw[-, thick] (a) -- (q);
        \draw[-, thick] (b) -- (q);
        \draw[-, thick] (b) -- (s);
        \draw[-, thick] (c) -- (s);

        \draw[-, thick] (a) -- (g);
        \draw[-, thick] (a) -- (h);
        \draw[-, thick] (a) -- (k);
        \draw[-, thick] (a) -- (n);

        \draw[-, thick] (b) -- (f);
        \draw[-, thick] (b) -- (i);
        \draw[-, thick] (b) -- (l);
        \draw[-, thick] (b) -- (o);

        \draw[-, thick] (c) -- (e);
        \draw[-, thick] (c) -- (j);
        \draw[-, thick] (c) -- (m);
        \draw[-, thick] (c) -- (p);

        \draw[decorate, decoration = {brace, mirror, amplitude = 10pt}, xshift = 170pt, yshift = -27pt] (0.5, 0.9) -- (0.5, 2.1) node[black, midway, xshift = 33pt] {$I_{n_X} \otimes H^{\mathrm T}_\ell$};
        \draw[decorate, decoration = {brace, amplitude = 10pt}, xshift = 170pt, yshift = -27pt] (0.5, 0.8) -- (0.5, -0.9) node[black, midway, xshift = 30pt] {$H_X \otimes I_\ell$};
    }
}

\def\thickeningHastingsbZ{
    \tikz[baseline={(current bounding box.center)}]{
        \node[circle, draw, fill = none, align = center] (a) {};
        \node[circle, draw, fill = none, align = center, right = 20pt] (b) at (a) {};
        \node[circle, draw, fill = none, align = center, right = 20pt] (c) at (b) {};
        \node[circle, draw, fill = none, align = center, right = 20pt] (d) at (c) {};
        \node[circle, draw, fill = none, align = center, right = 20pt] (e) at (d) {};
        \node[circle, draw, fill = none, align = center, right = 20pt] (f) at (e) {};
        \node[circle, draw, fill = none, align = center, right = 20pt] (g) at (f) {};
        \node[circle, draw, fill = none, align = center, right = 20pt] (h) at (g) {};
        \node[circle, draw, fill = none, align = center, right = 20pt] (i) at (h) {};
        \node[circle, draw, fill = none, align = center, right = 20pt] (j) at (i) {};
        \node[circle, draw, fill = none, align = center, right = 20pt] (k) at (j) {};
        \node[circle, draw, fill = none, align = center, right = 20pt] (l) at (k) {};

        \node[rectangle, draw, fill, align = center, below = 30pt] (m) at (b) {};
        \node[rectangle, draw, fill, align = center, below = 30pt] (n) at (c) {};
        \node[rectangle, draw, fill, align = center, below = 30pt] (o) at (d) {};
        \node[rectangle, draw, fill, align = center, below = 30pt] (p) at (f) {};
        \node[rectangle, draw, fill, align = center, below = 30pt] (q) at (g) {};
        \node[rectangle, draw, fill, align = center, below = 30pt] (r) at (h) {};

        \node[rectangle, draw, fill, align = center, above = 20pt, right = 10pt] (s) at (a) {};
        \node[rectangle, draw, fill, align = center, above = 20pt, right = 10pt] (t) at (b) {};
        \node[rectangle, draw, fill, align = center, above = 20pt, right = 10pt] (u) at (d) {};
        \node[rectangle, draw, fill, align = center, above = 20pt, right = 10pt] (v) at (e) {};
        \node[rectangle, draw, fill, align = center, above = 20pt, right = 10pt] (w) at (g) {};
        \node[rectangle, draw, fill, align = center, above = 20pt, right = 10pt] (x) at (h) {};
        \node[rectangle, draw, fill, align = center, above = 20pt, right = 10pt] (y) at (j) {};
        \node[rectangle, draw, fill, align = center, above = 20pt, right = 10pt] (z) at (k) {};

        \node[circle, draw, fill = none, align = center, above = 30pt] (aa) at (u) {};
        \node[circle, draw, fill = none, align = center, above = 30pt] (ab) at (x) {};

        \draw[-, thick] (a) -- (s);
        \draw[-, thick] (s) -- (b);
        \draw[-, thick] (b) -- (t);
        \draw[-, thick] (t) -- (c);

        \draw[-, thick] (d) -- (u);
        \draw[-, thick] (u) -- (e);
        \draw[-, thick] (e) -- (v);
        \draw[-, thick] (v) -- (f);

        \draw[-, thick] (g) -- (w);
        \draw[-, thick] (w) -- (h);
        \draw[-, thick] (h) -- (x);
        \draw[-, thick] (x) -- (i);

        \draw[-, thick] (j) -- (y);
        \draw[-, thick] (y) -- (k);
        \draw[-, thick] (k) -- (z);
        \draw[-, thick] (z) -- (l);

        \draw[-, thick] (s) -- (aa);
        \draw[-, thick] (u) -- (aa);
        \draw[-, thick] (w) -- (aa);
        \draw[-, thick] (y) -- (aa);
        \draw[-, thick] (t) -- (ab);
        \draw[-, thick] (v) -- (ab);
        \draw[-, thick] (x) -- (ab);
        \draw[-, thick] (z) -- (ab);

        \draw[-, thick] (a) -- (m);
        \draw[-, thick] (d) -- (m);
        \draw[-, thick] (b) -- (n);
        \draw[-, thick] (e) -- (n);
        \draw[-, thick] (c) -- (o);
        \draw[-, thick] (f) -- (o);
        \draw[-, thick] (a) -- (p);
        \draw[-, thick] (g) -- (p);
        \draw[-, thick] (b) -- (q);
        \draw[-, thick] (h) -- (q);
        \draw[-, thick] (c) -- (r);
        \draw[-, thick] (i) -- (r);

        \draw[decorate, decoration = {brace, mirror, amplitude = 10pt}, xshift = 265pt, yshift = 0pt] (0.5, 0.9) -- (0.5, 1.9) node[black, midway, xshift = 33pt] {$H^{\mathrm T}_X \otimes I_{\ell - 1}$};
        \draw[decorate, decoration = {brace, mirror, amplitude = 10pt}, xshift = 265pt, yshift = 0pt] (0.5, -0.1) -- (0.5, 0.8) node[black, midway, xshift = 30pt] {$I_n \otimes H_\ell$};
        \draw[decorate, decoration = {brace, amplitude = 10pt}, xshift = 265pt, yshift = 0pt] (0.5, -0.2) -- (0.5, -1.2) node[black, midway, xshift = 30pt] {$H_Z \otimes I_\ell$};
    }
}

\def\chHastingsZ{
    \tikz[baseline={(current bounding box.center)}]{
        \node[circle, draw, fill = none, align = center] (a) {};
        \node[circle, draw, fill = none, align = center, right = 20pt] (b) at (a) {};
        \node[circle, draw, fill = none, align = center, right = 20pt] (c) at (b) {};
        \node[circle, draw, fill = none, align = center, right = 20pt] (d) at (c) {};
        \node[circle, draw, fill = none, align = center, right = 20pt] (e) at (d) {};
        \node[circle, draw, fill = none, align = center, right = 20pt] (f) at (e) {};
        \node[circle, draw, fill = none, align = center, right = 20pt] (g) at (f) {};
        \node[circle, draw, fill = none, align = center, right = 20pt] (h) at (g) {};
        \node[circle, draw, fill = none, align = center, right = 20pt] (i) at (h) {};
        \node[circle, draw, fill = none, align = center, right = 20pt] (j) at (i) {};
        \node[circle, draw, fill = none, align = center, right = 20pt] (k) at (j) {};
        \node[circle, draw, fill = none, align = center, right = 20pt] (l) at (k) {};

        \node[rectangle, draw, fill, align = center, below = 30pt] (m) at (b) {};
        \node[rectangle, draw, fill, align = center, below = 30pt] (q) at (g) {};

        \node[rectangle, draw, fill, align = center, above = 20pt, right = 10pt] (s) at (a) {};
        \node[rectangle, draw, fill, align = center, above = 20pt, right = 10pt] (t) at (b) {};
        \node[rectangle, draw, fill, align = center, above = 20pt, right = 10pt] (u) at (d) {};
        \node[rectangle, draw, fill, align = center, above = 20pt, right = 10pt] (v) at (e) {};
        \node[rectangle, draw, fill, align = center, above = 20pt, right = 10pt] (w) at (g) {};
        \node[rectangle, draw, fill, align = center, above = 20pt, right = 10pt] (x) at (h) {};
        \node[rectangle, draw, fill, align = center, above = 20pt, right = 10pt] (y) at (j) {};
        \node[rectangle, draw, fill, align = center, above = 20pt, right = 10pt] (z) at (k) {};

        \node[circle, draw, fill = none, align = center, above = 30pt] (aa) at (u) {};
        \node[circle, draw, fill = none, align = center, above = 30pt] (ab) at (x) {};

        \draw[-, thick] (a) -- (s);
        \draw[-, thick] (s) -- (b);
        \draw[-, thick] (b) -- (t);
        \draw[-, thick] (t) -- (c);

        \draw[-, thick] (d) -- (u);
        \draw[-, thick] (u) -- (e);
        \draw[-, thick] (e) -- (v);
        \draw[-, thick] (v) -- (f);

        \draw[-, thick] (g) -- (w);
        \draw[-, thick] (w) -- (h);
        \draw[-, thick] (h) -- (x);
        \draw[-, thick] (x) -- (i);

        \draw[-, thick] (j) -- (y);
        \draw[-, thick] (y) -- (k);
        \draw[-, thick] (k) -- (z);
        \draw[-, thick] (z) -- (l);

        \draw[-, thick] (s) -- (aa);
        \draw[-, thick] (u) -- (aa);
        \draw[-, thick] (w) -- (aa);
        \draw[-, thick] (y) -- (aa);
        \draw[-, thick] (t) -- (ab);
        \draw[-, thick] (v) -- (ab);
        \draw[-, thick] (x) -- (ab);
        \draw[-, thick] (z) -- (ab);

        \draw[-, thick] (a) -- (m);
        \draw[-, thick] (d) -- (m);
        \draw[-, thick] (b) -- (q);
        \draw[-, thick] (h) -- (q);
    }
}

\begin{document}

\preprint{APS/123-QED}

\title{Weight Reduced Stabilizer Codes with Lower Overhead}

\author{Eric Sabo}
\thanks{These authors contributed equally. Corresponding author: \href{mailto:eric.sabo@xanadu.ai}{eric.sabo@xanadu.ai}.}
\affiliation{Xanadu, Toronto, Ontario M5G 2C8, Canada}
\author{Lane G. Gunderman}
\thanks{These authors contributed equally. Corresponding author: \href{mailto:eric.sabo@xanadu.ai}{eric.sabo@xanadu.ai}.}
% \email{lanegunderman@gmail.com}
\affiliation{Xanadu, Toronto, Ontario M5G 2C8, Canada}
\author{Benjamin Ide}
\thanks{These authors contributed equally. Corresponding author: \href{mailto:eric.sabo@xanadu.ai}{eric.sabo@xanadu.ai}.}
\affiliation{Xanadu, Toronto, Ontario M5G 2C8, Canada}
\author{Michael Vasmer}
\affiliation{Xanadu, Toronto, Ontario M5G 2C8, Canada}
\affiliation{Perimeter Institute for Theoretical Physics, Waterloo, ON N2L 2Y5, Canada}
\affiliation{Institute for Quantum Computing, University of Waterloo, Waterloo, ON N2L 3G1, Canada}
\author{Guillaume Dauphinais}
\affiliation{Xanadu, Toronto, Ontario M5G 2C8, Canada}

\date{\today}

\begin{abstract}
	Stabilizer codes are the most widely studied class of quantum error-correcting codes and form the basis of most proposals for a fault-tolerant quantum computer.  A stabilizer code is defined by a set of parity-check operators, which are measured in order to infer information about errors that may have occurred. In typical settings, measuring these operators is itself a noisy process and the noise strength scales with the number of qubits involved in a given parity check, or its weight. Hastings proposed a method for reducing the weights of the parity checks of a stabilizer code, though it has previously only been studied in the asymptotic regime. Here, we instead focus on the regime of small-to-medium size codes suitable for quantum computing hardware. We provide both a fully explicit description of Hastings's method and propose a substantially simplified weight reduction method that is applicable to the class of quantum product codes. Our simplified method allows us to reduce the check weights of hypergraph and lifted product codes to at most six, while preserving the number of logical qubits and at least retaining (in fact often increasing) the code distance. The price we pay is an increase in the number of physical qubits by a constant factor, but we find that our method is much more efficient than Hastings's method in this regard. We benchmark the performance of our codes in a photonic quantum computing architecture based on GKP qubits and passive linear optics, finding that our weight reduction method substantially improves code performance.
\end{abstract}

\maketitle
\setcounter{MaxMatrixCols}{30}

\section{Introduction}\label{sec:intro}
Quantum error correction (QEC) is believed to be necessary in order to run large-scale quantum algorithms~\cite{campbell2017, dalzell2023}. Recent years have seen remarkable progress in realizing QEC codes on various hardware platforms~\cite{ryan-anderson2021, krinner2022, sundaresan2023, googlequantumai2023}, with most experiments thus far demonstrating a particular QEC code called the surface code~\cite{kitaev2003, dennis2002, fowler2012}. 

In a stabilizer code (a widely-studied class of QEC codes), one measures parity-check operators that give partial information about errors that may have occurred on the physical qubits of the code. However, in a realistic setting the act of measuring the parity-check operators is itself a noisy process. In many hardware platforms, the noise associated with measuring a parity check scales with the weight of the check (the number of qubits that the check acts on non-trivially), as higher check weights correspond to deeper measurement circuits. Hence, one way to find practically useful stabilizer codes is to search for stabilizer codes with low-weight parity-checks, which are known as quantum low-density parity-check (qLDPC) codes~\cite{breuckmann2021a}.

The surface code is an example of a qLDPC code, as all of its checks are weight four. However, the surface code suffers from the drawback that it always encodes a single logical qubit (no matter its size). This low encoding rate means that for realistic noise rates we expect to have approximately $1,000$ physical qubits per logical qubit, which leads to estimates of millions of physical qubits being required for large-scale quantum algorithms~\cite{gidney2021, kim2022}. There exist other families of qLDPC codes with higher encoding rates than the surface code~\cite{breuckmann2021a}, though this often comes with slightly higher parity-check weights and the requirement of long-range connectivity~\cite{baspin2022, baspin2022a, baspin2023}. Nevertheless, the potential of these qLDPC codes has lead to recent proposals for implementing them in a variety of hardware platforms~\cite{bravyi2023high, xu2023constant}, and is particularly well suited to photonic architectures based on GKP qubits where there are minimal constraints on the locality of qubit connectivity~\cite{bourassa2021Blueprint, tzitrin2021Fault}.

Given a qLDPC code with favorable parameters (e.g.\ high encoding rate) but parity-check weights that are high enough to limit its performance under realistic noise assumptions, one can ask if there is any way to reduce the weights of the parity-checks while (mostly) retaining the favorable parameters of the code. Hastings~\cite{hastings2016, hastings2021quantum} provided a method known as weight reduction, which takes an input CSS code and can output a qLDPC code with $O(1)$  parity-check weights, while increasing the number of physical qubits by a constant factor and reducing the code distance by a constant factor. This method has thus far been applied mostly in the asymptotic regime~\cite{hastings2021quantum, wills2023tradeoff} where one is primarily interested in the scaling of properties such as the encoding rate of a family of codes as a function of the code size. 

Here, we investigate the utility of weight reduction techniques for modifying qLDPC codes of small-to-medium size that could potentially be implemented on hardware in the short-to-medium term. We first provide a self-contained presentation of Hastings's weight reduction procedure, providing an alternative, algorithmic perspective while also optimizing the procedure to improve its overhead. Next, we propose a new weight reduction technique for classical codes, the outputs of which can be used as input to quantum product code constructions~\cite{tillich2014quantum, panteleev2022, breuckmann2021, panteleev2022asymptotically}. We use these two techniques to construct qLDPC codes with low-weight parity-checks, finding that our method gives codes with better parameters in the finite-size regime of interest. In particular, for an input hypergraph product code with parameters $[\![45, 9, 3]\!]$, weight seven checks, and qubit degree four, our method gives a $[\![65, 9, 4]\!]$ code with weight six checks and qubit degree three, whereas Hastings's method gives a $[\![2892, 9, 5]\!]$ code with weight six checks and qubit degree six.

The default circuit for measuring a parity check in a circuit-based quantum computer is to introduce an ancilla qubit, apply controlled Pauli gates from the ancilla to the qubits in the support of the check, and then measure the ancilla~\cite{nielsen00}. If all the operations in the circuit are noisy, then the noise in this process scales with the weight of the parity check. In a measurement-based quantum computer the situation is similar: for each measurement of a weight $w$ parity check we introduce a degree-$w$ node in the cluster state\footnote{We note that the degree of the node may be greater than $w$ if the check comprises more than one type of Pauli operator.}, where the node represents a qubit prepared in the $|+\rangle$ state and each edges represents a control-$Z$ between this qubit and another qubit~\cite{raussendorf2006Afault, bolt2016Foliation, roberts2020Universal}. In Xanadu's architecture, cluster states are constructed from two-qubit entangled GKP states~\cite{gottesman2001Encoding} and $N$-body continuous-variable GHZ measurements~\cite{bourassa2021Blueprint, tzitrin2021Fault, xanadu2024inprep}.  In this architecture, the strength of the effective noise acting on a qubit in the cluster state scales (to first order) as $1 - \mathrm{erf}\left[c / \sqrt{ 2\sigma^{2} N }\right]$, where $\sigma^2$ is the variance of a single phase space peak of a GKP state's Wigner function (which can be related to the total transmissivity of the system~\cite{tzitrin2021Fault}), $N$ is the number of neighbors of the qubit, $c$ is a positive constant, and $\mathrm{erf}$ is the error function. Therefore, if we reduce the check weights, then we reduce the effective qubit-level noise in the cluster state for a fixed GKP state quality (or amount of optical loss).

We benchmark the performance of our weight-reduced codes using Monte Carlo simulations of the architecture described above as a quantum memory. We find that our weight reduction technique substantially improves the performance for cluster states constructed from hypergraph product codes and lifted product codes, in terms of both the logical error rates and the break-even point. Although we benchmarked our codes using a specific noise model relevant to photonic hardware based on GKP qubits, we would expect that similar results hold for other hardware platforms where the noise associated with measuring a parity-check scales with its weight.

The remainder of this paper is structured as follows. In \cref{sec:background}, we review the relevant background in classical and quantum coding theory, and we provide additional background material on homological algebra in \cref{app:homological}. In \cref{sec:revHast}, we describe the steps of Hastings's weight reduction method, discuss optimizations to reduce its overhead, and comment on its implications for iterative decoding. In \cref{sec:classicalWtR}, we present our alternative weight reduction method for classical codes and show how it can be applied to reduce the stabilizer weights of quantum product codes. In \cref{sec:num}, we present examples of weight-reduced codes and the results of our numerical simulations. We conclude in \cref{sec:conc}.

\section{Background \& Notation}\label{sec:background}
Throughout this paper, we assume familiarity with the stabilizer formalism~\cite{gottesman1997Stabilizer} of QEC codes and with the basics of classical coding theory~\cite{macwilliams1978Theory}.  For simplicity, we exclusively work with $\mathbb{F}_2$ vector spaces, although many of the ideas are easily extended to higher fields and do not depend on this choice. Missing entries in matrices are assumed to be zero, and horizontal and vertical dividing lines within matrices are added throughout solely for visual convenience. Let $H$ be a full-rank, $(n - k) \times n$ matrix. The $[n, k, d]$ (classical) linear code $\mathcal{C} = \mathcal{C}(H)$ associated with $H$ is the subspace orthogonal to the row space of $H$ with respect to the standard (Euclidean) inner product, $\mathcal{C} = \{ v \in \mathbb{F}^n_2 \mid Hv^{\mathrm{T}} = 0 \}$, where $v = (v_1, \hdots, v_n)$ is thought of as a row vector and superscript `$\mathrm{T}$' is the standard matrix transpose. The support of $v$ is the set $\mathrm{supp}(v) = \{i \mid v_i \neq 0\}$. In this context, $H$ is called the parity-check matrix of $\mathcal{C}$ and serves a role similar to the stabilizers in QEC. The minimum distance of the code is given by $d = \min \{\mathrm{wt}(v) \mid 0 \neq v \in \mathcal{C}\}$, where $\mathrm{wt} = |\mathrm{supp}(v)|$ denotes the Hamming weight.

An $[\![n, k, d]\!]$ Pauli stabilizer code is described numerically by its stabilizer matrix in symplectic form whose first $n$ columns denote Pauli $X$ operators and subsequent $n$ columns Pauli denote $Z$ operators. CSS codes~\cite{calderbank1996Good, steane1996Multiple} can be described by two matrices $H_X$ and $H_Z$ of $n$ columns comprising of the $X$ and $Z$ stabilizer generators, respectively, such that the stabilizer matrix is of the form $H_X \oplus H_Z$. The inputs of a procedure acting on a set of stabilizers appear as $H_X$ and $H_Z$ and the outputs appear with tildes, $\tilde{H}_X$ and $\tilde{H}_Z$. Following Refs.~\cite{hastings2021fiber, hastings2021quantum}, let $n_X$ and $n_Z$ be the number of $X$ and $Z$ stabilizers, respectively, $w_X$ and $w_Z$ be the maximum Hamming weight of the $X$ and $Z$ stabilizer generators, respectively, and $q_X$ and $q_Z$ be the maximum Hamming weight of the columns of $H_X$ and $H_Z$, respectively. The parameters $q_X$ and $q_Z$  are sometimes known as the $X$- and $Z$-qubit degrees, respectively, as, for example, $q_X$ denotes the maximum number of stabilizer generators that have nontrivial support on the same qubit. Note that these parameters are not inherent to the code but are relative to the specific form of the generators chosen.

Hypergraph product codes~\cite{tillich2014quantum}, $\mathrm{HGP}(H_1,H_2)$, are CSS codes constructed from two parity-check matrices $H_1$ and $H_2$ with stabilizers
\begin{equation}\label{hgp}
    H_X = \begin{pmatrix} H_1 \otimes I & I \otimes H^{\mathrm{T}}_2 \end{pmatrix}
    \quad , \quad
    H_Z = \begin{pmatrix} I \otimes H_2 & H^{\mathrm{T}}_1 \otimes I\end{pmatrix}.
\end{equation}
If $\mathcal{C}(H_i)$ has parameters $[n_i,k_i,d_i]$ and $\mathcal{C}(H_i^{\mathrm{T}})$ has parameters $[m_i,k^{\mathrm{T}}_i,d^{\mathrm{T}}_i]$, where $k^{\mathrm{T}}_i$ and $d^{\mathrm{T}}_i$ are the dimension and distance of $\mathcal{C}(H^{\mathrm{T}})$, respectively, then $\mathrm{HGP}(H_1,H_2)$ has parameters
\begin{equation}\label{eq:hgp-params}
    [\![n_1n_2 + m_1m_2, k_1k_2 + k_1^{\mathrm{T}}k_2^{\mathrm{T}}, \min(d_1,d_2,d_1^{\mathrm{T}},d_2^{\mathrm{T}})]\!].
\end{equation}

Let $R_\ell = \mathbb{F}_2[x] / (x^\ell - 1)$ be a polynomial quotient ring. A circulant matrix is a square matrix specified by the first row or column, where each subsequent row (column) is cyclically shifted to the right (down) by one index. An element $g(x) = g_0 + g_1 x + \hdots + g_{\ell - 1} x^{\ell - 1} \in R_\ell$ is associated with the $\ell \times \ell$ circulant matrix, $\mathbb{B}(g(x))$, whose first column is given by the coefficients of $g(x)$.\footnote{Assigning the coefficients to the column instead of the first row has precedent in the classical error correction and mathematical literatures but is opposite of recent convention used in the quantum literature. The difference comes down to left versus right multiplication in the ring of circulants.} This is called the lift of $g(x)$. The lift of a matrix $A \in M_{m\times n} (R_\ell)$ with elements in $R_\ell$ is the matrix $\mathbb{B}(A) \in M_{m \ell \times n \ell}$ constructed by replacing each element of $A$ with its lift. The matrix $A$ is called the base (or weight or protograph) matrix of the lift and $\ell$ is the lift size. 
For example, for $\ell = 2$,
\begin{equation}
    A = \begin{pmatrix}
        1 & x \\
        0 & 1+x \\
    \end{pmatrix},
    \quad
    \mathbb{B}(A) = 
    \begin{pmatrix}
        \mathbb{B}(1) & \mathbb{B}(x) \\
        \mathbb{B}(0) & \mathbb{B}(1+x) \\
    \end{pmatrix} =
    \begin{pNiceArray}{cc|cc}
        1 & 0 & 0 & 1 \\
        0 & 1 & 1 & 0 \\
        \hline
        0 & 0 & 1 & 1 \\
        0 & 0 & 1 & 1 \\
    \end{pNiceArray}.
\end{equation}

Although typically defined by the form of its generator matrix, here we define a quasi-cyclic code~\cite{chen1969some} to be the linear code defined by the parity-check matrix $H = \mathbb{B}(A)$. Quasi-cyclic lifted product codes~\cite{panteleev2021quantum} are a generalization of hypergraph product codes based on quasi-cyclic codes. Let $A_1 \in M_{m_1 \times n_1} (R_\ell)$ and $A_2 \in M_{m_2 \times n_2} (R_\ell)$ be base matrices and define
\begin{equation}\label{eq:lp}
    A_X = \begin{pmatrix} A_1 \otimes I & I \otimes A_2 \end{pmatrix}
    \quad , \quad
    A_Z = \begin{pmatrix} I \otimes A_2^{\mathrm{T}} & A_1^{\mathrm{T}} \otimes I\end{pmatrix},
\end{equation}
where the transpose of $g(x)$ is determined by the transpose of its lift: $g^{\mathrm{T}}(x) = g_0 + g_{\ell - 1} x + \hdots + g_1 x^{\ell - 1}$. The lifted product code $\mathrm{LP}(A_1, A_2)$ is the CSS code with parity-check matrices $H_X = \mathbb B(A_X)$ and $H_Z = \mathbb B (A_Z)$. The length of the code is $n = \ell (n_1 m_2 + n_2 m_1)$ but there are currently no general formulas for $k$ and $d$; however, lifted product codes often have superior parameters to hypergraph product codes of similar size~\cite{panteleev2021quantum, raveendran2022finite, roffe2023bias}.

We will often use the Tanner graph representation of linear codes and stabilizer codes. Check nodes will be denoted by rectangles and variable nodes by circles. For CSS codes, open rectangles will denote $X$ stabilizers and filled-in rectangles $Z$ stabilizers.

A chain complex $C$ is, for the sake of our purposes, an ordered sequence of vector spaces $\{C_i\}$ over $\mathbb{F}_2$ with maps $\partial_i$ between each ordered pair $\{C_{i - 1}, C_i\}$ such that $\partial_i \circ \partial_{i + 1} = 0$:
\begin{equation*}
    \cdots \to C_{i + 1} \xrightarrow{\partial_{i + 1}} C_i \xrightarrow{\partial_i} C_{i - 1} \to \cdots .
\end{equation*}
We will only be interested in chain complexes with a finite number of vector spaces. A chain complex with $\ell$ spaces is called an $\ell$-term chain complex. Since $\partial_i \circ \partial_{i + 1} = 0$, we have $\mathrm{im}\left(\partial_{i + 1}\right) \subseteq \ker\left(\partial_i\right)$. A sequence is said to be exact at $C_i$ if $\mathrm{im}\left(\partial_{i + 1}\right) = \ker\left(\partial_i\right)$ and is said to be exact if it is exact at each $C_i$. The $i$th homology group is defined as $H_i( \cdot ) = \ker \left(\partial_i\right) / \mathrm{im}\left(\partial_{i + 1}\right)$. The letter $H$ will also be used for parity-check and stabilizer matrices, but there should be no confusion as to the context. The dual of a chain complex is a cochain complex
\begin{equation*}
    \cdots \leftarrow C_{i + 1} \xleftarrow{\delta_{i+1}} C_i \xleftarrow{\delta_i} C_{i - 1} \leftarrow \cdots .
\end{equation*}
The corresponding dual of the homology groups are the cohomology groups $H^i( \cdot ) = \ker \delta_{i+1} / \mathrm{im} \, \delta_i$. In this work, $\delta_i = \partial^{\mathrm{T}}_i$, since we assume the standard basis.

There is a natural correspondence between codes and chain complexes. Let $H \in \mathbb{F}^{n - k \times n}_2$ be a parity-check matrix of an $[n, k, d]$-linear code. Then we can express this code as a chain complex
\begin{equation}\label{classicalchain}
    \mathbb{F}^n_2 \xrightarrow{H} \mathbb{F}^{n - k}_2.
\end{equation}
Any diagram consisting of a single map vacuously satisfies the definition of a chain complex.

Consider an $[\![n, k, d]\!]$ CSS code generated by $n_Z$ independent $Z$ stabilizers given by the matrix $H_Z$ and $n_X$ $X$ stabilizers given by $H_X$. Since $H_Z^{\mathrm{T}} H_X = 0$, we can treat this as the chain complex
\begin{equation}\label{CSSchain}
    \mathbb{F}^{n_Z}_2 \xrightarrow{H_Z^{\mathrm{T}}} \mathbb{F}^n_2 \xrightarrow{H_X} \mathbb{F}^{n_X}_2
\end{equation}
with the qubits in the center.\footnote{It is often assumed that the qubits are at $C_1$; this is true for this work but is not strictly necessary.} Conversely, we can derive a CSS code from any two consecutive boundary maps, setting $H_Z^{\mathrm{T}} = \partial_{i + 1}$ and $H_X = \partial_i$. The $Z$ logical operators commute with the $X$ stabilizers ($\ker H_X$) and are not $Z$ stabilizers ($\mathrm{im} \, H^{\mathrm{T}}_Z = \mathrm{rowspace}(H_Z)$), which make them elements of the first homology group $H_1$. The dual problem is the cochain
\begin{equation*}
    \mathbb{F}^{n_Z}_2 \xleftarrow{H_Z} \mathbb{F}^n_2 \xleftarrow{H^{\mathrm{T}}_X} \mathbb{F}^{n_X}_2,
\end{equation*}
which gives the $X$ logicals $H^1( \cdot ) = \ker H_Z / \mathrm{im} \, H^{\mathrm{T}}_X$.\footnote{There is a one-to-one correspondence between the usual set difference definition of logical operators and equivalence classes of the quotient.}

A detailed discussion of the tensor product of chain complexes, the mapping cone, and their respective homologies is given in \cref{app:homological}.

\section{Quantum Weight Reduction}\label{sec:revHast}
Reference~\cite{wills2023tradeoff} provides a good summary of Hastings's quantum weight reduction method~\cite{hastings2021quantum}. Rather than duplicating this work, we aim to complement it by providing a description of the method using diagrams and matrices without algebraic topology. Our notation is roughly aligned with~\cite{wills2023tradeoff}, which is a simplification of~\cite{hastings2021quantum}. In addition to reviewing previous work, we provide new insight into the method and discuss the subtleties of its implementation~\cite{Sabo_2021}.

The four steps of quantum weight reduction method are:
\begin{enumerate}
    \item Copying --- reduces $q_X$,
    \item Gauging --- reduces $w_X$,
    \item Thickening and choosing heights --- reduces $q_Z$,
    \item Coning --- reduces $w_Z$.
\end{enumerate}
We will examine each step independently, although they must be applied in this order as some care is required to avoid undoing the progress achieved in previous steps. While parameters from previous steps are indeed kept $O(1)$ in subsequent steps, the exact constants can be significant for constructing codes that are compatible with realistic architectures. Therefore, in this section we will focus on the explicit constants rather than on asymptotic results. The examples and figures provided are contrived to demonstrate a specific concept and are not meant to represent good stabilizer codes. Applying these operations to real codes produce large matrices and complicated Tanner graphs from which we believe it is difficult to discern the underlying structure.

\subsubsection*{Copying}
The goal of copying is to reduce $q_X$ to at most three. Start by making $q_X$ copies of each qubit. By this we mean to add $q_X - 1$ new qubits (initialized to zero) per original qubit; the value of each original qubit is not copied to the new qubits:
\begin{equation*}
    \begin{pmatrix}
        v_1 & v_2 & \hdots & v_n
    \end{pmatrix}
    \mapsto
    \begin{pNiceArray}{ccc|ccc|c|ccc}
        v_{1, 1} & \hdots & v_{1, q_X} & v_{2, 1} & \hdots & v_{2, q_X} & \hdots & v_{n, 1} & \hdots & v_{n, q_X}
    \end{pNiceArray}.
\end{equation*}
For every $X$ stabilizer of length $n$, make a new stabilizer of length $q_X n$ such that for every $v_i$ in the support of the stabilizer one of $\{v_{i, 1}, \hdots, v_{i, q_X}\}$ receives the value of $v_i$. If a stabilizer uses the qubit $v_{i, j}$, another stabilizer cannot use it. For example, valid copies of the stabilizer $\begin{pmatrix} 1 & 1 & 1 & 1 & 1 & 1\end{pmatrix}$ with $q_X = 3$ include
\begin{equation}\label{copyexample}
    \begin{pNiceArray}{ccc|ccc|ccc|ccc|ccc|ccc}
        1 & 0 & 0 & 1 & 0 & 0 & 1 & 0 & 0 & 1 & 0 & 0 & 1 & 0 & 0 & 1 & 0 & 0
    \end{pNiceArray}
\end{equation}
and
\begin{equation*}
    \begin{pNiceArray}{ccc|ccc|ccc|ccc|ccc|ccc}
        0 & 1 & 0 & 0 & 0 & 1 & 1 & 0 & 0 & 0 & 0 & 1 & 1 & 0 & 0 & 0 & 1 & 0
    \end{pNiceArray}.
\end{equation*}
As vectors, these are equivalent up to qubit permutations. For simplicity, this work always fills the columns from left to right, as in \cref{copyexample}. Suppose \cref{copyexample} is used and $\begin{pmatrix} 1 & 1 & 0 & 0 & 1 & 1\end{pmatrix}$ is another stabilizer. Then
\begin{equation*}
    \begin{pNiceArray}{ccc|ccc|ccc|ccc|ccc|ccc}
        1 & 0 & 0 & 0 & 1 & 0 & 0 & 0 & 0 & 0 & 0 & 0 & 0 & 1 & 0 & 0 & 1 & 0
    \end{pNiceArray}
\end{equation*}
is not a valid copy because the qubit $v_{1, 1}$ is already used by the first stabilizer. Note that every original $X$ stabilizer is kept, although now in a permuted form.

In addition to the copied stabilizers, $(q_X - 1)n$ new $X$ stabilizers are added to link the copies of $v_i$ such that they collectively ``behave'' like the original, single qubit. These are weight two and of the form $v_{i, j} v_{i, j + 1}$ for $1 \leq j \leq q_X - 1$. (These are not constrained by the ``validity'' concept required for the previous stabilizers.) For example,
\begin{gather*}
    \begin{pNiceArray}{ccc|ccc|ccc|ccc|ccc|ccc}
        1 & 1 & 0 & 0 & 0 & 0 & 0 & 0 & 0 & 0 & 0 & 0 & 0 & 0 & 0 & 0 & 0 & 0
    \end{pNiceArray},\\
    \begin{pNiceArray}{ccc|ccc|ccc|ccc|ccc|ccc}
        0 & 1 & 1 & 0 & 0 & 0 & 0 & 0 & 0 & 0 & 0 & 0 & 0 & 0 & 0 & 0 & 0 & 0
    \end{pNiceArray},\\
    \begin{pNiceArray}{ccc|ccc|ccc|ccc|ccc|ccc}
        0 & 0 & 0 & 1 & 1 & 0 & 0 & 0 & 0 & 0 & 0 & 0 & 0 & 0 & 0 & 0 & 0 & 0
    \end{pNiceArray},\\
    \begin{pNiceArray}{ccc|ccc|ccc|ccc|ccc|ccc}
        0 & 0 & 0 & 0 & 1 & 1 & 0 & 0 & 0 & 0 & 0 & 0 & 0 & 0 & 0 & 0 & 0 & 0
    \end{pNiceArray},\\
    \vdots
\end{gather*}
This imposes a classical repetition code on the copies; see \cref{fig:copying} below.

The commutativity with the $Z$ stabilizers is maintained on every copy by putting the value at $v_i$ at $v_{i, j}$ for all $1 \leq j \leq q_X$. For example, if $\begin{pmatrix}1 & 0 & 1 & 0 & 1 & 0\end{pmatrix}$ is a $Z$ stabilizer, then for $q_X = 3$, the new $Z$ stabilizer is
\begin{equation*}
    \begin{pNiceArray}{ccc|ccc|ccc|ccc|ccc|ccc}
        1 & 1 & 1 & 0 & 0 & 0 & 1 & 1 & 1 & 0 & 0 & 0 & 1 & 1 & 1 & 0 & 0 & 0
    \end{pNiceArray}.
\end{equation*}
This comes at the price of increasing $w_z$, which will be dealt with in subsequent steps.

\begin{example}
Copying the stabilizers
\begin{align*}
    H_X &= \begin{pmatrix}
        1 & 1 & 1 & 0 & 0 & 0\\
        1 & 1 & 0 & 0 & 1 & 1\\
        1 & 0 & 1 & 1 & 1 & 0\\
        1 & 0 & 0 & 0 & 0 & 1
    \end{pmatrix}\\
    H_Z &= \,\, \begin{pmatrix}
        1 & 0 & 1 & 0 & 0 & 1
    \end{pmatrix}
\end{align*}
gives
\begin{align}\label{tmatrixcopying}
    \tilde{H}_X &= \begin{pNiceArray}{cccc|cccc|cccc|cccc|cccc|cccc}[columns-width=auto,margin]
        1 &   &   &   & 1 &   &   &   & 1 &   &   &   &   &   &   &   &   &   &   &   &   &   &   &   \\
  & 1 &   &   &   & 1 &   &   &   &   &   &   &   &   &   &   & 1 &   &   &   & 1 &   &   &   \\
  &   & 1 &   &   &   &   &   &   & 1 &   &   & 1 &   &   &   &   & 1 &   &   &   &   &   &   \\
  &   &   & 1 &   &   &   &   &   &   &   &   &   &   &   &   &   &   &   &   &   & 1 &   &   \\
  \hline
1 & 1 &   &   &   &   &   &   &   &   &   &   &   &   &   &   &   &   &   &   &   &   &   &   \\
  & 1 & 1 &   &   &   &   &   &   &   &   &   &   &   &   &   &   &   &   &   &   &   &   &   \\
  &   & 1 & 1 &   &   &   &   &   &   &   &   &   &   &   &   &   &   &   &   &   &   &   &   \\
  \hline
  &   &   &   & 1 & 1 &   &   &   &   &   &   &   &   &   &   &   &   &   &   &   &   &   &   \\
  &   &   &   &   & 1 & 1 &   &   &   &   &   &   &   &   &   &   &   &   &   &   &   &   &   \\
  &   &   &   &   &   & 1 & 1 &   &   &   &   &   &   &   &   &   &   &   &   &   &   &   &   \\
  \hline
  &   &   &   &   &   &   &   & 1 & 1 &   &   &   &   &   &   &   &   &   &   &   &   &   &   \\
  &   &   &   &   &   &   &   &   & 1 & 1 &   &   &   &   &   &   &   &   &   &   &   &   &   \\
  &   &   &   &   &   &   &   &   &   & 1 & 1 &   &   &   &   &   &   &   &   &   &   &   &   \\
  \hline
  &   &   &   &   &   &   &   &   &   &   &   & 1 & 1 &   &   &   &   &   &   &   &   &   &   \\
  &   &   &   &   &   &   &   &   &   &   &   &   & 1 & 1 &   &   &   &   &   &   &   &   &   \\
  &   &   &   &   &   &   &   &   &   &   &   &   &   & 1 & 1 &   &   &   &   &   &   &   &   \\
  \hline
  &   &   &   &   &   &   &   &   &   &   &   &   &   &   &   & 1 & 1 &   &   &   &   &   &   \\
  &   &   &   &   &   &   &   &   &   &   &   &   &   &   &   &   & 1 & 1 &   &   &   &   &   \\
  &   &   &   &   &   &   &   &   &   &   &   &   &   &   &   &   &   & 1 & 1 &   &   &   &   \\
  \hline
  &   &   &   &   &   &   &   &   &   &   &   &   &   &   &   &   &   &   &   & 1 & 1 &   &   \\
  &   &   &   &   &   &   &   &   &   &   &   &   &   &   &   &   &   &   &   &   & 1 & 1 &   \\
  &   &   &   &   &   &   &   &   &   &   &   &   &   &   &   &   &   &   &   &   &   & 1 & 1 
    \end{pNiceArray}\\
    \tilde{H}_Z &= \,\, \begin{pNiceArray}{cccc|cccc|cccc|cccc|cccc|cccc}[columns-width=auto,margin]
        1 & 1 & 1 & 1 &   &   &   &   & 1 & 1 & 1 & 1 &   &   &   &   &   &   &   &   & 1 & 1 & 1 & 1 
    \end{pNiceArray}.
\end{align}
The parameters have transformed from $(w_X = 4, q_X = 4, w_Z = 3, q_Z = 1)$ to  $(\tilde{w}_X = 4, \tilde{q}_X = 3, \tilde{w}_Z = 12, \tilde{q}_Z = 1)$.
\end{example}

It follows from above that $\tilde{n} = q_X n$, $\tilde{n}_X = n_X + (q_X - 1) n$, $\tilde{n}_Z = n_Z$, $\tilde{w}_X = w_X$, $\tilde{q}_X = \min \{q_X, 3\}$, $\tilde{w}_Z = q_x w_Z$, and $\tilde{q}_Z = q_Z$. It follows that $\tilde{k} = \tilde{n} - \tilde{n}_X - \tilde{n}_Z = n - n_x - n_z = k$. The dimension is traditionally computed by counting the number of logical operators, but since we know the dimension already, we will write down $k$ independent operators that commute with the stabilizers and are therefore logical operators. A $Z$-logical operator must commute with both the old and new $X$ stabilizers. The old $X$ stabilizers are supported on qubits $v_{i, 1}$ for $1 \leq i \leq n$. Fix an $i$ in the overlap between an old $X$ stabilizer and old $Z$ logical. This does not commute with the new $X$ stabilizer $v_{i, 1} v_{i, 2}$. The only way to fix this is to extend the support of the $Z$ stabilizer to all of the copied qubits $\{v_{i, 1}, \hdots, v_{i, q_X}\}$. New $Z$ logicals are thus of the form $z \otimes \begin{pmatrix} 1 & \hdots & 1\end{pmatrix}$, where $z$ is a $Z$-logical operator of the input code and the all-ones vector has length $q_X$. The $X$ logicals of the input code still commute with the new $Z$ stabilizers on the bits $v_{i, 1}$. This gives $\tilde{d}_X = d_X$ and $\tilde{d}_Z = q_X d_Z$.

\begin{figure}[t!]
    \centering
    \subfloat[][]{
        \centering
        \includegraphics[height=3cm]{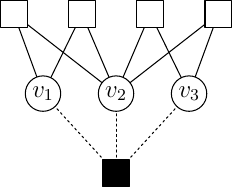}
        \label{fig:Tannercopyinga}
    }
    \hspace{4cm}
    \subfloat[][]{
        \centering
        \includegraphics[height=3cm]{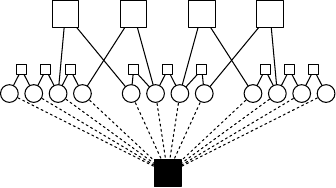}
        \label{fig:Tannercopyingb}
    }
    \caption{Open squares represent X stabilizers, filled squares Z stabilizers, and circles qubits. (a) The maximum column weight is $q_X = 4$ at vertex $v_2$. (b) The result of applying the copying procedure to (a). The copied variables in the repetition codes would have labels \(v_{1, 1}, v_{1, 2}, v_{1, 3}, v_{1, 4}, v_{2, 1},\dots\), and so on from left-to-right.}
    \label{fig:copying}
\end{figure}
Graphically, copying replaces each variable node in the Tanner graph with a repetition code of $X$ stabilizers protecting against phase errors. For example, the Tanner graph of \cref{fig:Tannercopyinga} is transformed to that of \cref{fig:Tannercopyingb}. The manner in which the edges from the check nodes are attached to the repetition codes correspond to the choice in placing $v_i$ in its copies. In particular, the ordering of the original stabilizers induces a potential permutation of the edges. We will discuss implications of this in \cref{sec:itdec}.

\subsubsection*{Gauging}
The goal of gauging is to reduce $w_X$ to less than or equal to three without increasing $q_X$. Consider an $X$ stabilizer of weight $w > 3$ with support on qubits labeled by $\{v_1, \hdots, v_w\}$. For each such stabilizer, gauging introduces $w - 3$ new qubits, $\{v^\prime_1, \hdots, v^\prime_{w - 3}\}$. These are in addition to any new qubits introduced by copying. The input $X$ stabilizer is replaced by new $X$ stabilizers with supports
\begin{equation*}
    \{v_1, v_2, v^\prime_1\}, \{v_3, v^\prime_1, v^\prime_2\}, \hdots, \{v_{w - 2}, v^\prime_{w - 4}, v^\prime_{w - 3}\}, \{v_{w - 1}, v^\prime_{w - 3}, v_w\}.
\end{equation*}
To see this in matrix form, assume without loss of generality that the support of the stabilizer is permuted to the first $w$ qubits. Then,
\begin{equation}\label{gaugingeq}
    \begin{pNiceArray}{ccccccc}[first-row]
        v_1 & \cdots & v_w & & & &\\
        1 & \cdots & 1 & 0 & \cdots & 0
    \end{pNiceArray}
    \mapsto
    \begin{pNiceArray}{ccccccccccccc}[first-row]
        v_1 & v_2 & v_3 & \cdots & v_{w - 2} & v_{w - 1} & v_w & & v^\prime_1 & v^\prime_2 & \cdots & v^\prime_{w - 2} & v^\prime_{w - 3}\\
        1 & 1 &   &        &   &   &   & \cdots & 1 &   &        &   & \\
          &   & 1 &        &   &   &   & \cdots & 1 & 1 &        &   & \\
          &   &   & \ddots &   &   &   & \cdots &   &   & \ddots &   & \\
          &   &   &        & 1 &   &   & \cdots &   &   &        & 1 & 1\\
          &   &   &        &   & 1 & 1 & \cdots &   &   &        &   & 1
    \end{pNiceArray},
\end{equation}
where the column of dots represent qubits not in the support of the current stabilizer. Note that the right-hand side is $H^{\mathrm{T}}_{w - 2}$, where the transpose
\begin{equation}\label{repcode}
    H_\ell = \begin{pmatrix}
        1 & 1 &  & & & \\
          & 1 & 1 & & & \\
          &    & \ddots & & & \\
          &    &        & 1 & 1 & \\
          &    &        &    & 1 & 1
    \end{pmatrix}
\end{equation}
is a parity-check matrix for the $[\ell, 1, \ell]$ classical repetition code. The left side of the matrix has support on the original qubits. The new (primed) qubits are not used by any other $X$ stabilizer, leaving the right side as the direct sum $\oplus_i H^{\mathrm{T}}_{w_i - 2}$, where $w_i$ is the weight of the $i$th reduced $X$ stabilizer.

At this point, the $Z$ stabilizers only have support on the original qubits and may no longer commute with the rows of \cref{gaugingeq}. Consider the $w - 2$ new rows of a reduced $X$ stabilizers, where the new stabilizers are arranged in the order of \cref{gaugingeq}. If the $j$th $Z$ stabilizer anti-commutes with the product of new $X$ stabilizers 1 to $m$ for $i \in \{1, \hdots, w - 3\}$, then set the $m$th new column of the $j$th row of $\tilde{H}_Z$ to one, where the new columns are those that were introduced when applying gauging to the original $X$ stabilizer. We will mention an alternative method to build commuting $Z$ stabilizers in \cref{sec:modifiedcopying}.

\begin{example}\label{ex:gauging}
Gauging the stabilizers
\begin{align*}
    H_X &= \begin{pNiceArray}{ccccccccccccccc}[columns-width=auto,margin]
          & 1 & 1 &   &   & 1 & 1 &   &   & 1 & 1 &   &   & 1 & 1 \\
        1 &   & 1 &   & 1 &   & 1 &   & 1 &   & 1 &   & 1 &   & 1
    \end{pNiceArray}\\
    H_Z &= \begin{pNiceArray}{ccccccccccccccc}[columns-width=auto,margin]
          &   &   &   &   &   &   & 1 & 1 & 1 & 1 & 1 & 1 & 1 & 1\\
          &   &   & 1 & 1 & 1 & 1 &   &   &   &   & 1 & 1 & 1 & 1\\
          & 1 & 1 &   &   & 1 & 1 &   &   & 1 & 1 &   &   & 1 & 1\\
        1 &   & 1 &   & 1 &   & 1 &   & 1 &   & 1 &   & 1 &   & 1
    \end{pNiceArray}
\end{align*}
gives 
\begin{align*}
    \tilde{H}_X &= \begin{pNiceArray}{ccccccccccccccc|ccccc|ccccc}[columns-width=auto,margin]
      & 1 & 1 &   &   &   &   &   &   &   &   &   &   &   &   & 1 &   &   &   &   &   &   &   &   &   \\
      &   &   &   &   & 1 &   &   &   &   &   &   &   &   &   & 1 & 1 &   &   &   &   &   &   &   &   \\
      &   &   &   &   &   & 1 &   &   &   &   &   &   &   &   &   & 1 & 1 &   &   &   &   &   &   &   \\
      &   &   &   &   &   &   &   &   & 1 &   &   &   &   &   &   &   & 1 & 1 &   &   &   &   &   &   \\
      &   &   &   &   &   &   &   &   &   & 1 &   &   &   &   &   &   &   & 1 & 1 &   &   &   &   &   \\
      &   &   &   &   &   &   &   &   &   &   &   &   & 1 & 1 &   &   &   &   & 1 &   &   &   &   &   \\
  \hline
    1 &   & 1 &   &   &   &   &   &   &   &   &   &   &   &   &   &   &   &   &   & 1 &   &   &   &   \\
      &   &   &   & 1 &   &   &   &   &   &   &   &   &   &   &   &   &   &   &   & 1 & 1 &   &   &   \\
      &   &   &   &   &   & 1 &   &   &   &   &   &   &   &   &   &   &   &   &   &   & 1 & 1 &   &   \\
      &   &   &   &   &   &   &   & 1 &   &   &   &   &   &   &   &   &   &   &   &   &   & 1 & 1 &   \\
      &   &   &   &   &   &   &   &   &   & 1 &   &   &   &   &   &   &   &   &   &   &   &   & 1 & 1 \\
      &   &   &   &   &   &   &   &   &   &   &   & 1 &   & 1 &   &   &   &   &   &   &   &   &   & 1 
\end{pNiceArray}\\
    \tilde{H}_Z &= \begin{pNiceArray}{ccccccccccccccc|ccccc|ccccc}[columns-width=auto,margin]
          &   &   &   &   &   &   & 1 & 1 & 1 & 1 & 1 & 1 & 1 & 1 &   &   &   & 1 &   &   &   &   & 1 &   \\
          &   &   & 1 & 1 & 1 & 1 &   &   &   &   & 1 & 1 & 1 & 1 &   & 1 &   &   &   &   & 1 &   &   &   \\
          & 1 & 1 &   &   & 1 & 1 &   &   & 1 & 1 &   &   & 1 & 1 &   & 1 &   & 1 &   & 1 & 1 &   &   & 1 \\
        1 &   & 1 &   & 1 &   & 1 &   & 1 &   & 1 &   & 1 &   & 1 & 1 & 1 &   &   & 1 &   & 1 &   & 1 &   
    \end{pNiceArray}.
\end{align*}
The parameters have transformed from $(w_X = 8, q_X = 2, w_Z = 8, q_Z = 4)$ to  $(\tilde{w}_X = 3, \tilde{q}_X = 2, \tilde{w}_Z = 13, \tilde{q}_Z = 4)$.
\end{example}

It follows from above that $\tilde{n}$ and $\tilde{n}_X$ increase by $w_i - 3$ for each $w_i > 3$ while $\tilde{n}_Z = n_Z$, and hence, $\tilde{k} = k$. We also have $\tilde{w}_X =\min \{w_X, 3\}$, $\tilde{q}_X = q_X$, $\tilde{w}_Z \geq w_Z$, and $\tilde{q}_Z \geq q_Z$. The old $X$ logical operators (appended with zeros for the new primed qubits) still commute with the new $Z$ stabilizes, however, $\tilde{d}_X$ could decrease. Take an element of $\ker \tilde{H}_Z$ with non-zero support on the new qubits that is not an $X$ stabilizer. The matrices $H^{\mathrm{T}}_{w_i - 2}$ row reduce to the identity (and a zero row) and may be used to clean the logical operator off the support of the new qubits \cite{bravyi2009no}. This could lead to a lower weight logical representative than the logical operators of the input code. The $Z$ distance is at least $d_Z$, but we delay the proof until \cref{thm:classicaldistthm}.

Graphically, gauging transforms the $X$ stabilizers of the Tanner graph from \cref{fig:gauginga} to \cref{fig:gaugingb}. The matrices $H^{\mathrm{T}}_{w - 2}$ induce a repetition code on the check nodes instead of the variable nodes.\\
\begin{figure}[t!]
    \centering
    \subfloat[][]{
        \centering
        \includegraphics[width=.3\textwidth]{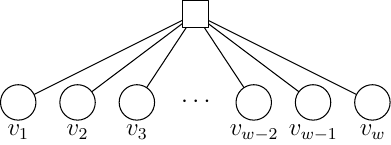}
        \label{fig:gauginga}
    }
    \hspace{1cm}
    \subfloat[][]{
        \centering
        \includegraphics[width=.3\textwidth]{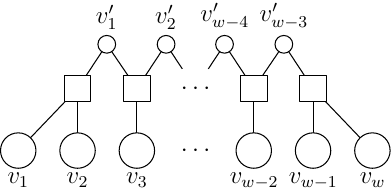}
        \label{fig:gaugingb}
    }
    \caption{(a) The input Tanner graph. (b) The affect of applying gauging to the $X$ stabilizer in (a).}
    \label{fig:gauging}
\end{figure}

\noindent {\bf Remark}: Hastings treats copying and gauging as a single operation \cite{hastings2021quantum}. Here the new qubits resulting from copying are called ``copied'' qubits and the new qubits resulting from gauging ``new'' qubits.

\subsubsection*{Thickening \& Choosing Heights}
The goal of thickening is to increase $d_X$, and the goal of choosing heights is to reduce $q_Z$. On its own, thickening is commonly referred to as (a special case of) distance balancing using the chain complexes
\begin{align*}
    &\mathcal{A} \, : \, \mathbb{F}^{n_Z}_2 \xrightarrow{H_Z^{\mathrm{T}}} \mathbb{F}^n_2 \xrightarrow{H_X} \mathbb{F}^{n_X}_2\\
    &\mathcal{B} \, : \, \mathbb{F}^{\ell - 1}_2 \xrightarrow{H^{\mathrm{T}}_\ell} \mathbb{F}^\ell_2,
\end{align*}
where $H_\ell$ is defined in \cref{repcode}. The details of the tensor product of these two chains are worked out in \cref{app:tenprodch}. There we show that $\mathcal{A} \otimes \mathcal{B}$ gives
\begin{equation*}
    C_3 \xrightarrow{\partial_3} C_2 \xrightarrow{\left(\tilde{H}_Z\right)^{\mathrm{T}}} C_1 \xrightarrow{\tilde{H}_X} C_0
\end{equation*}
where the code is described by stabilizer matrices (\cref{3by2chainmatsb} and the transpose of \cref{3by2chainmatsa})
\begin{equation}\label{thickening}
    \tilde{H}_X = \begin{pmatrix}
        H_X \otimes I_\ell & I_{n_X} \otimes H^{\mathrm{T}}_\ell
    \end{pmatrix}
    \quad , \quad
    \tilde{H}_Z = \begin{pmatrix}
        H_Z \otimes I_{\ell} & 0\\
        I_n \otimes H_\ell & H^{\mathrm{T}}_X \otimes I_{\ell - 1}
    \end{pmatrix},
\end{equation}
and has distances $\tilde{d}_X = \ell d_X$ and $\tilde{d}_Z = d_Z$. Thickening can be used to counteract the decrease in $X$ distance that gauging may have caused. In addition to $\tilde{H}_X$ and $\tilde{H}_Z$, there is a third matrix (\cref{thirdmap})\footnote{The map $\partial_3$ can be thought of as metachecks \cite{campbell2019theory,pryadko2019higher,quintavalle2021single} for the $Z$ stabilizers induced by the tensor product in thickening. Choosing heights eliminates the metachecks.}
\begin{equation}\label{partial3}
    \partial_3 = \begin{pmatrix}
        I_{n_Z} \otimes H^{\mathrm{T}}_\ell\\
        H^{\mathrm{T}}_Z \otimes I_{\ell - 1}
    \end{pmatrix},
\end{equation}
which satisfies $\tilde{H}^{\mathrm{T}}_Z \partial_3 = 0$. This implies that some of the $Z$ stabilizers are redundant with linear dependencies determined by $\partial_3$. Note that $I_{n_Z} \otimes H^{\mathrm{T}}_\ell = \bigoplus_{i = 1}^{n_z} H^{\mathrm{T}}_\ell$. Each row $h_j$ of $H_Z$ corresponds to a block in $\bigoplus_{i = 1}^{n_z} H^{\mathrm{T}}_\ell$, and the set of stabilizers $\begin{pmatrix} h_j \otimes I_\ell & 0 \end{pmatrix}$. The block diagonal structure of $\bigoplus_{i = 1}^{n_z} H^{\mathrm{T}}_\ell$ implies that the $\ell$ stabilizers in $h_j \otimes I_\ell$ are related via $\ell - 1$ constraints. Since these stabilizers are the cause of the potentially high column weights, removing the redundant stabilizers can help lower $q_Z$. Hastings calls making the choice of which row of $\begin{pmatrix} h_j \otimes I_\ell & 0 \end{pmatrix}$ to keep \emph{choosing heights}. We represent this by a length $n_Z$ vector $\mathrm{heights}$ whose $j$th element is an integer between $1$ and $\ell$ specifying which row of $\begin{pmatrix} h_j \otimes I_\ell & 0 \end{pmatrix}$ is kept. See \cref{thchex} for an explicit example of thickening, using $\partial_3$ to derive the row dependencies, and choosing heights.

When a different height is chosen for two different rows of $H_Z$, the resulting two rows in $\tilde{H}_Z$ have no qubits in common, thus a good choice of heights and a sufficient amount of thickening reduces $q_Z$. In the extreme case, choosing $\ell = n_Z$ and $\mathrm{heights} = (1, \hdots, n_Z)$ ensures a column weight of one within the block and therefore $q_Z = 3$ at the cost of extra qubits. A greedy algorithm can be used when $\ell < n_Z$ to choose heights that satisfy certain parameters such as a target $q_Z$. The properties of the resulting code depend highly on the choice of $\ell$ and $\mathrm{heights}$.

\begin{example}\label{thchex}
Thickening the stabilizers
\begin{equation}\label{thexample}
    H_X = \begin{pmatrix} 1 & 1 & 1 & 1 \end{pmatrix} \quad , \quad H_Z = \begin{pmatrix} 1 & 1 & 0 & 0\\ 1 & 0 & 1 & 0 \end{pmatrix},
\end{equation}
with $\ell = 3$ gives
\begin{align*}
    \tilde{H}_X &= \begin{pNiceArray}{cccccccccccc|cc}
        1 &   &   & 1 &   &   & 1 &   &   & 1 &   &   & 1 &   \\
          & 1 &   &   & 1 &   &   & 1 &   &   & 1 &   & 1 & 1 \\
          &   & 1 &   &   & 1 &   &   & 1 &   &   & 1 &   & 1
    \end{pNiceArray},\\
    \tilde{H}^\prime_Z &= \begin{pNiceArray}{cccccccccccc|cc}
        1 &   &   & 1 &   &   &   &   &   &   &   &   &   & \\
          & 1 &   &   & 1 &   &   &   &   &   &   &   &   & \\
          &   & 1 &   &   & 1 &   &   &   &   &   &   &   & \\
        1 &   &   &   &   &   & 1 &   &   &   &   &   &   & \\
          & 1 &   &   &   &   &   & 1 &   &   &   &   &   & \\
          &   & 1 &   &   &   &   &   & 1 &   &   &   &   & \\
        \hline
        1 & 1 &   &   &   &   &   &   &   &   &   &   & 1 &   \\
          & 1 & 1 &   &   &   &   &   &   &   &   &   &   & 1 \\
          &   &   & 1 & 1 &   &   &   &   &   &   &   & 1 &   \\
          &   &   &   & 1 & 1 &   &   &   &   &   &   &   & 1 \\
          &   &   &   &   &   & 1 & 1 &   &   &   &   & 1 &   \\
          &   &   &   &   &   &   & 1 & 1 &   &   &   &   & 1 \\
          &   &   &   &   &   &   &   &   & 1 & 1 &   & 1 &   \\
          &   &   &   &   &   &   &   &   &   & 1 & 1 &   & 1 
    \end{pNiceArray},\qquad
    \partial_3 = \begin{pNiceArray}{cccc}
        1 &   &   & \\
        1 & 1 &   & \\
          & 1 &   & \\
          &   & 1 & \\
          &   & 1 & 1 \\
          &   &   & 1 \\
        \hline
        1 &   & 1 & \\
          & 1 &   & 1 \\
        1 &   &   & \\
          & 1 &   & \\
          &   & 1 & \\
          &   &   & 1 \\
          &   &   & \\
          &   &   &
    \end{pNiceArray}.
\end{align*}
Let $\{h_1, \hdots, h_{14}\}$ denote the rows of $\tilde{H}^\prime_Z$, where $h_1$ to $h_6$ are above and $h_7$ to $h_{14}$ below the horizontal line, respectively. Multiplying $\partial_3 \left(\tilde{H}^\prime_Z\right)^{\mathrm{T}}$ gives four equations, with the first two:
\begin{align}
    h_1 + h_2 + h_7 + h_9 &= 0 \label{eq1}\\ 
    h_2 + h_3 + h_8 + h_{10} &= 0. \label{eq2}
\end{align}
If we remove rows $h_2$ and $h_3$ but keep rows $h_1$ and $h_7$ to $h_{14}$, we can recover $h_2$ from \cref{eq1} and $h_3$ from \cref{eq2}. Hence, $h_2, h_3 \in \mathrm{rowspace}(\tilde{H}^\prime_Z)$ even if removed as rows of $\tilde{H}^\prime_Z$. A similar argument using the next two equations generated from $\partial_3 \left(\tilde{H}^\prime_Z\right)^{\mathrm{T}}$ shows that we only need to keep one of the rows $h_4$, $h_5$, and $h_6$. Choosing $\mathrm{heights} = (1, 2)$ keeps rows $h_1$ and $h_5$ and removes rows $h_2$, $h_3$, $h_4$, and $h_6$. The final $Z$ stabilizers are
\begin{equation}
    \tilde{H}_Z = \begin{pNiceArray}{cccccccccccc|cc}
        1 &   &   & 1 &   &   &   &   &   &   &   &   &   & \\
          & 1 &   &   &   &   &   & 1 &   &   &   &   &   & \\
        \hline
        1 & 1 &   &   &   &   &   &   &   &   &   &   & 1 &   \\
          & 1 & 1 &   &   &   &   &   &   &   &   &   &   & 1 \\
          &   &   & 1 & 1 &   &   &   &   &   &   &   & 1 &   \\
          &   &   &   & 1 & 1 &   &   &   &   &   &   &   & 1 \\
          &   &   &   &   &   & 1 & 1 &   &   &   &   & 1 &   \\
          &   &   &   &   &   &   & 1 & 1 &   &   &   &   & 1 \\
          &   &   &   &   &   &   &   &   & 1 & 1 &   & 1 &   \\
          &   &   &   &   &   &   &   &   &   & 1 & 1 &   & 1 
    \end{pNiceArray}.
\end{equation}
The parameters have transformed from $(w_X = 4, q_X = 1, w_Z = 2, q_Z = 2)$ to  $(\tilde{w}_X = 6, \tilde{q}_X = 2, \tilde{w}_Z = 3, \tilde{q}_Z = 4)$. This example is shown in Tanner graphs in \cref{fig:thickening}.
\end{example}

It follows from above that $\tilde{n} = \ell n + (\ell - 1)n_X$, $\tilde{n}_X = \ell n_X$, $\tilde{w}_X = w_X + 2$, $\tilde{q}_X = \max \{q_X, 2\}$, and $\tilde{w}_Z = \max \{w_Z, q_X + 2\}$. Before choosing heights, $\tilde{n}_Z = \ell n_Z + (\ell - 1)n$ and $\tilde{q}_Z = \max \{w_X, q_Z + 2\}$. The logical operators are given by the K\"{u}nneth formula (\cref{Kunneth}),
\begin{equation*}
    H_1(\mathcal{C}) = (H_1(\mathcal{A}) \otimes H_0(\mathcal{B})) \oplus (H_0(\mathcal{A}) \otimes H_1(\mathcal{B})) = H_1(\mathcal{A}) \otimes H_0(\mathcal{B}),
\end{equation*}
from which it follows that $\tilde{k} = \mathrm{dim} \, H_1(\mathcal{C}) = \mathrm{dim} \, H_1(\mathcal{A}) \, \mathrm{dim} \, H_0(\mathcal{B}) = k \cdot 1$. Neither the dimension nor the logicals are affected by removing redundant stabilizers. The homology of the repetition code is given in \cref{app:tenprodch}.

Graphically, recall that the Tanner graph has a variable node for each column (qubit) and a check node for each row (stabilizer). For the $X$ stabilizers, there are therefore variable nodes for each column of $I_{n_X} \otimes H^{\mathrm{T}}_\ell$ and separate variable nodes for each column of $H_X \otimes I_\ell$, both of which are connected to the same check nodes. The Tanner graph of $H^{\mathrm{T}}_\ell$ is simply the Tanner graph of $H_\ell$ with the variable and check nodes (columns and rows) switched. The term $I_{n_X} \otimes H^{\mathrm{T}}_\ell$ makes $n_X$ identical and unconnected copies of the Tanner graph of $H^{\mathrm{T}}_\ell$. The term $H_X \otimes I_\ell$ is equivalent to $I_\ell \otimes H_X$ up to row and column permutations but connects the $\ell$ copies of $H_X$ in a different pattern. The $Z$ stabilizers are similar but now there is an extra set of check nodes for $H_Z \otimes I_\ell$ acting on the same variable nodes as $I_n \otimes H_\ell$. An example is given in \cref{fig:thickening}.
\begin{figure}[t!]
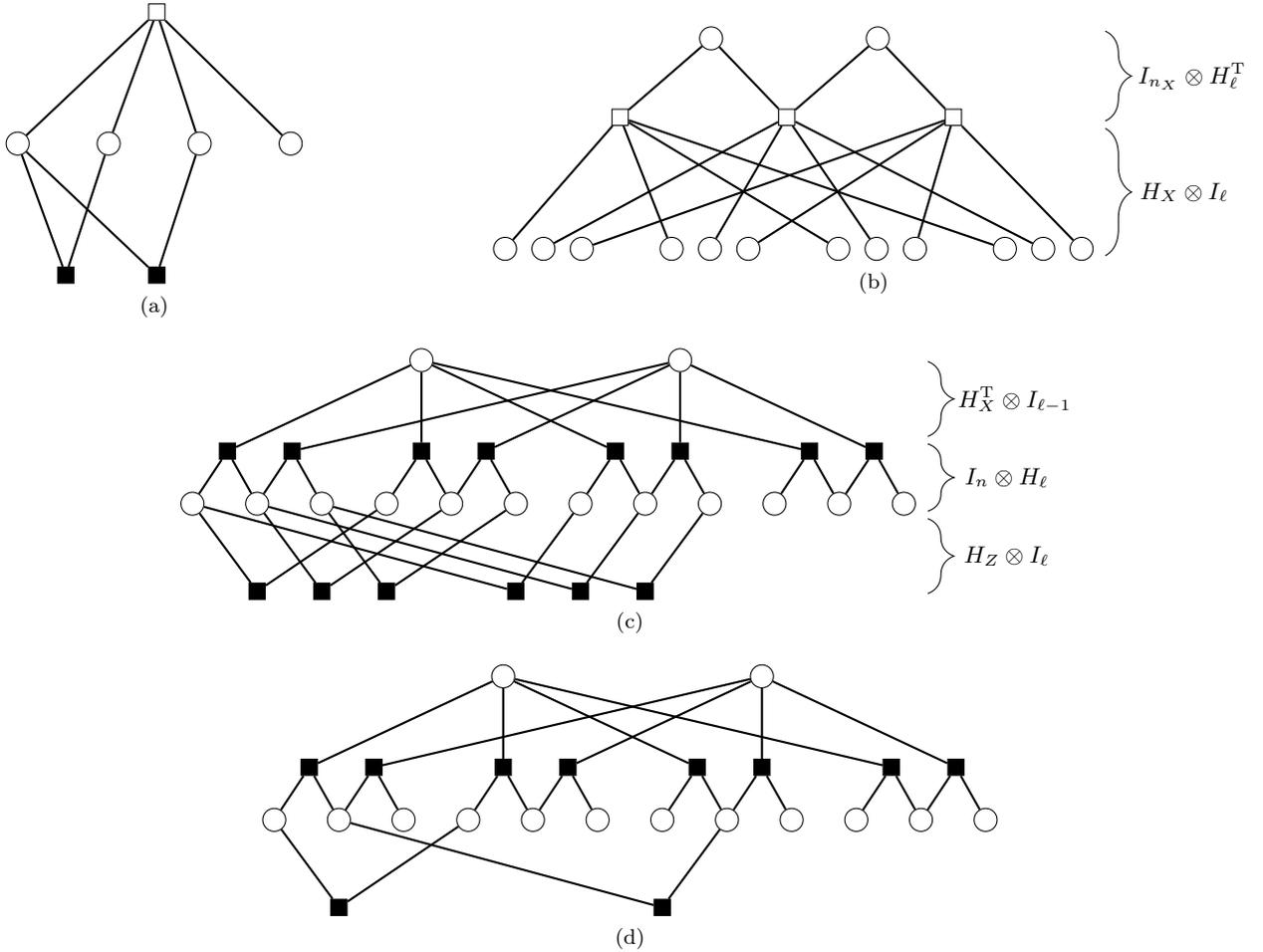

    \centering
    \subfloat[][]{
        \centering
        \thickeningHastingsa
        \label{fig:thickeningHastingsa}
    }
    \hspace{2cm}
    \subfloat[][]{
        \centering
        \thickeningHastingsbX
        \label{fig:thickeningHastingsbX}
    }\\
    \subfloat[][]{
        \centering
        \thickeningHastingsbZ
        \label{fig:thickeningHastingsbZ}
    }\\
    
    \subfloat[][]{
        \centering
        \chHastingsZ
        \label{fig:thickeningHastingscZ}
    }
    \caption{(a) The Tanner graph of \cref{thexample} with $X$ stabilizer on top and two $Z$ stabilizers on bottom. (b) The effect of applying thickening to the $X$ stabilizer of (a). The $Z$ stabilizers are omitted for clarity. (c) The effect of applying thickening to the $Z$ stabilizers of (a). The $X$ stabilizers are omitted for clarity. (d) Using the notation of the text, the effect of choosing $\mathrm{heights} = (1, 2)$ on (c). The $X$ stabilizers remain unchanged. The Tanner graph for the full code consists of the stabilizers shown in both (b) and (d).}
    \label{fig:thickening}
\end{figure}

\subsubsection*{Coning}
The necessary homological algebra to motivate the formulas here is reviewed in \cref{app:homological}. It is suggested for readers unfamiliar with tensor products of chain complexes and the mapping cone to read this first. Here, we will only describe the coning of what Hastings calls \emph{reasonable} codes, which will be defined later in the context it arises. Readers interested in the modifications required for \emph{unreasonable} codes are referred to the discussion in Ref.~\cite{hastings2021quantum}.

A set of stabilizer generators for the 15-qubit quantum Reed-Muller code $\mathrm{QRM}(4)$~\cite{knill1996threshold,anderson2014fault} is
\begin{align}\label{RMstabs}
    H_X = &\begin{pmatrix}
        1 &   & 1 &   & 1 &   & 1 &   & 1 &   & 1 &   & 1 &   & 1\\
          & 1 & 1 &   &   & 1 & 1 &   &   & 1 & 1 &   &   & 1 & 1\\
          &   &   & 1 & 1 & 1 & 1 &   &   &   &   & 1 & 1 & 1 & 1\\
          &   &   &   &   &   &   & 1 & 1 & 1 & 1 & 1 & 1 & 1 & 1
    \end{pmatrix},\\
    H_Z = &\begin{pmatrix}
        1 &   & 1 &   & 1 &   & 1 &   & 1 &   & 1 &   & 1 &   & 1\\
          & 1 & 1 &   &   & 1 & 1 &   &   & 1 & 1 &   &   & 1 & 1\\
          &   &   & 1 & 1 & 1 & 1 &   &   &   &   & 1 & 1 & 1 & 1\\
          &   &   &   &   &   &   & 1 & 1 & 1 & 1 & 1 & 1 & 1 & 1\\
          &   & 1 &   &   &   & 1 &   &   &   & 1 &   &   &   & 1\\
          &   &   &   & 1 &   & 1 &   &   &   &   &   & 1 &   & 1\\
          &   &   &   &   & 1 & 1 &   &   &   &   &   &   & 1 & 1\\
          &   &   &   &   &   &   &   &   & 1 & 1 &   &   & 1 & 1\\
          &   &   &   &   &   &   &   &   &   &   & 1 & 1 & 1 & 1\\
          &   &   &   &   &   &   &   & 1 &   & 1 &   & 1 &   & 1
    \end{pmatrix}, \nonumber
\end{align}
which has $(w_X = 8, q_X = 4, w_Z = 8, q_Z = 10)$. Applying copying, gauging, and then thickening and choosing heights with $\ell = 3$, $\mathrm{heights} = (2, 1, 2, 1, 2, 3, 1, 3, 3, 1)$ gives the sequence of parameters $(w_X, q_X, w_Z, q_Z)$:
\begin{equation*}
    (8, 4, 8, 10) \xrightarrow{copying} (8, 3, 32, 10) \xrightarrow{gauging} (3, 3, 43, 10) \xrightarrow{th. \, \& \, c.h.} (5, 3, 43, 5).
\end{equation*}
Na\"{i}vely trying to reduce $w_Z$ using a second round of gauging applied to the $Z$ stabilizers gives $(w_X = 85, q_X = 22, w_Z = 3, q_Z = 5)$, which is counterproductive. This occurs in general and is not unique to this example. Coning using the mapping cone of \cref{app:mappingcone} serves to reduce $w_Z$ while preserving $w_X$, $q_X$ and $q_Z$.

Coning is more involved than copying, gauging, or thickening and choosing heights. For this reason, we begin with a brief overview, which is aligned with the discussion in \cref{app:mappingcone}. First, pick a $Z$ stabilizer to be reduced, $Z_i$, and remove it from $H_Z$ to create $H^{(r)}_Z$. Second, view this code as a chain complex in the standard way,
\begin{equation*}
    \begin{tikzcd}
	\mathcal{B}: {C_2} & {C_1} & {C_0}
	\arrow["{\left(H^{(r)}_Z\right)^{\mathrm{T}}}", from=1-1, to=1-2]
	\arrow["{H_X}", from=1-2, to=1-3]
    \end{tikzcd},
\end{equation*}
with $C_2 = \mathbb{F}^{n_Z}_2$, $C_1 = \mathbb{F}^n_2$, and $C_0 = \mathbb{F}^{n_X}_2$. Next, define a new chain complex $\mathcal{A}_i: \mathcal{Q}_i \to \mathcal{X}_i \to \mathcal{R}_i$, where $C_1 = \mathcal{Q}_i = \mathbb{F}^{|\mathrm{supp}(Z_i)|}_2$, $C_0 = \mathcal{X}_i$ is derived from the overlap of $Z_i$ with the $X$ stabilizers, and $C_{-1} = \mathcal{R}_i$ is constructed, if necessary, based on $\mathcal{Q}_i$ and $\mathcal{X}_i$ to ensure that no new logical operators are added with support only on the new qubits. All three spaces are defined in detail below. These two chain complexes will be connected with appropriate chain maps $f^{(i)}$,
\begin{equation*}
    \begin{tikzcd}
	& {\mathcal{Q}_i} & {\mathcal{X}_i} & {\mathcal{R}_i} \\
	{C_2} & {C_1} & {C_0}
	\arrow["{\partial^{(i)}_1}", from=1-2, to=1-3]
	\arrow["{\partial^{(i)}_0}", from=1-3, to=1-4]
	\arrow["{\left(H^{(r)}_Z\right)^{\mathrm{T}}}"', from=2-1, to=2-2]
	\arrow["{H_X}"', from=2-2, to=2-3]
	\arrow["{f^{(i)}_1}"', from=1-2, to=2-2]
	\arrow["{f^{(i)}_0}"', from=1-3, to=2-3]
    \end{tikzcd},
\end{equation*}
which induces the mapping cone
\begin{equation}\label{mappingconesingle}
    \begin{tikzcd}
	& {\mathcal{Q}_i} & {\mathcal{X}_i} & {\mathcal{R}_i} \\
	& \oplus & \oplus & \oplus \\
	& {C_2} & {C_1} & {C_0} \\
	{\mathrm{cone}\left(f^{(i)}\right):} & {C_2 \oplus \mathcal{Q}_i} & {C_1 \oplus \mathcal{X}_i} & {C_0 \oplus \mathcal{R}_i}
	\arrow["{\partial^{(i)}_1}", from=1-2, to=1-3]
	\arrow["{\partial^{(i)}_0}", from=1-3, to=1-4]
	\arrow["{\left(H^{(r)}_Z\right)^{\mathrm{T}}}"', from=3-2, to=3-3]
	\arrow["{H_X}"', from=3-3, to=3-4]
	\arrow["{f^{(i)}_0}"', from=1-3, to=3-4]
	\arrow["{f^{(i)}_1}"', from=1-2, to=3-3]
	\arrow["{\tilde{H}^{\mathrm{T}}_Z}", from=4-2, to=4-3]
	\arrow["{\tilde{H}_X}", from=4-3, to=4-4]
    \end{tikzcd}
\end{equation}
with
\begin{equation}\label{coningstabssingle}
    \tilde{H}_X = \begin{pmatrix}
        \partial^{(i)}_0 & \\
        f^{(i)}_0 & H_X\\
        \end{pmatrix}
    \quad , \quad
    \tilde{H}_Z = \begin{pmatrix}
        \left( \partial^{(i)}_1\right)^{\mathrm{T}} & \left( f^{(i)}_1 \right)^{\mathrm{T}} \\
        & H^{(r)}_Z
        \end{pmatrix}.
\end{equation}
This is for a fixed $i$. The case for $m$ stabilizers to be weight reduced gives\footnote{The $n_Z-m$ $Z$ stabilizers that remain in $H_Z^{(r)}$ are directly visible in the resulting code and are referred to as ``direct stabilizers'', while the other $m$ $Z$ stabilizers are induced by the resulting code and are referred to as ``induced stabilizers'' \cite{hastings2021quantum}.}
\begin{equation}\label{mappingconemultiple}
    \begin{tikzcd}
    	{\mathcal{Q}_1} && {\mathcal{X}_1} && {\mathcal{R}_1} \\
    	\oplus && \oplus && \oplus \\
    	\vdots && \vdots && \vdots \\
    	\oplus && \oplus && \oplus \\
    	{\mathcal{Q}_m} && {\mathcal{X}_m} && {\mathcal{R}_m} \\
    	\\
    	{C_2} && {C_1} && {C_0}
    	\arrow["{\partial^{(1)}_1}"{description}, from=1-1, to=1-3]
    	\arrow["{\partial^{(1)}_0}"{description}, from=1-3, to=1-5]
    	\arrow["{\left(H^{(r)}_Z\right)^{\mathrm{T}}}"{description}, from=7-1, to=7-3]
    	\arrow["{H_X}"{description}, from=7-3, to=7-5]
    	\arrow["{f^{(1)}_0}"{description}, from=1-3, to=7-5]
    	\arrow["{f^{(1)}_1}"{description}, from=1-1, to=7-3]
    	\arrow["{f^{(m)}_1}"{description}, from=5-1, to=7-3]
    	\arrow["{f^{(m)}_0}"{description}, from=5-3, to=7-5]
    	\arrow["{\partial^{(m)}_1}"{description}, from=5-1, to=5-3]
    	\arrow["{\partial^{(m)}_0}"{description}, from=5-3, to=5-5]
    \end{tikzcd}.
\end{equation}
and
\begin{equation}\label{coningstabs}
    \tilde{H}_X = \begin{pmatrix}
            \partial^{(1)}_0 & & & \\
            & \ddots & & \\
            & & \partial^{(m)}_0 & \\
        f_0^{(1)} & \cdots & f_0^{(m)} & H_X
    \end{pmatrix}
    \quad , \quad
    \tilde{H}_Z = \begin{pmatrix}
        \left(\partial^{(1)}_1\right)^{\mathrm{T}} & & & \left(f_1^{(1)}\right)^{\mathrm{T}} \\
         & \ddots & & \vdots \\
         & & \left(\partial^{(m)}_1\right)^{\mathrm{T}} & \left(f_1^{(m)}\right)^{\mathrm{T}}\\
         & & & H_Z^{(r)}
    \end{pmatrix}.
\end{equation}
These commute by definition of the chain complex. If all $Z$ stabilizers need to be reduced, then $H^{(r)}_Z$ and $C_2$ will be empty. There are several degrees of freedom in constructing these objects and the final parameters of the code are highly dependent on the choices made.

Now in full detail, let $Z_i = (z_1, \hdots, z_n)$ be a $Z$-stabilizer generator whose weight needs to be reduced. Define $\mathcal{Q}_i$ be the vector space $\mathbb{F}^{|\mathrm{supp}(Z_i)|}_2$. There is a natural embedding of the standard basis of $\mathcal{Q}_i$ to all weight-one patterns of $\mathbb{F}^n_2$ contained in the support of $Z_i$, $f^{(i)}_1 : \mathcal{Q}_i \to C_1$, which extends to all of $\mathcal{Q}_i$ by linearity: if the $k$th element of $\mathrm{supp}(Z_i)$ is $z_p$, then $f^{(i)}_1$ maps the $k$th elementary basis element of $\mathcal{Q}_i$ to the $p$th elementary basis element of $\mathbb{F}^n_2$. A matrix representation of $f^{(i)}_1$ has maximum row and column weights equal to one, which helps control $w_Z$ and $q_Z$ in \cref{coningstabs}.

The support of any $X$ stabilizer overlaps with the support of $Z_i$ an even number of times. Construct a set of tuples $\{(S, j, k)\}$, where $S$ is an $X$-stabilizer generator with overlapping support on indices $j, k \in \mathrm{supp}(S) \cap \mathrm{supp}(Z_i)$. It is not necessary to include all possible pairs $(j, k)$ but the entire overlap must be represented for each $S$. For example, if the overlap occurs on indices $\{1, 2, 3, 4\}$, then $\{(S, 1, 2), (S, 3, 4)\}$ or $\{(S, 1, 2), (S, 2, 3), (S, 3, 4)\}$ or $\{(S, 1, 2), (S, 1, 3), (S, 1, 4), (S, 2, 3), (S, 2, 4), (S, 3, 4)\}$ are three valid options. Define $\mathcal{X}_i$ to be the vector space $\mathbb{F}^{|\{(S, j, k)\}|}_2$, where the $c$th standard basis element may be associated to the $c$th element of $\{(S, j, k)\}$. There is a map from $\mathcal{X}_i$ to the space of $X$ syndromes, $f^{(i)}_0 : \mathcal{X}_i \to C_0$, which, in matrix form, has a $1$ in row $r$ and column $c$ if the $X$ stabilizer associated with the $c$th basis element is the $r$th $X$ stabilizer represented in $H_X$.

There is also a natural connection between $\mathcal{Q}_i$ and $\mathcal{X}_i$: the elements of $\mathcal{Q}_i$ are associated with qubits and $\mathcal{X}_i$ with pairs of qubits. Make a graph with a vertex for every basis element of $\mathcal{Q}_i$ with an edge between the vertices associated with the qubits in $(S,j, k) \in \mathcal{X}_i$.\footnote{This is potentially a multigraph because two different $X$ stabilizers could have had the same overlap with $Z_i$.} The map $\partial^{(i)}_1: \mathcal{Q}_i \to \mathcal{X}_i$ is called the coboundary map of the graph and takes the vertices to the basis elements of $\mathcal{X}_i$ that include those vertices. For example, the coboundary map of the square with vertices $V = \{v_1, v_2, v_3, v_4\}$ and edges $E = \{v_1 v_2, v_2 v_3, v_3 v_4, v_4 v_1\}$ is the $|E| \times |V|$ edge-vertex incidence matrix
\begin{equation*}
    \begin{pNiceArray}{cccc}[first-row, first-col]
                & v_1 & v_2 & v_3 & v_4\\
        v_1 v_2 & 1 & 1 & 0 & 0\\
        v_2 v_3 & 0 & 1 & 1 & 0\\
        v_3 v_4 & 0 & 0 & 1 & 1\\
        v_4 v_1 & 1 & 0 & 0 & 1
    \end{pNiceArray}.
\end{equation*}

Looking back at \cref{coningstabs}, the number of rows of $\partial^{(i)}_1$, and hence the spaces themselves, determine the number of new qubits in the weight reduction ($|\mathcal{X}_i|$). Assuming coning is applied to the output of all the previous steps, the input can be made to satisfy $w_X \leq 5$, $q_X \leq 3$, and $q_Z \leq 3$ and are of the form \cref{thickening}. The $Z$ stabilizers corresponding to the bottom row of \cref{thickening} have weight no more than five (since $H_X$ there has $q_X \leq 3$) and therefore do not need to be reduced. The $I_{n_X} \otimes H^{\mathrm{T}}$ term in the $X$ stabilizers contributes two to the weight but overlaps the $0$ block in the $Z$ stabilizers that need reduction. Thus, the support of an $X$ stabilizer can therefore only intersect the support of $Z_i$ zero or two times. This keeps $\mathcal{X}_i$ smaller compared to applying coning to a completely random input.

Finally, we construct the space $\mathcal{R}_i$. The logical operators are determined by the chain complex of homologies (long exact sequence \cref{LES})
\begin{equation}\label{homeq}
    0 \to H_1(\mathcal{A}_i) \to H_1(\mathcal{B}) \to H_1\left(\mathrm{cone}\left(f^{(i)}\right)\right) \to H_0(\mathcal{A}_i).
\end{equation}
The term $H_1(\mathcal{B})$ represents the logical operators of the original code plus the logical operator created by deleting the $Z$ stabilizer that is being reduced. The logical operators of the output of coning are elements of $H_1\left(\mathrm{cone}\left(f^{(i)}\right)\right)$. In order to find a relationship between the two, recall that weight reduction is considered an operation on a code as compared to a method of constructing a new code. Hence, we require $k$ remains the same throughout the entire process. In order for this to happen, we need $H_0(\mathcal{A}_i)$ to be zero. The space $\mathcal{R}_i$ is designed to make this happen. Recall (\cref{app:homological}) that $H_0(\mathcal{A}_i) = \ker \partial^{(i)}_0 / \mathrm{im}\, \partial^{(i)}_1$. It follows that we should choose $\mathcal{R}_i$ and $\partial^{(i)}_0$ such that $\ker \partial^{(i)}_0 = \mathrm{im}\, \partial^{(i)}_1$. Define $\mathcal{R}_i$ to be the $\mathbb{F}_2$-vector space with a basis element for every element of the cycle basis of the graph defined by $\mathcal{Q}_i$ and $\mathcal{X}_i$, and $\partial^{(i)}_0: \mathcal{X}_i \to \mathcal{R}_i$ by the coboundary map sending the edges to the cycles they are contained in. If the graph has no cycles, it is not necessary to construct $\mathcal{R}_i$. In practice, we have observed that we have always needed it.

The following example shows that not all cycle bases are equivalent for our purposes. Freedman and Hastings provide an algorithm called the Decongestion Lemma \cite{freedman2021building} to find a cycle basis such that each edge appears in at most $O(\log |\mathcal{Q}_i|^2)$ cycles and whose cycle length is $O(|\mathcal{Q}_i| \log |\mathcal{Q}_i|)$. This may not be the minimum-weight cycle basis, but these conditions are designed to control the column and row weights of $\partial^{(i)}_0$. Alternative algorithms, such as Horton's algorithm \cite{horton1987polynomial} for finding a minimum-weight cycle basis, could be explored in future work.

\begin{figure}[t!]
    \centering
    \subfloat[][]{
        \centering
        \includegraphics[width=.2\textwidth]{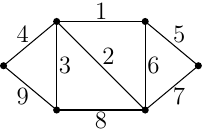}
        \label{fig:cycle-basis-a}
    }
    \quad
    \subfloat[][]{
        \centering
        \includegraphics[width=.2\textwidth]{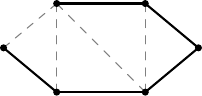}
        \label{fig:cycle-basis-b}
    }
    \quad
    \subfloat[][]{
        \centering
        \includegraphics[width=.2\textwidth]{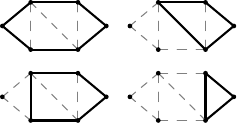}
        \label{fig:cycle-basis-c}
    }
    \quad
    \subfloat[][]{
        \centering
        \includegraphics[width=.2\textwidth]{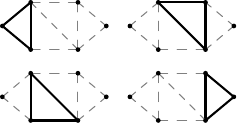}
        \label{fig:cycle-basis-d}
    }
    \caption{(a) An input graph. (b) A spanning tree for the graph in (a). (c) The fundamental cycle basis resulting from (b). (d) A minimum-weight cycle basis for (a).}
    \label{fig:cycle-basis}
\end{figure}

\begin{example}
    To find a cycle basis of the graph in \cref{fig:cycle-basis-a}, start with a spanning tree. One possible choice is \cref{fig:cycle-basis-b}. The fundamental cycle basis associated with this tree is given by adding edges from \cref{fig:cycle-basis-a} not in \cref{fig:cycle-basis-b}. Adding edges $4$, $2$, $3$, and $6$ produce the fundamental cycle basis in \cref{fig:cycle-basis-c}. If this was interpreted as $\mathcal{Q}_i$ and $\mathcal{X}_i$, we would have
    \begin{equation*}
        \partial^{(i)}_0 = \begin{pmatrix}
            1 &   &   & 1 & 1 &   & 1 & 1 & 1\\
            1 & 1 &   &   & 1 &   & 1 &   &  \\
            1 &   & 1 &   & 1 &   & 1 & 1 &  \\
              &   &   &   & 1 & 1 & 1 &   &  \\
        \end{pmatrix},
    \end{equation*}
    where the rows correspond to the cycles read left-to-right then top-to-bottom in \cref{fig:cycle-basis-c} and the columns correspond to the edges in \cref{fig:copying-cycles-a} in numerical order. The two columns of ones correspond to the edges $5$ and $7$ which appear in every cycle of \cref{fig:cycle-basis-c}. This is a potentially undesirable increase in $q_X$ (\cref{coningstabs}).
    
    Alternatively, a so-called minimum-weight cycle basis whose total sum of the number of edges in each basis element is minimal is given in \cref{fig:cycle-basis-d}. Now,
    \begin{equation*}
        \partial^{(i)}_0 = \begin{pmatrix}
              &   & 1 & 1 &   &   &   &   & 1\\
            1 & 1 &   &   &   & 1 &   &   &  \\
              & 1 & 1 &   &   &   &   & 1 &  \\
              &   &   &   & 1 & 1 & 1 &   &
        \end{pmatrix},
    \end{equation*}
    where the rows correspond to the cycles read left-to-right then top-to-bottom in \cref{fig:cycle-basis-d} and the columns correspond to the edges in \cref{fig:copying-cycles-a} in numerical order. The minimum-weight basis avoids increasing $q_X$ but is not guaranteed to always do so. The Decongestion Lemma reduces both the row and column weights, but only with high probability.
\end{example}

Having chosen a cycle basis, some elements may contain a large number of edges. Large cycles lead to high-weight rows of $\partial^{(i)}_0$ and thus high-weight $X$ stabilizers. If the input code was LDPC, this would produce a non-LDPC output and undermine the reduction of $w_X$ in gauging. To fix this, we introduce auxiliary edges to break up the cycles into smaller cycles in a process called cellulation. The map and the resulting stabilizers are highly dependent on the choice of cellulation. In this sense, some cellulations are better than others. Here we follow Reference~\cite{wills2023tradeoff}: any time a cycle has length greater than four, simply add edges to bring it down to four. We do not add any new vertices (qubits) to do this. These new edges must be added to $\mathcal{X}_i$ but without an associated $X$ stabilizer, i.e., $(\_, j, k)$. This affects the maps $\partial^{(i)}$ but not $f^{(i)}_0$. The additions to $\partial^{(i)}$ appear in the stabilizers \cref{coningstabs} and could affect the LDPC properties of the output. We leave the question of ``optimal'' cellulations to future work. An explicit example is done below in \cref{ex:coning}, and another graphical demonstration is provided in \cref{fig:coning}.

If $q_X$ is higher than desired, we can perform an extra round of thickening and choosing heights with the roles of $X$ and $Z$ switched:
\begin{align*}
    &\mathcal{A} \, : \, \mathbb{F}^{n_Z}_2 \xleftarrow{H_Z} \mathbb{F}^n_2 \xleftarrow{H^{\mathrm{T}}_X} \mathbb{F}^{n_X}_2\\
    &\mathcal{B} \, : \, \mathbb{F}^{\ell - 1}_2 \xleftarrow{H_\ell} \mathbb{F}^\ell_2.
\end{align*}
The cellulation and the optional second thickening and choosing heights is referred to as the reduced cone in \cite{hastings2021quantum, wills2023tradeoff}.

\begin{example}\label{ex:coning}
    Consider coning the following inputs
    \begin{align*}
        H_X &= \begin{pNiceArray}{cccccccccc}
            1 & 1 &   &   &   &   &   &   &   &\\
              & 1 & 1 &   &   &   &   &   &   &\\
              &   & 1 & 1 &   &   &   &   &   &\\
              &   &   & 1 & 1 &   &   &   &   &\\
              &   &   &   & 1 & 1 &   &   &   &\\
              &   &   &   &   & 1 & 1 &   &   &\\
            1 &   &   &   &   &   & 1 &   &   &\\
              &   &   & 1 &   &   &   & 1 &   &\\
              &   &   &   &   &   &   & 1 & 1 &\\
              &   &   &   &   &   &   &   & 1 & 1\\
              &   &   &   &   &   &   & 1 &   & 1
        \end{pNiceArray}\\
        H_Z &= \, \, \begin{pNiceArray}{cccccccccc}
            1 & 1 & 1 & 1 & 1 & 1 & 1 & 1 & 1 & 1
        \end{pNiceArray}.
    \end{align*}
    There is only one $Z$ stabilizer to be reduced with support on all the qubits, so $\mathcal{Q}_1 = \mathbb{F}^{10}_2$. Since $Z_1$ has full support, $f^{(1)}_1$ is simply the identity map. The overlaps of $H_X$ and $H_Z$ give $\mathcal{X}_1 = \mathbb{F}^{11}_2$ with standard basis elements in correspondence with the edges
    \begin{equation*}
         \{(S_1, 1, 2), (S_2, 2, 3), (S_3, 3, 4), (S_4, 4, 5), (S_5, 5, 6), (S_6, 6, 7), (S_7, 7, 1), (S_8, 4, 8), (S_9, 8, 9), (S_{10}, 9, 10), (S_{11}, 10, 8)\}.
    \end{equation*}
    Since the $j$th element of $\mathcal{X}_1$ is associated with the $j$th $X$ stabilizer, $f^{(1)}_0$ is also the identity.

    We can visualize the graphical relationship between $\mathcal{Q}_i$ and $\mathcal{X}_i$ by treating $H_X$ restricted to the qubits on the support of the $Z$ stabilizer being reduce as the edge-vertex incidence matrix
    \begin{equation*}
        \includegraphics{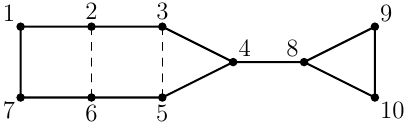} .
    \end{equation*}
    The cycles are $\mathcal{R}_i = \{(1, 2, 3, 4, 5, 6, 7), (8, 9, 10)\}$ giving
    \begin{equation*}
        \partial^{(1)}_0 = \begin{pNiceArray}{cccccccccc}
            1 & 1 & 1 & 1 & 1 & 1 & 1 &   &   & \\
              &   &   &   &   &   &   & 1 & 1 & 1
        \end{pNiceArray}.
    \end{equation*}
    This will produce a high-weight $X$ stabilizer in \cref{coningstabs}, so we cellulate by adding the dashed edges $2-6$ and $3-5$ and $(\_, 2, 6), (\_, 3, 5)$ to $\mathcal{X}_1$. Now the cycles are
    \begin{equation*}
        \mathcal{R}_i = \{(1, 2, 6, 7), (2, 3, 5, 6), (3, 4, 5), (8, 9, 10)\}
    \end{equation*}
    and
    \begin{equation*}
        \partial^{(1)}_0 = \begin{pNiceArray}{ccccccccccc|cc}
            1 &   &   &   &   & 1 & 1 & & &   &   & 1 &\\
              & 1 &   &   & 1 &   &   & & &   &   & 1 & 1\\
              &   & 1 & 1 &   &   &   & & &   &   &   & 1\\
              &   &   &   &   &   &   & & 1 & 1 & 1 &   &
        \end{pNiceArray},
    \end{equation*}
    where the columns to the right of the vertical line correspond to the cellulation. Note that the edge $4-8$ does not participate in any cycles and corresponds to the empty column. The new edges are not associated to any stabilizers, so $f^{(1)}$ simply adjoins zero rows. With the new $\mathcal{X}_1$,
    \begin{equation*}
        \partial^{(1)}_1 = \begin{pmatrix}
            1 & 1 &   &   &   &   &   &   &   &\\
              & 1 & 1 &   &   &   &   &   &   &\\
              &   & 1 & 1 &   &   &   &   &   &\\
              &   &   & 1 & 1 &   &   &   &   &\\
              &   &   &   & 1 & 1 &   &   &   &\\
              &   &   &   &   & 1 & 1 &   &   &\\
            1 &   &   &   &   &   & 1 &   &   &\\
              &   &   & 1 &   &   &   & 1 &   &\\
              &   &   &   &   &   &   & 1 & 1 &\\
              &   &   &   &   &   &   &   & 1 & 1\\
              &   &   &   &   &   &   & 1 &   & 1\\
              \hline
              & 1 &   &   &   & 1 &   &   &   &\\
              &   & 1 &   & 1 &   &   &   &   &
        \end{pmatrix},
    \end{equation*}
    where the rows below the horizontal line correspond to the cellulation.
\end{example}

\noindent {\bf Remark:} The process of finding $\partial^{(i)}_1$, $f^{(i)}_1$, and $f^{(i)}_0$ (pre-cellulation) can be simplified to following these operations:
\begin{equation}\label{derive-f-partial}
  \begin{pmatrix}H_X & I_{n_X} \\ I_n & \end{pmatrix}
  \xrightarrow[\text{columns}]{\text{delete}}
  \begin{pmatrix}H_X|_{\mathrm{supp}(Z_i)} & I_{n_X} \\ f_1^{(i)} & \end{pmatrix}
  \xrightarrow[\text{rows}]{\text{delete}}
  \begin{pmatrix} \partial^{(i)}_1 & \left(f^{(i)}_0\right)^{\mathrm{T}} \\ f^{(i)}_1 & \end{pmatrix},
\end{equation}
where we delete the columns of \(\begin{pmatrix} H_X \\ I_n \end{pmatrix}\) corresponding to qubits \(\{1,\dots,n\} \setminus \mathrm{supp}(Z_i)\) followed by deleting the rows of \(\begin{pmatrix} H_X|_{\mathrm{supp}(Z_i)} & I_{n_X} \end{pmatrix}\) where \(H_X|_{\mathrm{supp}(Z_i)}\) is zero.\\

\begin{example}
    We began this section by pointing out why we could not use gauging to reduce $w_Z$. Instead we found a cycle basis, cellulated, and used the mapping cone. Here we follow~\cite{hastings2021quantum} and~\cite{wills2023} in cellulating cycles down to weight four, but one could choose an alternative scheme. Consider cellulating the octogon of solid edges $1$ to $8$ \cref{octogon} by triangulation via the dashed edges $9$ to $13$:
    \begin{equation}\label{octogon}
        \includegraphics{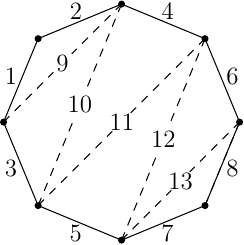}.
    \end{equation}
    This transforms the cycle $\begin{pmatrix}
        1 & 1 & 1 & 1 & 1 & 1 & 1 & 1
    \end{pmatrix}$ to
    \begin{equation*}
        \begin{pmatrix}
            1 & 1 & & & & & & & 1 & & & & \\
              & & 1 & & & & & & 1 & 1 & & & \\
              & & & 1 & & & & & & 1 & 1 & & \\
              & & & & 1 & & & & & & 1 & 1 & \\
              & & & & & 1 & & & & & & 1 & 1 \\
              & & & & & & 1 & 1 & & & & & 1
        \end{pmatrix},
    \end{equation*}
    where the columns represent the edges of \cref{octogon} in numerical order. This is the gauging equation, \cref{gaugingeq}. In this sense, cellulating plays the role of gauging for cycles by gauging a single row of $\partial^{(i)}_0$. If every cycle in the cycle basis is cellulated, then every row of $\partial^{(i)}_0$ is ``gauged'' for a fixed $i$. However, this will not necessarily be equivalent to applying gauging directly because the edge labels must be consistent between all of the cycles and therefore the columns cannot always be simultaneously permuted into the form \cref{gaugingeq} for every cycle.
\end{example}

Now that the full coning procedure has been described, we return to the logical operators (\cref{homeq}). We will assume to simplify the discussion that we are reducing $Z$ stabilizers one at a time, i.e., we are using \cref{mappingconesingle} and not \cref{mappingconemultiple}. If an extra round of thickening and choosing heights is used, the logical operators described here will be modified as discussed in that section.

By definition, $H_1(\mathcal{A}_i) = \ker \partial^{(i)}_1$. The space $\mathcal{Q}_i$ contains the full support of $Z_i$. Each element of $\mathcal{X}_i$ is associated with two elements of $\mathcal{Q}_i$ by construction, so rows of $\partial^{(i)}_1$ will always have weight two. Thus, the all-ones vector will always be an element of $\ker \partial^{(i)}_1$. This is equivalent to saying that the $Z_i$ commutes with all of the $X$ stabilizers it overlaps with. Anything else in the kernel commutes with all the $X$ stabilizers in the overlap of $Z_i$ and is therefore a $Z$ stabilizer or a $Z$ logical operator contained entirely in the support of $Z_i$. A stabilizer doesn't change the result when we take the quotient of the old logical operators and it. A $Z$ logical operator is more problematic, and we simply define away this problem by declaring coning to only apply to codes in which no $Z$ logical operator is entirely contained within the support of a $Z$ stabilizer generator. Hastings called these codes \emph{reasonable}; see \cite{hastings2021quantum} for the modifications to coning required for \emph{unreasonable} codes. Assuming we have a reasonable code, applying the exactness condition throughout \cref{homeq}, the Isomorphism Theorem (for groups) gives $H_1\left(\mathrm{cone}\left(f^{(i)}\right)\right) \cong H_1(\mathcal{B}) / \ker \partial^{(i)}_1$. This removes the extra logical operator caused by removing $Z_i$ for reduction, showing $\tilde{k} = k$.

It is clear from the matrices that the parameters after coning are more complicated than the previous steps. We have already showed that $\tilde{k} = k$. Before cellulation, $w_Z \leq 5$. Cellulation can increase this. The more a vertex is reused for adding new edges to create the cellulation, the more $w_Z$ can increase. If a vertex is used $x$ times in cellulation, there will be a $Z$-stabilizer generator of at least weight $x + 2$, up to potentially $x + 4$. Coning does not change $q_Z$. The effect on $q_X$ depends on how many edges are reused in cycles. This could be a large number, in which case a second round of thickening and choosing heights could be necessary, this time swapping the roles of $X$ and $Z$. If this is done, $w_Z$ will increase by two, $q_Z$ will still remain the same, and $\tilde{w}_X$ will be the maximum of $w_X$ before the second round of thickening and $2 + q_Z$. The following lemma clarifies the asymptotic result stated in \cite{hastings2021quantum}.
\begin{lemma}
    Given a code that has undergone copying, gauging, thickening, and choosing heights, let $h$ be the maximum number of times any single height was chosen. Then the reduced cone code has $\tilde{w}_X \leq 5 + h$.
\end{lemma}
\begin{proof}
    Having undergone copying, gauging, thickening, and choosing heights, $w_X \leq 5$ before coning. Cellulation ensures that new $X$ stabilizers $\partial^{(i)}_0$ have a maximum weight of four. Thus we only need to check how much the weight of the old $X$ stabilizers increases under coning, i.e., the bottom row of $\tilde{H}_X$ in \cref{coningstabs}: $\begin{pmatrix} f^{(1)}_0 & \cdots & f^{(m)}_0 & H_X \end{pmatrix}$. Having undergone copying, gauging, thickening, and choosing heights, the $X$ stabilizers overlap with the $Z$ stabilizers we reduce either zero or two times (\cref{thickening}), contributing either none or one element to $\mathcal{X}_i$, respectively. So $f^{(i)}_0$ has a maximum row weight of one by construction.
    
    Recall that we only reduce the weight of $Z$ stabilizers from the rows $\begin{pmatrix} H_Z \otimes I_\ell & 0\end{pmatrix}$ in \cref{thickening} where only one row is kept per $\ell$ rows. The support of two $Z$ stabilizers are disjoint if they are associated with different heights. This row overlaps the support of the $X$ stabilizers in the $H_X \otimes I_\ell$ block of the thickening stabilizers, which has $w_X \leq 3$ due to gauging. If an $X$ stabilizer overlaps with one of the $Z$ stabilizers we reduce, then at least two qubits of its support are within the set of qubits associated with the height of that $Z$ stabilizer. That leaves a maximum of one qubit that could be in any other height. By commutativity, it therefore cannot overlap with any $Z$ stabilizer from another height. Thus, an $X$ stabilizer can only get added weight from $f^{(i)}_0$ within exactly one height.
\end{proof}

\begin{figure}[p]\centering
  \subfloat[width=\textwidth][We wish to reduce the weight of the stabilizer $Z_i$ via coning. The shaded/dashed portion of the Tanner graph will remain unchanged, and for clarity, we omit it in the following part (b).]{
  \makebox[16.5cm][c]{
    \includegraphics[width=.3\textwidth]{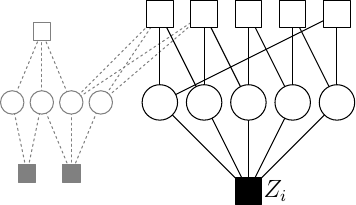}
  }}
  \\
  \subfloat[][Qubits $v_1$ through $v_5$ are new. The number of new qubits is always exactly equal to the number of $X$ stabilizers that share some support with $Z_i$, and the number of $Z$ stabilizers replacing $Z_i$ is always exactly the number of qubits in the support of $Z_i$. The map $f^{(i)}_i$ always has this form, and $\partial^{(i)}_1$ is always the same structure as $H_X\vert_{\mathrm{supp}(Z_i)}$, i.e., $\left(\partial^{(i)}_1\right)^{\mathrm{T}}$ looks like $H_X\vert_{\mathrm{supp} (Z_i)}$ with variables and checks swapped. Note that this created a new $X$-logical operator on the qubits $v_1, v_2, v_3, v_4, v_5$. This is why we need $\partial^{(i)}_0$.]{
  \makebox[16.5cm][c]{
    \includegraphics[width=.5\textwidth]{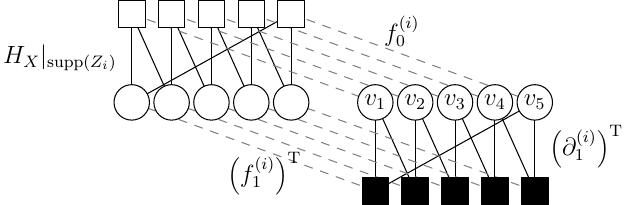}
  }}
  \\
  \subfloat[][Redrawing just $\left(\partial^{(i)}_1\right)^{\mathrm{T}}$ on its own, we can clearly see the cycle that causes the new logical operator.]{\centering
    \includegraphics[width=.2\textwidth]{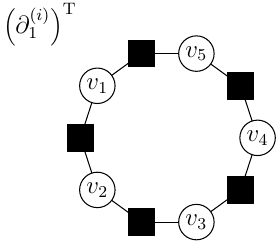}
  }
  \hspace{1.5cm}
  \subfloat[][We add the logical operator to the $X$ stabilizers, creating the structure of $\partial^{(i)}_0$.]{\centering
    \includegraphics[width=.2\textwidth]{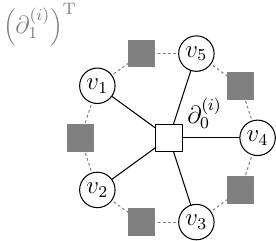}
  }
  \hspace{1.5cm}
  \subfloat[][The weight of the new $X$ stabilizer is high, so we cellulate which requires an additional qubit $v_6$.]{\centering
    \includegraphics[width=.2\textwidth]{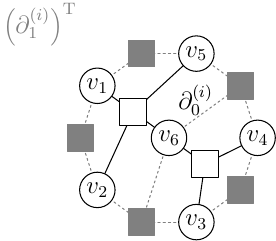}
  }
  \\
  \subfloat[][Putting the entire graph back together, we have the completed result of coning.]{\centering
    \includegraphics[width=.6\textwidth]{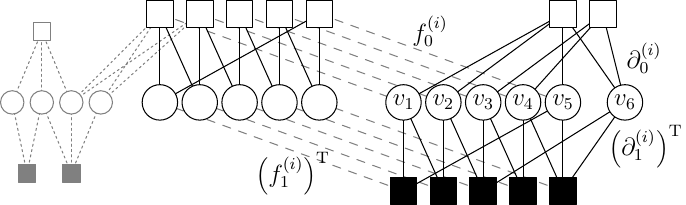}
  }
  \caption{Illustration of coning on the Tanner graph.}
  \label{fig:coning}
\end{figure}

\subsection{Reducing The Overhead With Improved Copying}\label{sec:modifiedcopying}
With the above steps using $l = 3$ and $\mathrm{heights} = (2, 1, 2, 1, 2, 3, 1, 3, 3, 1)$, the $\mathrm{QRM}(4)$ stabilizers \cref{RMstabs} produce a $[\![724, 1, 3]\!]$ code. This is a significant increase in resources to protect a single logical qubit. A small modification to Hastings’s method can make a large difference with respect to this problem. Looking at the Tanner graph in \cref{fig:Tannercopyingb}, any column whose weight is not equal to $q_X$ ends up with unused qubits. The simplest and most natural modification to make is to only copy a qubit as many times as its column weight, as in \cref{fig:modifiedcopying}. Applying this to the Reed-Muller code produces a $[\![512, 1, 2]\!]$ code. The reduction of resources propagated through the full procedure is even more dramatic when $q_X$ is much larger than the median column weight.

\begin{figure}[t!]\centering
  \subfloat[][]{\centering
    \includegraphics[width=.2\textwidth]{figures/fig-copying2-a.pdf}
  }
  \hspace{2cm}
  \subfloat[][]{\centering
    \includegraphics[width=.2\textwidth]{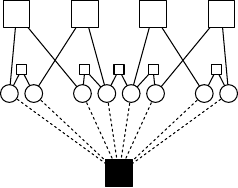}
    \label{fig:modifiedcopying}
  }
  \hspace{2cm}
  \subfloat[][]{\centering
    \includegraphics[width=.2\textwidth]{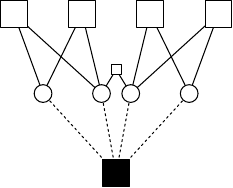}
    \label{fig:targetedcopying}
  }
  \caption{(a) The original Tanner graph. (b) Applying reduced copying to (a). (c) Applying targeted copying to (a) with $\mathrm{targ}_{q_X} = 3$. See \cref{fig:copying} for the original copying procedure applied to (a).}
  \label{fig:copying2}
\end{figure}

Note that copying results in $q_X = 3$, but some of the new qubits are only in the support of two $X$ stabilizers for both the original and the modified copying procedures. We can choose to only expand enough so that all new qubits are in the support of exactly $\mathrm{targ}_{q_X}$ $X$ stabilizers for some target column weight $\mathrm{targ}_{q_X} \geq 3$. Referring to \cref{fig:Tannercopyingb}, observe that $v_{i, 1}$ and $v_{i, q_X}$ can accept $\mathrm{targ}_{q_X} - 1$ edges instead of $\mathrm{targ}_{q_X} - 2$, reducing the number of qubits needed for copying. Applying this to \cref{fig:Tannercopyingb} with $\mathrm{targ}_{q_X} = 3$ produces \cref{fig:targetedcopying}. Applying this to the Reed-Muller code with $\mathrm{targ}_{q_X} = 3$ produces a $[\![315, 1, 2]\!]$ code.

The improved copying methods just introduced do not decrease the distance. The drop in distance from the original $[\![15, 1, 3]\!]$ code in the previous examples is due to the other steps of quantum weight reduction. The codes after copying, reduced copying, and targeted copying have parameters $[\![60, 1, 7]\!]$, $[\![32, 1, 4]\!]$, and $[\![16, 1, 3]\!]$, respectively.

\begin{figure}[t!]\centering
  \subfloat[][]{\centering
    \includegraphics[height=2cm]{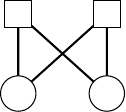}
    \label{fig:copying-cycles-a}
  }
  \hspace{2cm}
  \subfloat[][]{\centering
    \includegraphics[height=2cm]{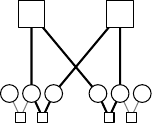}
    \label{fig:copying-cycles-b}
  }
  \hspace{2cm}
  \subfloat[][]{\centering
    \includegraphics[height=2cm]{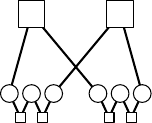}
    \label{fig:copying-cycles-c}
  }
  \\
  \subfloat[][]{\centering
    \includegraphics[height=2cm]{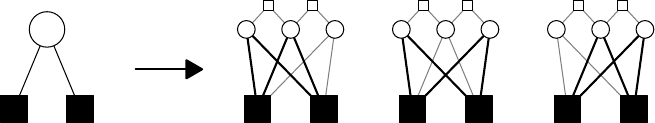}
    \label{fig:copying-cycles-d}
  }
  \caption{(a) A 4-cycle in $X$. Under copying, it could be expanded to an 8-cycle (b) or a 12-cycle (c) depending on choices for the connections of the original $X$ stabilizers. (d) The three new 4-cycles in $Z$ from the two $Z$ stabilizers which shared the same variable node.}
  \label{fig:copying-cycles}
\end{figure}

\subsection{Quantum Weight Reduction Reinterpreted}
All three copying variants, Hastings's original (\cref{fig:Tannercopyingb}), the reduced copying (\cref{fig:modifiedcopying}), and targeted copying (\cref{fig:targetedcopying}) may be viewed in the same framework. With the hindsight of coning, it is clear that \cref{tmatrixcopying} is in the form of a mapping cone where $H_\ell$ is given by \cref{repcode}:
\begin{equation}
    \tilde{H}_X = \begin{pmatrix}
        H^{(r)}_X & f^{\mathrm{T}}_1\\
        & H_\ell
    \end{pmatrix}
    \quad , \quad
    \tilde{H}_Z = \begin{pmatrix}
        H^{(r)}_Z & f_2
    \end{pmatrix}.
\end{equation}
Proceeding one column at a time, remove the column of weight $q_i$ to be reduced to get $H^{(r)}_X$ and $H^{(r)}_Z$. Without all columns, these may not commute and therefore cannot form the standard chain complex. Instead, we have
\begin{equation*}
    \begin{tikzcd}
	{C_1} & {C_2} \\
	{C_0} & {\mathbb{F}^\ell_2} & {\mathbb{F}^{\ell - 1}_2}
	\arrow["{f_2}"', from=2-2, to=2-1]
	\arrow["H^{\mathrm{T}}_\ell"', from=2-3, to=2-2]
	\arrow["{f_1}", from=1-2, to=2-2]
	\arrow["{\left(H^{(r)}_X\right)^{\mathrm{T}}}"', from=1-2, to=1-1]
	\arrow["{H^{(r)}_Z}"', from=1-1, to=2-1]
    \end{tikzcd}.
\end{equation*}
In order to be a valid chain complex, we need $f_2 H^{\mathrm{T}}_\ell = 0$, and in order to be a valid chain map, we need $f_2 f_1 = H^{(r)}_Z \left(H^{(r)}_X\right)^{\mathrm{T}}$. The number of solutions for $f_2$ is the size of the left kernel of $\begin{pmatrix} f_1 & H^{\mathrm{T}}_\ell\end{pmatrix}$. Viewed as a matrix, each column of $f_1$ must have weight zero or one, and the columns should have support exactly corresponding to the row-support of the deleted column of $H_X$. Choosing $\ell=q_X$ uniformly for all variables and limiting the row weight of $f_1$ to one gives Hastings' copying. Choosing $\ell=q_i$ separately for each variable and limiting the row weight of $f_1$ to one results in reduced copying. Choosing $\ell=q_i-2$ and letting the first and last rows of $f_1$ have weight two with other rows having weight one results in targeted copying with $\mathrm{targ}_{q_X} = 3$.

Gauging removes rows of $H_X$, preserving the commutativity between the stabilizers and can therefore form the standard chain complex. Removing a row of weight $w_i$ from $H_X$ to obtain $H_X^{(r)}$, the associated complex is
\begin{equation*}
    \begin{tikzcd}
	{C_2} & {C_1} & {C_0} \\
	{\mathbb{F}_2^{\ell-1}} & {\mathbb{F}_2^\ell}
	\arrow["{H_Z^{\mathrm{T}}}", from=1-1, to=1-2]
	\arrow["{H^{(r)}_X}", from=1-2, to=1-3]
	\arrow["{H^{\mathrm{T}}_\ell}", from=2-1, to=2-2]
	\arrow["{f_2}"', from=1-1, to=2-1]
	\arrow["{f_1}"', from=1-2, to=2-2]
    \end{tikzcd}
\end{equation*}
with $\ell = w_i - 2$, the columns of $f_1$ corresponding to the support of the deleted $X$ stabilizer having weight one, other columns having weight zero, first and last rows having weight two, and all other rows having weight one. The stabilizers are (compare to \cref{ex:gauging})
\begin{equation*}
    \tilde{H}_X = \begin{pmatrix}
        H^{(r)}_X & \\
        f_1 & H^{\mathrm{T}}_\ell
    \end{pmatrix}
    \quad , \quad
    \tilde{H}_Z = \begin{pmatrix}
        H_Z & f^{\mathrm{T}}_2
    \end{pmatrix}
\end{equation*}
and $H^{\mathrm{T}}_\ell f_2 = f_1 H_Z^{\mathrm{T}}$. There is a unique solution for $f_2$ since $H^{\mathrm{T}}_\ell$ has a left inverse.

Quantum weight reduction may therefore be seen as two cones, a tensor product of chain complexes, and then another cone.

\subsection{Effects On Iterative Decoding}\label{sec:itdec}
Despite its advantages, quantum weight reduction generally alters the underlying structure of the original code. In the absence of any obvious structure in the output, the available decoders are those applicable to wide ranges of codes, such as belief propagation. These iterative algorithms are highly sensitive to the topology of its Tanner graph \cite{ryan2009channel}. The fact that stabilizers must commute forces all stabilizer codes to have 4-cycles. For CSS codes, these unavoidable cycles are between the $X$ and $Z$ stabilizers. If $X$-$Z$ correlations are ignored and the two types of stabilizers are decoded independently, it is beneficial to reduce the number of short cycles from only $H_X$ and only $H_Z$. Unfortunately, the weight reduction procedure above introduces a significant amount of new 4-cycles.

Since the Tanner graph of the repetition code (and its transpose) has no cycles, expanding the variable nodes in copying does not introduce any new cycles in the $X$ stabilizers. However, the degree of freedom present is assigning a qubit to one of its copies in the repetition code can have important consequences. Consider the 4-cycle in \cref{fig:copying-cycles-a}. Either \cref{fig:copying-cycles-b} or \cref{fig:copying-cycles-c} are valid choices, but one leads to an 8-cycle and the other to a 12-cycle. If the cycle structure of the input is known, selectively assigning edges can be used to lengthen short cycles of the original graph. On the other hand, an identical copy of the $Z$ stabilizers is imposed on every variable node in the repetition code, leading to many new 4-cycles; see \cref{fig:copying-cycles-d}. A simple counting argument shows the following.
\begin{lemma}
    Splitting a variable node into $s$ variable nodes during copying introduces $\binom{s}{2} \binom{c}{2}$ new 4-cycles in the $Z$-only Tanner graph and $c (s - 1)$ new 4-cycles between $X$ and $Z$, where $c$ is the number of $Z$ check nodes connected to the original variable node.
\end{lemma}
\noindent Note that this lemma applies to both the original and modified copying procedures. Applying this to the 15-qubit Reed-Muller code \cref{RMstabs}, Hastings's copying produces $738 \, Z + 168 \, X\text{-}Z = 906$ new 4-cycles; the reduced copying produces $468 \, Z + 96 \, X\text{-}Z = 564$ new 4-cycles; and the targeted copying with $\mathrm{targ}_{q_X} = 3$ produces $45 \, Z + 10 \, X\text{-}Z = 55$ new 4-cycles.

Similar to copying, gauging only modifies the length of the $X$ cycles, while new 4-cycles are created in both $Z$ and between $X$ and $Z$. Although the number of each is unpredictable due to the nature of the new $Z$ stabilizers, it is often significant due to the density of the solution.

The terms $H_X \otimes I_\ell$, $H^{\mathrm{T}}_X \otimes I_{\ell - 1}$, and $H_Z \otimes I_\ell$ in thickening make copies of each $X$ or $Z$ cycle present up to this point in the procedure. Additionally, $\tilde{H}_Z$ contains both $H_X$ and $H_Z$, converting cycles between $X$ and $Z$ to purely $Z$ cycles. These new cycles are connected through a copy of the repetition code in $I_n \otimes H_\ell$ and are therefore potentially lengthened. (See \cref{fig:thickeningHastingsbZ}.) Choosing heights removes some of the cycles coming from the $H_Z \otimes I_\ell$ term.

Every edge of the graph $\mathcal{Q}_i \to \mathcal{X}_i$ gives an $X\text{-}Z$ 4-cycle in the resulting code, but exactly half that number were removed by deleting $Z_i$. Every multi-edge in $\mathcal{Q}_i \to \mathcal{X}_i$ results in a $Z$ 4-cycle. Any two cycle basis elements that share two edges in $\mathcal{X}_i \to \mathcal{R}_i$ results in an $X$ 4-cycle. Cellulation results in new $X\text{-}Z$ 4-cycles. If thickening is performed here to reduce $q_X$ after coning, then the results of the previous paragraph should be applied to these results.

The number of new short cycles introduced by quantum weight reduction severely degrades the performance of iterative decoding schemes based on the Tanner graph, requiring cycle mitigation techniques or expensive post-processing such as ordered-statistics decoding (OSD).

\section{Classical Weight Reduction}\label{sec:classicalWtR}
A natural question is how the above quantum weight reduction method (and our proposed modifications) compares to simply reducing the column and row weights of the classical codes prior to their use in constructing a quantum code. We restrict ourselves to some well-known code constructions that produce valid stabilizers regardless of the input---such as the hypergraph and lifted product codes, \cref{hgp,eq:lp}, respectively---whose row and column weights are completely determined by the classical inputs. This technique can also be used for other constructions, but we leave this to future work.

It is not common to consider weight reduction in classical coding theory since there is no equivalent hardware requirement and there are numerous methods to construct excellent LDPC codes for efficient decoding applications. Nevertheless, we draw inspiration from the quantum case to define a procedure which is simple and maintains or increases the (classical) minimum distance. The result is similar to the copying/gauging steps described above but with a few minor changes that have important consequences. For clarity, we refer to Hastings's work as \emph{quantum} weight reduction and this method as \emph{classical} weight reduction.

Similar to the previous section, let $H$ be a parity-check matrix of an $[n, k, d]$ linear code with rows $\{h_i\}$, let $w_i = \mathrm{wt}(h_i)$ be the weight of the $i$th row, and let $q_i$ be the weight of the $i$th column. Consider a row $h_i$ with $w_i > 3$, and assume, as before, that the support of $h_i$ has been permuted to the first $w_i$ bits. Then to weight reduce, add $w_i - 1$ new columns to $H$ and replace $h_i$ with the matrix $\begin{pmatrix} I_{w_i} & 0 & H^{\mathrm{T}}_{w_i - 1}\end{pmatrix}$, where $I_{w_i}$ is the $w_i \times w_i$ identity matrix, 0 represents the rest of the original columns not in the support of $h_i$, and $H^{\mathrm{T}}_{w_i}$ is the transpose of \cref{repcode}. In matrix form,
\begin{equation}\label{classicalwtred}
    \begin{pNiceArray}{ccccccc}[first-row]
        v_1 & \cdots & v_w & & & &\\
        1 & \cdots & 1 & 0 & \cdots & 0
    \end{pNiceArray}
    \mapsto
    \begin{pNiceArray}{ccccccccccc}[first-row]
        v_1 & v_2 & \cdots & v_{w - 1} & v_w & & v^\prime_1 & v^\prime_2 & \cdots & v^\prime_{w - 2} & v^\prime_{w - 1}\\
        1 &   &        &   &   & \cdots & 1 &   &   &   & \\
          & 1 &        &   &   & \cdots & 1 & 1 &   &   & \\
          &   & \ddots &   &   & \cdots &   & & \ddots &   & \\
          &   &        & 1 &   & \cdots &   &   &   & 1 & 1\\
          &   &        &   & 1 & \cdots &   &   &   &   & 1
    \end{pNiceArray}.
\end{equation}
Repeat this for all rows which need to be reduced. To reduce the columns, simply transpose and use the same method. This is summarized in \cref{alg:classicalwtred}. Note that the row-wise sum of $\begin{pmatrix} I_{w_i} & 0 & H^{\mathrm{T}}_{w_i - 1}\end{pmatrix}$ is just $\begin{pmatrix} h_i & 0\end{pmatrix}$.

There is only a slight difference between \cref{classicalwtred,gaugingeq} (gauging) and \cref{fig:TGwtred,fig:gaugingb}. This change increases the (classical) minimum distance of the reduced code at the cost of decreased encoding rate.
\begin{theorem}\label{thm:classicaldistthm}
    Let $H$ be the parity-check matrix of an $[n, k, d]$ binary, linear code, $w_H$ be the maximum row weight of $H$, and $q_H$ be the maximum column weight of $H$. Then \cref{alg:classicalwtred} outputs a parity-check matrix $\tilde{H}$ with $w_{\tilde{H}} = \rho_{\tilde{H}} = 3$ whose code has parameters $[O(n \rho), k, \tilde{d}]$, where $\rho = \max \{w_H, q_H\}$ and $\tilde{d} \geq d$. Additionally,
    \begin{itemize}
        \item[(a)] if all rows of $H$ have weight greater than three, then $\tilde{d} \geq 3d/2$;
        \item[(b)] if all columns of $H$ have weight greater than three, then $\tilde{d} \geq d \min_i q_i$;
        \item[(c)] if all rows and all columns of $H$ have weight greater than three, then $\tilde{d} \geq (3d \min q_i)/2$.
    \end{itemize}
\end{theorem}
\begin{proof}
    The claims about $w_H$, $q_H$, and the length follow from \cref{classicalwtred}. Weight reducing a row does not change column weights and vice versa, so implementing one does not undermine the other. Reducing a row or a column appends the same number of linearly independent rows as columns, so the rank is always preserved.

    Let $H$ be a parity-check matrix with rows $\{h_1, \hdots, h_{n - k}\}$. Reducing row $h_i$ produces the parity-check matrix
    \begin{equation}\label{proofmatrix}
        \left(\begin{array}{@{}c|c@{}}
          h_1 & \\
          \vdots & 0 \\
          h_{i-1} & \\
          \hline
          f & H^{\mathrm{T}}_{w_i} \\
          \hline
          h_{i+1} & \\
          \vdots & 0 \\
          h_{n-k} & 
        \end{array}\right),
    \end{equation}
    where $w_i = |\mathrm{supp} \, h_i |$ is the weight of row $i$ and $f\big|_{\mathrm{supp} \, h_i} = I_{w_i}$ is the $w_i \times w_i$ identity. Denoting the columns to the left of the vertical line by $\mathrm{old}$ and those to the right by $\mathrm{new}$, $f\big|_{\mathrm{old} \backslash \mathrm{supp} \, h_i} = 0$. If $c$ is a codeword of \cref{proofmatrix}, then $c\big|_{\mathrm{old}}$ satisfies the checks $h_j$ for $j \neq i$ and $\begin{pmatrix} f & H^{\mathrm{T}}_{w_i - 1}\end{pmatrix} c = 0$. Then
    \begin{equation*}
	0 = \begin{pmatrix} f & H^{\mathrm{T}}_{w_i - 1}\end{pmatrix} c = \begin{pmatrix} f & H^{\mathrm{T}}_{w_i - 1}\end{pmatrix} \begin{pmatrix} c\big|_{\mathrm{old}} \\ c\big|_{\mathrm{new}} \end{pmatrix} = f c\big|_{\mathrm{old}} + H^{\mathrm{T}}_{w_i - 1} c\big|_{\mathrm{new}}
    \end{equation*}
    gives $H^{\mathrm{T}}_{w_i - 1} c\big|_{\mathrm{new}} = f c\big|_{\mathrm{old}} = c\big|_{\mathrm{supp} \, h_i}$. The image of $H^{\mathrm{T}}_{w_i - 1}$ are even-weight vectors, so $c\big|_{\mathrm{supp} \, h_i}$ has even weight and therefore $c\big|_{\mathrm{old}}$ satisfies $h_i$. Hence, $c\big|_{\mathrm{old}}$ is a codeword of the original code.

    Now suppose $c$ is a codeword of \cref{proofmatrix} such that $c\big|_{\mathrm{old}} = 0$. Then $H^{\mathrm{T}}_{w_i - 1} c\big|_{\mathrm{new}} = 0$. But $\ker H^{\mathrm{T}}_{w_i - 1} = 0$ so $c\big|_{\mathrm{new}} = 0$. There is therefore a one-to-one correspondence between codewords of the original code and codewords of \cref{proofmatrix} ($\tilde{k} = k$). It also follows that the minimum-weight codeword of \cref{proofmatrix} is at least as large as the old code ($\tilde{d} \geq d$).

    To prove (a), assume $w_i > 3$ for all $1 \leq i \leq n - k$ and weight reduce all rows. A codeword $c$ has $\mathrm{wt}\left(c\big|_{\mathrm{old}}\right) \geq d$ and $f c\big|_{\mathrm{old}}$ flips at least $d$ syndromes. A single one in $c\big|_{\mathrm{new}}$ flips two syndromes, so at least $d/2$ ones are required to make $c$ commute: $\mathrm{wt}(c) = \mathrm{wt}\left(c\big|_{\mathrm{old}}\right) + \mathrm{wt}\left(c\big|_{\mathrm{new}}\right) \geq d + d/2 = 3d/2$.

    To prove (b), observe that any bit of the original code that is supported on a column that is expanded must now have support on the full repetition code in the expanded section, as this is the only way to commute with the newly added rows. This increases the weight of any codeword supported on this bit; however, not all minimum weight codewords may have support here, so the minimum distance may not increase. If all columns are expanded, the distance increases by at least a multiplicative factor of the smallest column weight.

    To prove (c), combine (a) and (b).
\end{proof}

\begin{algorithm}[t]
    \DontPrintSemicolon
    \SetAlgoLined
    \SetKwComment{Comment}{// }{}
	
    \KwInput{$H$}
    \KwOutput{$\tilde{H}$}

    \BlankLine

    \For{$i \gets 1 \textbf{ to } \text{number of rows of $H$}$}{
        $h_i \gets \text{$i$th row of $H$}$ \;
        \uIf{$\mathrm{wt}(h_i) > 3$}{
            \uIf{compressed}{
                $f \gets \begin{pmatrix}
                    1 & 1 &   &        &   & & \\
                      &   & 1 &        &   & & \\
                      &   &   & \ddots &   & & \\
                      &   &   &        & 1 & & \\
                      &   &   &        & & 1 & 1
                \end{pmatrix}$ on $\mathrm{supp}(h_i)$\;
                $\ell \gets \mathrm{wt}(h_i) - 3$ \;
            }
            \Else{
                $f \gets $ identity matrix on $\mathrm{supp}(h_i)$\;
                $\ell \gets \mathrm{wt}(h_i) - 1$ \;
            }
            \BlankLine
            \If{permute}{
                \Comment*[l]{randomly permute the nonzero columns of $f$}
            }
            \BlankLine
            \uIf{$\tilde H$ exists}{
                $\tilde{H} \gets \begin{pmatrix} \multicolumn{2}{c}{\tilde{H}} & 0 \\ f & 0 & H^{\mathrm{T}}_{\ell} \end{pmatrix}$}
            \Else{
                $\tilde H \gets \begin{pmatrix}f & H^{\mathrm{T}}_\ell \end{pmatrix}$
            }
        }
        \Else{
            \uIf{$\tilde H$ exists}{
                $\tilde{H} \gets \begin{pmatrix} \multicolumn{2}{c}{\tilde{H}} \\ h_i & 0\end{pmatrix}$
            }
            \Else{
                $\tilde{H} \gets h_i$
            }
        }
    }

    \BlankLine

    \Comment*[l]{Apply the above to $\tilde{H}^{\mathrm T}$ and transpose back}

    \BlankLine
	
    \KwRet{$\tilde{H}$}\;
	
    \caption{Classical Weight Reduction}
    \label{alg:classicalwtred}
\end{algorithm}

\noindent Note that nowhere in the proof did we assume that $H$ is full rank, and therefore the proof holds for parity-check matrices with redundant checks.\\

\noindent {\bf Remark:} The proof of \cref{thm:classicaldistthm} shows that weight reduction affects the generator and parity-check matrices of the code asymmetrically. This has implications for properties based on duality. In particular, for a linear code with parity-check matrix $H$ and whose (Eulidean) orthogonal space is determined by the row space of $G$, the linear codes based on $\tilde{G}$ and $\tilde{H}$ are no longer dual. This rules out using classical weight reduction in the CSS construction of stabilizer codes.\\

Suppose that we instead perform weight reduction using \cref{gaugingeq} (gauging) rather than \cref{classicalwtred}. Comparing the two forms of $f$, we will refer to this as \emph{compressed} classical weight reduction since it requires the addition of fewer additional bits. We still have $\tilde{d} \geq d$ but now part (a) of the proof does not hold (consider a codeword with overlap of two with the check being reduced such that the overlap falls within a single row of the expansion $f$, and we see that a generalized proof of part (a) would fail), and part (b) holds but with a different constant, $\tilde{d} \geq d \min_i (q_i - 2)$. In general, there are many possible choices of row expansions: Any choice of $f$ with weight one in columns corresponding to the support of the check being reduce, and weight zero columns elsewhere is a valid expansion, though some choices may not result in effective weight reduction. Further increasing the number of columns of $H^{\mathrm{T}}_{w_i - 1}$ induces a small, additive constant in the distance but at the expense of a nearly linear increase in $n$. A row $h_i$ can be expanded to $\begin{pmatrix} f & H^{\mathrm{T}}_{r_i - 1}\end{pmatrix}$ where $f$ is any $r_i \times n$ matrix with column weight one on bits of the support of $h_i$ and column weight zero elsewhere. In this case, if $f$ has any row with weight greater than one, then part (a) of \cref{thm:classicaldistthm} can't be generalized to this choice of $f$, and instead the guarantee is that reduction will never reduce the minimum distance.  Applying this to the transpose results in column expansion, and if all columns are expanded, then the minimum distance further increases by a factor of at least the minimum size of all column expansions.

We note that classical weight reduction was previously presented in Appendix A of Ref.~\cite{hastings2021fiber}, using the language of cellular homology. As with quantum weight reduction, we believe that our presentation is complementary and remark that \cref{thm:classicaldistthm} also gives a tighter bound on the distance of the weight reduced code. In addition, the independent row and column permutations described in \cref{subsec:perm} are novel and can give significant improvement in the distances of the weight-reduced codes; see \cref{tab:hgp-examples}.

\subsection{Decoding}
Unlike quantum weight reduction, classical weight reduction is able to easily map back onto the original code. Inspired by the proof of \cref{thm:classicaldistthm}, we can use this to map the decoding problem of the weight-reduced code back to the decoding problem of the original code. To see how, suppose we are using the weight-reduced parity-check matrix to decode. An error on the bits corresponding to weight-reduced columns can be uniquely identified and corrected using the rows corresponding to the $H^{\mathrm{T}}_\ell$ blocks. Collapsing these columns into a single bit puts the matrix into the form: ``original bits'' followed by ``new bits''. The original bits are protected by the original parity-check matrix and the original decoder may be used. Errors on the new bits may be corrected using the $H^{\mathrm{T}}_\ell$ blocks. This may or may not be beneficial depending on the original decoder. Message-passing based decoders may perform better on the weight-reduced code than the original.

Graphically, a high-degree node (\cref{fig:TG}) is replaced by \cref{fig:TGwtred}. This has no cycles and lengthens any cycles the replaced node was previously a part of. The exact increase in cycle length depends on the ordering of the initial edges. The length-2 path from $v_{w - 1}$ to $v_w$ becomes a length-4 path through $v_{w - 1}$, $v^\prime_{w - 1}$, and $v_w$, while the length-2 path from $v_1$ to $v_w$ becomes the length-$(2w + 1)$ path through $v_1$, $v^\prime_1$, $v^\prime_2, \cdots$, $v^\prime_{w - 1}$, and $v_w$. If the cycle structure of the input code is known, permutations of the check sides of the original edges in the weight reduced Tanner graph can be used to selectively increase harmful short cycles. These permutations are the graphical representation of the permutations described in \cref{subsec:perm}. We have observed large increases in girth and improvements in iterative decoder performance using this approach.
\begin{figure}[t!]
    \centering
    \subfloat[][]{
        \centering
        \includegraphics[width=.3\textwidth]{figures/fig-gauging-a.pdf}
        \label{fig:TG}
    }
    \hspace{1cm}
    \subfloat[][]{
        \centering
        \includegraphics[width=.3\textwidth]{figures/fig-gauging-b.pdf}
    }
    \hspace{1cm}
    \subfloat[][]{
        \centering
        \includegraphics[width=.17\textwidth]{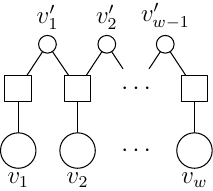}
        \label{fig:TGwtred}
    }
    \caption{(a) The input Tanner graph. (b) The compressed classical weight reduction of the check node in (a). (c) The classical weight reduction of the check node in (a). (a) and (b) are the same as \cref{fig:gauginga,fig:gaugingb}, respectively.}
\end{figure}

The above has effects on iterative decoding.
Recall the unique-neighbor property of expander codes\footnote{See~\cite{shankar2005expander} for a nice, low-level description of expander codes and how the unique-neighbor property of expander graphs plays a role in bounding the distance of the code.}: Consider a $(\ell, r, \varepsilon, \ell / 2)$-left expander with vertex bipartition $L \cup R$, (for all sets $S \subset L$ of size no more than $\varepsilon |L|$, $|\mathcal{N}(S)| \geq \ell / 2 |S|$, where $\mathcal{N}(S)$ is the neighborhood of $S$), then for these sets $\exists y \in R$ such that $|\mathcal{N}(y) \cap S| = 1$, (there exists a unique neighbor). Now suppose $S$ is a trapping set. Then unless there are many such $y$’s, there is not going to be a large amount of extrinsic information flowing into this set of nodes to help overcome the trapping set. A cycle with a lot of extrinsic information flowing into it will be able to overcome what would otherwise become a trapping set. In this sense, cycle connectivity is more important than cycle length for decoder performance.

On one hand, weight reduction creates a potentially large number of low degree nodes, which typically strongly degrade the performance of iterative decoders. On the other hand, any given cycle has the same number of extrinsic connections to the rest of the graph as in the original graph since weight reduction introduces no new cycles. What has changed is the ratio of extrinsic connections to the cycle length. Due to their low degree, certain patterns in the variable nodes introduced by the $H^{\mathrm{T}}_\ell$ block will cause the entire repetition code pattern to flip to errors under belief propagation. This can be overcome if the other variable nodes receive a high amount of incoming extrinsic information (have high degree) but becomes significantly more difficult if the columns are also weight reduced. However, the more rows and columns that are reduced, the more new vertices are created in the Tanner graph, the longer the cycles, and therefore probability that theses structured patterns of errors occur decreases. Hence, weight reduction introduces a significant amount of trapping sets but they become larger and therefore less harmful. In this sense, weight reduction introduces uneven error protection in the variable nodes.

See~\cite{ryan2009channel, tian2004selective} for further discussion of the concepts in this subsection.

\subsection{Permutations}\label{subsec:perm}
Permutations of the input are also important for reasons besides selectively increasing the girth. Consider an $[n, k, d]$ code with parity-check matrix $H$. Permuting the columns of $H$ produces an equivalent code with the same parameters and corresponds to a relabeling of the Tanner graph without any effect on cycle structure. However, weight reducing the original code and the permuted code produce \emph{nonequivalent} codes with potentially different distances.
\begin{example}
    Weight reducing
    \begin{equation*}
        \begin{pmatrix}
            1 & 1 & 1 & 1 & 0 & 0\\
            0 & 0 & 1 & 1 & 1 & 1
        \end{pmatrix}
    \end{equation*}
    and its permutation
    \begin{equation*}
        \begin{pmatrix}
            1 & 0 & 1 & 0 & 1 & 1\\
            0 & 1 & 1 & 1 & 1 & 0
        \end{pmatrix}
    \end{equation*}
    produces
    \begin{equation*}
        H = \begin{pmatrix}
            1 &   &   &   &   &   & 1 &   &   &   &   &  \\
              & 1 &   &   &   &   & 1 & 1 &   &   &   &  \\
              &   & 1 &   &   &   &   & 1 & 1 &   &   &  \\
              &   &   & 1 &   &   &   &   & 1 &   &   &  \\
              &   & 1 &   &   &   &   &   &   & 1 &   &  \\
              &   &   & 1 &   &   &   &   &   & 1 & 1 &  \\
              &   &   &   & 1 &   &   &   &   &   & 1 & 1\\
              &   &   &   &   & 1 &   &   &   &   &   & 1
        \end{pmatrix}
        \quad \text{and} \quad
        H^\prime = \begin{pmatrix}
            1 &   &   &   &   &   & 1 &   &   &   &   &  \\
              &   & 1 &   &   &   & 1 & 1 &   &   &   &  \\
              &   &   &   & 1 &   &   & 1 & 1 &   &   &  \\
              &   &   &   &   & 1 &   &   & 1 &   &   &  \\
              & 1 &   &   &   &   &   &   &   & 1 &   &  \\
              &   & 1 &   &   &   &   &   &   & 1 & 1 &  \\
              &   &   & 1 &   &   &   &   &   &   & 1 & 1\\
              &   &   &   & 1 &   &   &   &   &   &   & 1
        \end{pmatrix},
    \end{equation*}
    respectively. The linear code given by $H$ has parameters $[12, 4, 3]$ but $H^\prime$ is $[12, 4, 4]$.
\end{example}

More generally, we can use the replacement $\begin{pmatrix} \Pi & 0 & H^{\mathrm{T}}_\ell\end{pmatrix}$, where $\Pi$ is any permutation of the identity. Different permutations can be used for reducing each row and column independently. We use this in the next section, where the positive effect on distance can be seen in \cref{tab:hgp-examples,tab:lp-examples,tab:hgp-examples-app-A,tab:hgp-examples-app-B}. 
Quantum weight reduction should be equally susceptible to permutations; we leave this to future work.

\section{Examples And Numerical Results}\label{sec:num}
{\renewcommand{\arraystretch}{1.4}
\begin{center}
\begin{table}[]
    \begin{tabular}{||c|c|c|c|c|c|c||}
        \hline
        $\mathcal{C}(H)$ & $\mathrm{HGP}(H)$ & $R$ & $\mathrm{HGP}(\tilde{H})$ & $R$ & $\mathrm{HGP}(\tilde{H}^{(c)})$ & $R$ \\ \hline
        $[6, 3, 3]$ & $[\![45, 9, 3]\!]$ & 0.200 & $[\![117, 9, 4]\!]$ & 0.077 & $[\![65, 9, 4]\!]$ & 0.138\\ \hline
        $[7, 2, 4]$ & $[\![74, 4, 4]\!]$ & 0.054 & $[\![164, 4, 4]\!]$ & 0.024 & $[\![100, 4, 4]\!]$ & 0.040\\ \hline
        $[7, 3, 4]$ & $[\![65, 9, 4]\!]$ & 0.138 & $[\![149, 9, 5]\!]$ & 0.060 & $[\![89, 9, 4]\!]$ & 0.101\\ \hline
        $[7, 4, 3]$ & $[\![58, 16, 3]\!]$ & 0.276 & $[\![400, 16, 6\rightarrow7]\!]$ & 0.040 & $[\![136, 16, 3\rightarrow4]\!]$ & 0.118 \\ \hline
        $[8, 2, 5]$ & $[\![100, 4, 5]\!]$ & 0.040 & $[\![394, 4, 8]\!]$ & 0.010 & $[\![202, 4, 6]\!]$ & 0.020\\ \hline
        $[8, 3, 4]$ & $[\![89, 9, 4]\!]$ & 0.101 & $[\![317, 9, 6\rightarrow7]\!]$ & 0.028 & $[\![149, 9, 4\rightarrow5]\!]$ & 0.060 \\ \hline
        $[8, 4, 4]$ & $[\![80, 16, 4]\!]$ & 0.200 & $[\![656, 16, 7\rightarrow10]\!]$ & 0.024 & $[\![208, 16, 4\rightarrow6]\!]$ & 0.077 \\ \hline
        $[9, 2, 6]$ & $[\![130, 4, 6]\!]$ & 0.031 & $[\![514, 4, 10]\!]$ & 0.008 & $[\![290, 4, 8]\!]$ & 0.014\\ \hline
        $[9, 3, 4]$ & $[\![117, 9, 4]\!]$ & 0.077 & $[\![369, 9, 6\rightarrow7]\!]$ & 0.024 & $[\![185, 9, 4\rightarrow5]\!]$ & 0.049 \\ \hline
        $[9, 4, 4]$ & $[\![106, 16, 4]\!]$ & 0.151 & $[\![730, 16, 7\rightarrow10]\!]$ & 0.022 & $[\![250, 16, 4\rightarrow6]\!]$ & 0.064 \\ \hline
        $[9, 5, 3]$ & $[\![97, 25, 3]\!]$ & 0.258 & $[\![1313, 25, 7\rightarrow11]\!]$ & 0.019 & $[\![493, 25, 4\rightarrow8]\!]$ & 0.051 \\ \hline
        $[10, 2, 6]$ & $[\![164, 4, 6]\!]$ & 0.024 & $[\![650, 4, 10]\!]$ & 0.006 & $[\![394, 4, 8]\!]$ & 0.010\\ \hline
        $[10, 3, 5]$ & $[\![149, 9, 5]\!]$ & 0.060 & $[\![929, 9, 10]\!]$ & 0.010 & $[\![369, 9, 7]\!]$ & 0.024\\ \hline
        $[10, 4, 4]$ & $[\![136, 16, 4]\!]$ & 0.118 & $[\![458, 16, 6\rightarrow7]\!]$ & 0.035 & $[\![250, 16, 4\rightarrow5]\!]$ & 0.064\\ \hline
        $[10, 5, 4]$ & $[\![125, 25, 4]\!]$ & 0.200 & $[\![937, 25, 7\rightarrow9]\!]$ & 0.027 & $[\![377, 25, 4\rightarrow6]\!]$ & 0.066\\ \hline
        $[10, 6, 3]$ & $[\![116, 36, 3]\!]$ & 0.310 & $[\![1586, 36, 7\rightarrow10]\!]$ & 0.023 & $[\![666, 36, 4\rightarrow8]\!]$ & 0.054\\ \hline
    \end{tabular}
    \caption{Classical weight reduction applied to hypergraph product codes with the best known linear codes of small size as input. $\mathcal{C}(H)$ is the linear code with parity-check matrix $H$ obtained from GAP (see main text). For each hypergraph product code, we give its encoding rate $R=k/n$. For each weight-reduction method, we apply the relevant algorithm 10,000 times using different permutations of the input parity-check matrix. In cases where permutations improved the distance, we use the notation $d_1 \rightarrow d_2$, where $d_1$ indicates the distance without permutations, and $d_2$ indicates the highest obtained distance.}
    \label{tab:hgp-examples}
\end{table}
\end{center}
}

In this section, we present examples of codes constructed using classical and quantum weight reduction. We focus on two classes of QEC codes: hypergraph product codes and lifted product codes (see \cref{sec:background} for the relevant background) for two reasons. First, both classical and quantum weight reduction are applicable to these code classes, so we can use them to compare the two weight reduction methods directly. Second, both of these classes contain code families with constant encoding rate~\cite{tillich2014quantum, panteleev2021quantum}, and almost linear minimum distance in the case of lifted product codes~\cite{panteleev2021quantum}, and therefore present compelling alternatives to code families with limited parameters such as two-dimensional topological codes~\cite{bravyi2010tradeoffs}.

We first consider hypergraph product codes. Recall that the distance of a hypergraph product code $\mathrm{HGP}(H_1,H_2)$ is $d = \min (d_1, d_2, d_1^{\mathrm{T}}, d_2^{\mathrm{T}})$. If $H_i$ is full rank, then $\mathcal{C}\left(H_i^{\mathrm{T}}\right)$ has $k=0$. In this case, $d_i^{\mathrm{T}}$ is defined to be infinite. By \cref{thm:classicaldistthm}, weight reduction does not decrease the code distance of the (classical) input codes and therefore we can conclude that $\tilde{d} \geq d$, where $\tilde{d}$ is the distance of $\mathrm{HGP}(\tilde{H}_1,\tilde{H}_2)$. We specialize to the case of square hypergraph product codes, $\mathrm{HGP}(H) = \mathrm{HGP}(H,H)$, where $H$ is the parity-check matrix of a linear code. As above, we use $\tilde{H}$ to denote the matrix produced from $H$ by classical weight reduction and introduce $\tilde{H}^{(c)}$ to denote the matrix produced from $H$ by compressed classical weight reduction. We use $\widetilde{\mathrm{HGP}}(H)$ to denote the code produced from $\mathrm{HGP}(H)$ by applying quantum weight reduction.

To compare the weight reduction methods, we consider small parity-check matrices returned by the GAP function \texttt{BestKnownLinearCode} (from the package GUAVA~\cite{guava}) as inputs to the hypergraph product.  The output matrices of classical weight reduction have row and column weights upper bounded by three, and by using these as inputs we are able to construct weight-reduced hypergraph product codes with $(\tilde{w}_X, \tilde{q}_X, \tilde{w}_Z, \tilde{q}_Z) = (6, 3, 6, 3)$. The weight-reduced codes produced here often have higher distance than their inputs, although at a reduced encoding rate; see \cref{tab:hgp-examples} (and \cref{tab:hgp-examples-app-A,tab:hgp-examples-app-B} in \cref{app:examples-hgp}). We find that compressed classical weight reduction gives us codes with improved encoding rate but often worse distance, as is to be expected from the analysis in \cref{sec:classicalWtR}. We additionally find that independent row and column permutations (see \cref{subsec:perm}) can dramatically improve the distance of the weight-reduced codes. We demonstrate this in \cref{tab:hgp-examples} for the hypergraph product codes where the number to the left of an arrow is the distance in the unpermuted case and the number to the right is the highest distance found amongst $10,000$ codes constructed with independent row/column permutation. We apply the same methodology to the family of $(w_X, q_X, w_Z, q_Z) = (7, 4, 7, 4)$ hypergraph product codes from~\cite[Table 1]{roffe2020decoding} in \cref{tab:hgp-joschka}.

For quantum weight reduction we achieve parameters of $(\tilde{w}_X, \tilde{q}_X, \tilde{w}_Z, \tilde{q}_Z) = (6, 6, 6, 3)$ but at the cost of a prohibitively large increase in the number of physical qubits; see \cref{tab:qwr-examples}. For example, consider the code $\mathrm{HGP}(H)$ where $H$ is the parity-check matrix of the $[6, 3, 3]$ code from \cref{tab:hgp-examples}. This is a $[\![45, 9, 3]\!]$ code with $(w_X, q_X, w_Z, q_Z) = (7, 4, 7, 4)$. The best code we found by applying quantum weight reduction after approximately $100$ random cycle bases has parameters $[\![2892, 9, 5]\!]$ and $(\tilde{w}_X, \tilde{q}_X, \tilde{w}_Z, \tilde{q}_Z) = (6, 6, 6, 3)$. We did not utilize a second round of thickening and choosing heights to reduce $\tilde{q}_X$, as the increase in $n$ was already excessively large. Compare this with the code with parameters $[\![117, 9, 4]\!]$ and $(\tilde{w}_X, \tilde{q}_X, \tilde{w}_Z, \tilde{q}_Z) = (6, 3, 6, 3)$ produced by classical weight reduction. These disparities combined with the expected degradation of iterative decoding performance in the case of quantum weight reduction (see \cref{sec:itdec}) lead us to conclude that classical weight reduction is a superior technique for weight reducing hypergraph product codes in the regime of interest.

{\renewcommand{\arraystretch}{1.4}
\begin{center}
\begin{table}[]
    \begin{tabular}{||c|c|c|c|c|c|c||}
        \hline
        $\mathcal{C}(H)$ & $\mathrm{HGP}(H)$ & $R$ & $\mathrm{HGP}(\tilde{H})$ & $R$ & $\mathrm{HGP}(\tilde{H}^{(c)})$ & $R$ \\ \hline
        $[16, 4, 6]$ & $[\![400, 16, 6]\!]$ & 0.040 & $[\![5008, 16, 18 \rightarrow 23]\!]$ & 0.003 & $[\![1360, 16, 10 \rightarrow 13]\!]$ & 0.012\\ \hline
        $[20, 5, 8]$ & $[\![625, 25, 8]\!]$ & 0.040 & $[\![7825, 25, 23 \rightarrow 29]\!]$ & 0.003 & $[\![2125, 25, 12 \rightarrow 15]\!]$ & 0.012\\ \hline
        $[24, 6, 10]$ & $[\![900, 36, 10]\!]$ & 0.040 & $[\![11268, 36, 29 \rightarrow 33]\!]$ & 0.003 & $[\![3060, 36, 17 \rightarrow 18]\!]$ & 0.012\\ \hline
    \end{tabular}
    \caption{Classical weight reduction applied to the family of hypergraph product codes from~\cite[Table 1]{roffe2020decoding}. For each hypergraph product code, we give its encoding rate $R = k / n$. For each weight-reduction method, we apply the relevant algorithm 10,000 times using different permutations of the input parity-check matrix. We use the notation $d_1 \rightarrow d_2$, where $d_1$ indicates the distance without permutations, and $d_2$ indicates the highest obtained distance.}
    \label{tab:hgp-joschka}
\end{table}
\end{center}
}

{\renewcommand{\arraystretch}{1.4}
\begin{center}
\begin{table}[h]
    \begin{tabular}{||c|c|c|c|c|c|c||}
        \hline
        $\mathcal{C}(H)$ & $\mathrm{HGP}(H)$ & $R$ & $(w_X,q_X,w_Z,q_Z)$ & $\widetilde{\mathrm{HGP}}(H)$ & $R$ & $(\tilde{w}_X, \tilde{q}_X, \tilde{w}_Z, \tilde{q}_Z)$\\ \hline
        $[6, 3, 3]$ & $[\![45, 9, 3]\!]$ & 0.200 & (7, 4, 7, 4) & $[\![2892, 9, 5]\!]$ & 0.003 & (6, 6, 6, 3)\\ \hline
        $[7, 2, 4]$ & $[\![74, 4, 4]\!]$ & 0.054 & (7, 4, 7, 4) & $[\![7466 ,4, 6]\!]$ & 0.0005 & (6, 6, 8, 3)\\ \hline
        $[7, 3, 4]$ & $[\![65, 9, 4]\!]$ & 0.138 & (7, 4, 7, 4) & $[\![6844, 9, 5]\!]$ & 0.001 & (6, 6, 8, 3)\\ \hline
        $[7, 4, 3]$ & $[\![58, 16, 3]\!]$ & 0.276 & (7, 4, 7, 4) & $[\![5085, 16, 3]\!]$ & 0.003 & (6, 6, 8, 3)\\ \hline
    \end{tabular}
\caption{
Quantum weight reduction applied to small hypergraph product codes. For each code we ran the algorithm approximately $100$ times with random cellulations and kept the output with the lowest $(\tilde{w}_X, \tilde{q}_X, \tilde{w}_Z, \tilde{q}_Z)$.
}
\label{tab:qwr-examples}
\end{table}
\end{center}
}

We also consider the more efficient lifted product construction. Classical weight reduction is applicable to quasi-cyclic lifted product codes with various inputs including group algebras and polynomial rings, e.g.,
\begin{equation}\label{QCwtred}
    \begin{pNiceArray}{cccccc}
        g_1(x) & \hdots & g_w(x) & 0 & \cdots & 0
    \end{pNiceArray}
    \mapsto
    \begin{pNiceArray}{ccccccccccc}
        g_1(x) &   &        &   &   & \cdots & 1 &   &   &   & \\
          & g_2(x) &        &   &   & \cdots & 1 & 1 &   &   & \\
          &   & \ddots &   &   & \cdots &   & & \ddots &   & \\
          &   &        & g_{w - 1}(x) &   & \cdots &   &   &   & 1 & 1\\
          &   &        &   & g_w(x) & \cdots &   &   &   &   & 1
    \end{pNiceArray},
\end{equation}
where $1$ is the constant polynomial. Although this reduces the row weight, the final weight is still determined by the number of coefficients in each polynomial and could be larger than our target of three. 

We restrict our attention to codes of the form $\mathrm{LP}(A) = \mathrm{LP}(A, A^{\mathrm{T}})$ and use base matrices that have previously appeared in the literature, which are given in \cref{app:examples-lp}. All of our weight-reduced lifted codes have $(\tilde{w}_X, \tilde{q}_X, \tilde{w}_Z, \tilde{q}_Z) = (6, 3, 6, 3)$.  As with hypergraph product codes, we find that the weight-reduced codes have higher distance but lower encoding rate than the input codes; see \cref{tab:lp-examples}. We again observe that independent row and column permutations improve the distances of the weight-reduced codes, though we only construct ten codes for each base matrix as the distance calculations are more time-consuming in this case. The upper bounds on the distances are computed using the BP+OSD method described in~\cite{bravyi2023high}. Note that the bounds for the codes produced with (uncompressed) classical weight reduction are likely loose.

Comparing our weight-reduced lifted product and hypergraph product codes, we observe that both families have similar encoding rates but the lifted product codes generally offer improved distances. To illustrate this, consider the quantity $Rd^2$, which is equal to one for the rotated surface code~\cite{bombin2007optimal}. For the hypergraph product codes, the highest value we obtain is $Rd^2 = 6.4$ for the $[\![1850, 100, 9]\!]$ code in \cref{tab:hgp-examples-app-B}. In contrast, for the lifted product codes produced using compressed classical weight reduction (where we are have more confidence in the distance estimates), we obtain values up to $Rd^2 = 37.6$ for the $[\![2635, 43, \leq 48]\!]$ code in \cref{tab:lp-examples}.

{\renewcommand{\arraystretch}{1.4}
\begin{center}
\begin{table}[ht!]
    \begin{tabular}{||c|c|c|c|c|c|c|c|c||}
        \hline
        $\mathcal{C}(A)$ & $\mathrm{LP}(A)$ & $R$ & $\mathcal{C}(\tilde{A})$ & $\mathrm{LP}(\tilde{A})$ & $R$ & $\mathcal{C}(\tilde{A}^{(c)})$ & $\mathrm{LP}(\tilde{A}^{(c)})$ & $R$ \\
        \hline
        $[52,27,6]$ & $[[260,58,\leq 6]]$ & 0.223 & $[130,27,12\rightarrow 14]$ & $[[2132,58,\leq 14]]$ & 0.027 & $[78,27,6\rightarrow8]$ & $[[676,58,\leq 8]]$ & 0.086\\ \hline
        $[28,9,10]$ & $[[175,19,\leq 10]]$ & 0.109 & $[91,9,28\rightarrow33]$ & $[[2191,19,\leq 39]]$ & 0.009 & $[49,9,14\rightarrow18]$ & $[[595,19,\leq 18]]$ & 0.032\\ \hline
        $[36,11,12]$ & $[[225,21,\leq 12]]$ & 0.093 & $[117,11,36\rightarrow40]$ & $[[2817,21,\leq 48]]$ & 0.007 & $[63,11,18\rightarrow22]$ & $[[765,21,\leq 22]]$ & 0.027\\ \hline
        $[68,19,18]$ & $[[425,29,\leq 18]]$ & 0.068 & $[221,19,54\rightarrow62]$ & $[[5321,29,\leq 74]]$ & 0.005 & $[119,19,32\rightarrow34]$ & $[[1445,29,\leq 34]]$ & 0.020\\ \hline
        $[76,21,20]$ & $[[475,31,\leq 20]]$ & 0.065 & $[247,21,60\rightarrow68]$ & $[[5947,31,\leq 93]]$ & 0.005 & $[133,21,37\rightarrow38]$ & $[[1615,31,\leq 38]]$ & 0.019\\ \hline
        $[124,33,24]$ & $[[775,43,\leq 24]]$ & 0.055 & $[403,33,71\rightarrow84]$ & $[[9703,43,\leq 115]]$ & 0.004 & $[217,33,44\rightarrow48]$ & $[[2635,43,\leq 48]]$ & 0.016\\ \hline
    \end{tabular}
    \caption{Classical weight reduction applied to lifted product codes. The base matrices $A$ are given in \cref{app:examples-lp} and $\mathcal{C}(A)$ is the quasi-cyclic code defined by $A$. For each weight reduced base matrix, we apply the relevant algorithm 10 times with different permutations of the input and retain the matrix whose associated quasi-cyclic code has the highest minimum distance. For each lifted product code we also provide the encoding rate $R = k / n$. We use the notation $d_1 \rightarrow d_2$, where $d_1$ indicates the distance without permutations, and $d_2$ indicates the highest obtained distance.}
    \label{tab:lp-examples}
\end{table}
\end{center}
}

\subsection{Numerical Simulations}
We simulate the performance of our codes as a quantum memory using a noise model derived from Xanadu's photonic architecture based on GKP qubits~\cite{bourassa2021Blueprint, tzitrin2021Fault}. Specifically, we consider a cluster state formed by foliating a QEC code~\cite{bolt2016Foliation, roberts2020Universal}, which is realized using GKP qubits and passive linear optics. Concatenation of the GKP code with discrete-variable QEC codes has previously been considered in~\cite{fukui2018high, vuillot2019quantum, noh2020fault, hanggli2020enhanced, zhang2021quantum, noh2022low, raveendran2022finite, zhang2023concatenation}. For a QEC code with weights $(w_X, q_X, w_Z, q_Z)$, the data and ancilla nodes of the corresponding cluster state have degrees $q_X + 2$ and $w_X$, respectively, in primal layers and $q_Z + 2$ and $w_Z$ in dual layers. The foliated stabilizers have weight $w_X + 2$ and $w_Z + 2$, although we emphasize that these are reconstructed from single-qubit measurement outcomes rather than being measured directly. We consider $2d$ foliation layers, where $d$ is the distance of the QEC code.

Following Ref.~\cite{tzitrin2021Fault}, once the cluster state is specified, the resource state construction involves:
\begin{enumerate}
    \item \textbf{Preparing GKP two-qubit cluster states.} These can be constructed using two GKP sensor states, a 50:50 beamsplitter, and a $\pi/2$ phase shifter~\cite{walshe2020continuous, tzitrin2021Fault}. Prior to the beamsplitter, each mode is assumed to be a perfect GKP qubit up to a single-mode Gaussian blurring channel parameterized by a variance $\sigma^{2}$~\cite{tzitrin2021Fault}. For Gaussian blurring channels, it is convenient to express the variance of the Gaussian in terms of decibels, $\tfrac{\sigma^{2}}{\sigma_{\text{vac}}^{2}}[\mathrm{dB}] = -10 \log_{10} (\sigma^{2}/\sigma^{2}_{\text{vac}})$, where $\sigma^{2}_{\text{vac}} $ is the variance of the vacuum, which is $1/2$ in units where $\hbar =1$. This model is equivalent to uniform photon loss experienced throughout the cluster state generation and measurement process~\cite{tzitrin2021Fault}.
    \item \textbf{Placing GKP pairs.} Place one GKP cluster state   pair for each edge of the cluster state graph. Now each original node in the cluster state is associated with one mode (half of a pair) per neighbor. This collection of modes is referred to as a macronode. 
    \item \textbf{Measuring macronodes.} To measure a given cluster state site in the $X$ basis, a continuous-variable GHZ measurement is applied on each mode within the corresponding macronode. This can be achieved by sending each mode through a beamsplitter network, followed by measurement of momentum on a single mode, and position on the rest. To measure in the $Z$ basis the process is the same except that all modes are measured in the position basis. 
    The Pauli error rate of these measurements increases monotonically with variance $\sigma^2$, which corresponds to decreasing $\tfrac{\sigma^{2}}{\sigma_{\text{vac}}^{2}}[\mathrm{dB}]$.
    \item \textbf{Applying feedforward corrections.} As described in Ref.~\cite{tzitrin2021Fault}, feedforward corrective displacement operators that are conditioned on binned homodyne outcomes from a given macronode must be applied on all modes residing in macronodes that are neighbors with respect to the original cluster state graph. These corrections can be applied in post-processing (no physical displacement gates are required).  
\end{enumerate}

We utilize the binning strategy described in~\cite{tzitrin2021Fault} as a soft-in-soft-out inner decoder and BP+OSD as a soft-in-hard-out outer decoder~\cite{panteleev2021quantum, roffe2020decoding}. We use the min-sum variant of BP with $N/10$ iterations (where $N$ is the number of qubits in the foliated cluster state) and a flooding schedule. For OSD, we choose the combination sweep strategy with search depth parameter $\lambda = 60$. We did not attempt to optimize the decoder to take the structure of our codes into account, rather we used it because of its applicability across the class of quantum LDPC codes~\cite{roffe2020decoding}.

We quantify the performance of our foliated codes using the logical error rate, which we estimate using Monte Carlo simulations. We consider a logical error to have occurred if any logical qubit has an error after decoding. For $n_{\mathrm{tot}}$ trials and $n_{\mathrm{fail}}$ logical errors we compute the logical error rate according to 
\begin{equation}
    p_{\mathrm{fail}} = \frac{n_{\mathrm{fail}} + \kappa^2 / 2}{n_{\mathrm{tot}} + \kappa^2},
\end{equation}
with error bars given by
\begin{equation}
    \kappa \sqrt{\frac{p_{\mathrm{fail}} (1 - p_{\mathrm{fail}})}{n_{\mathrm{tot}} + \kappa^2}},
\end{equation}
where $\kappa$ is the desired quantile of a standard normal distribution~\cite{brown2001interval}. We use $\kappa = 1.96$ which corresponds to a $95\%$ confidence interval.

\begin{figure}
    \centering
    \subfloat[]{
    \begin{tikzpicture}
        % Define axes
        \begin{axis}[
            xlabel={$\sigma^2/\sigma_{\mathrm{vac}}^2$ [dB]}, % X-axis label with LaTeX expression
            ylabel={$p_{\mathrm{fail}}$}, % Y-axis label with LaTeX expression
            axis lines=box,
            ymode=log, % Set y-axis to log scale
            % ymin=0.1, ymax=100, % Set appropriate y-axis limits for log scale
            grid=both, % Add grids to both major and minor ticks
            legend style={at={(0.03,0.02)},anchor=south west}, % Place legend inside the plot
            cycle list/Dark2,
            legend cell align={left},
            legend style={nodes={scale=0.9, transform shape}}
            ]
        \addlegendimage{mark=square*,only marks,mark size=2,Dark2-A}
        \addlegendimage{mark=*,only marks,mark size=2,Dark2-B}
        \addlegendimage{mark=triangle*,only marks,mark size=2,Dark2-C}
        \addlegendimage{mark=diamond*,only marks,mark size=2,Dark2-D}
        
        % Plot data series 1 with error bars
        \addplot+[mark=square*, error bars/.cd,y dir=both, y explicit] table[x=delta, y=p_logical, y error=p_logical_error_bar, col sep=comma] {data/results_HGP_1.csv};
        \addlegendentry{$\mathrm{HGP}(H) : [\![45,9,3]\!]$}
        
        % Plot data series 2 with error bars
        \addplot+[mark=*, error bars/.cd,y dir=both, y explicit] table[x=delta, y=p_logical, y error=p_logical_error_bar, col sep=comma] {data/results_HGP_3.csv};
        \addlegendentry{$\mathrm{HGP}(\tilde{H}) : [\![117,9,4]\!]$}
        
        % Plot data series 3 with error bars
        \addplot+[mark=triangle*, error bars/.cd,y dir=both, y explicit] table[x=delta, y=p_logical, y error=p_logical_error_bar, col sep=comma] {data/results_HGP_2.csv};
        \addlegendentry{$\mathrm{HGP}(\tilde{H}^{(c)}) : [\![65,9,4]\!]$}
        
        % Plot data series 4 with error bars
        \addplot+[mark=diamond*, error bars/.cd,y dir=both, y explicit] table[x=delta, y=p_logical, y error=p_logical_error_bar, col sep=comma] {data/results_HWR_1.csv};
        \addlegendentry{$\widetilde{\mathrm{HGP}}(H) : [\![2892,9,5]\!]$}
        
        \end{axis}
    \end{tikzpicture}\label{fig:633-sim}}
    \quad
    \subfloat[]{
    \begin{tikzpicture}
        % Define axes
        \begin{axis}[
            xlabel={$\sigma^2/\sigma_{\mathrm{vac}}^2$ [dB]}, % X-axis label with LaTeX expression
            ylabel={$p_{\mathrm{fail}}$}, % Y-axis label with LaTeX expression
            axis lines=box,
            ymode=log, % Set y-axis to log scale
            % ymin=0.1, ymax=100, % Set appropriate y-axis limits for log scale
            grid=both, % Add grids to both major and minor ticks
            legend style={at={(0.03,0.02)},anchor=south west}, % Place legend inside the plot
            cycle list/Dark2,
            legend cell align={left},
            legend style={nodes={scale=0.9, transform shape}}
            ]
        \addlegendimage{mark=square*,only marks,mark size=2,Dark2-A}
        \addlegendimage{mark=*,only marks,mark size=2,Dark2-B}
        \addlegendimage{mark=triangle*,only marks,mark size=2,Dark2-C}
        \addlegendimage{mark=diamond*,only marks,mark size=2,Dark2-D}
        
        % Plot data series 1 with error bars
        \addplot+[mark=square*, error bars/.cd,y dir=both, y explicit] table[x=delta, y=p_logical, y error=p_logical_error_bar, col sep=comma] {data/results_HGP_169.csv};
        \addlegendentry{$\mathrm{HGP}(H_1) : [\![400,16,6]\!]$}
        
        % Plot data series 2 with error bars
        \addplot+[mark=*, error bars/.cd,y dir=both, y explicit] table[x=delta, y=p_logical, y error=p_logical_error_bar, col sep=comma] {data/results_HGP_170.csv};
        \addlegendentry{$\mathrm{HGP}(\tilde{H}_1^{(c)}) : [\![1360,16,13]\!]$}
        
        % Plot data series 3 with error bars
        \addplot+[mark=triangle*, error bars/.cd,y dir=both, y explicit] table[x=delta, y=p_logical, y error=p_logical_error_bar, col sep=comma] {data/results_HGP_172.csv};
        \addlegendentry{$\mathrm{HGP}(H_2) : [\![625,25,8]\!]$}
        
        % Plot data series 4 with error bars
        \addplot+[mark=diamond*, error bars/.cd,y dir=both, y explicit] table[x=delta, y=p_logical, y error=p_logical_error_bar, col sep=comma] {data/results_HGP_173.csv};
        \addlegendentry{$\mathrm{HGP}(\tilde{H}_2^{(c)}) : [\![2125,25,16]\!]$}
        
        \end{axis}
    \end{tikzpicture}\label{fig:hgp-sim}}
    \caption{Simulation results for hypergraph product codes.
    Recall that larger values of $\sigma^2/\sigma_{\mathrm{vac}}^2$ [dB] correspond to lower effective Pauli error rates. (a) The baseline $\mathrm{HGP}(H)$ where $H$ is the parity-check matrix of the $[6, 3, 3]$ code from \cref{tab:hgp-examples}. $\mathrm{HGP}(\tilde{H})$ and $\mathrm{HGP}(\tilde{H}^{(c)})$ are constructed via applying classical weight reduction to $H$. $\widetilde{\mathrm{HGP}}(H)$ is constructed by applying quantum weight reduction to $\mathrm{HGP}(H)$. $\mathrm{HGP}(\tilde{H})$ and $\mathrm{HGP}(\tilde{H}^{(c)})$ have superior performance compared to the baseline code, whereas $\widetilde{\mathrm{HGP}}(H)$ only becomes competitive for large values of $\sigma^2/\sigma_{\mathrm{vac}}^2$. (b) Comparison of the codes from rows one and two of \cref{tab:hgp-joschka}. The waterfall region for the weight-reduced codes starts at approximately $10.5$dB compared to approximately $11$dB for the original codes.}
    \label{fig:sim-hgp}
\end{figure}
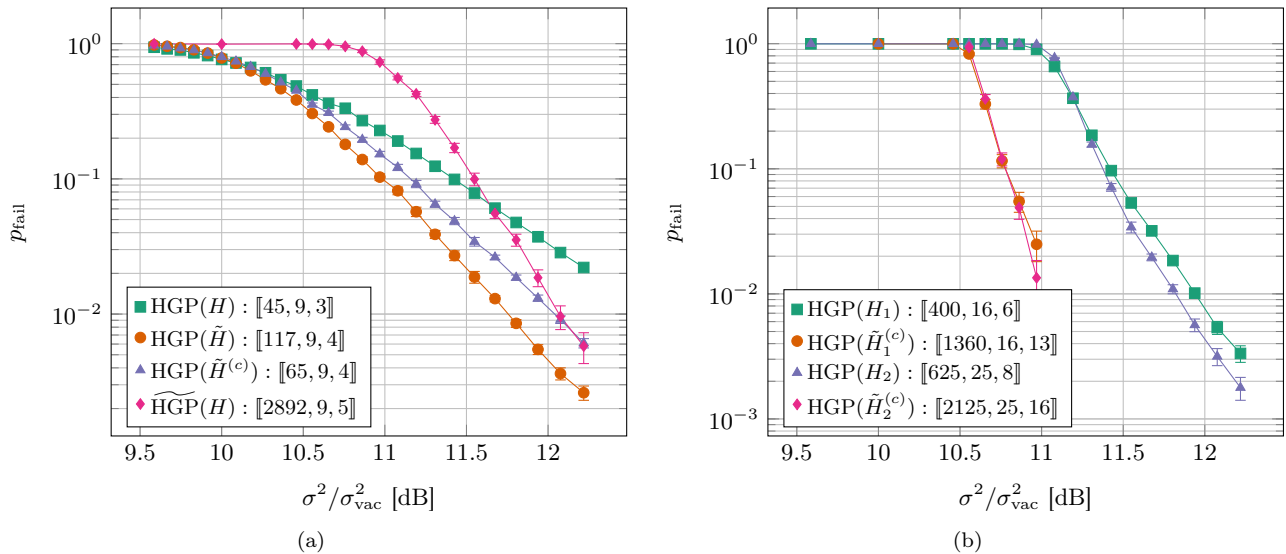

\begin{figure}[ht]
    \centering
    \subfloat[]{
    \begin{tikzpicture}
        % Define axes
        \begin{axis}[
            xlabel={$\sigma^2/\sigma_{\mathrm{vac}}^2$ [dB]}, % X-axis label with LaTeX expression
            ylabel={$p_{\mathrm{fail}}$}, % Y-axis label with LaTeX expression
            axis lines=box,
            ymode=log, % Set y-axis to log scale
            % ymin=0.1, ymax=100, % Set appropriate y-axis limits for log scale
            grid=both, % Add grids to both major and minor ticks
            legend style={at={(0.03,0.02)},anchor=south west}, % Place legend inside the plot
            cycle list/Dark2,
            legend cell align={left},
            legend style={nodes={scale=0.9, transform shape}},
            ]
        \addlegendimage{mark=square*,only marks,mark size=2,Dark2-A}
        \addlegendimage{mark=*,only marks,mark size=2,Dark2-B}
        \addlegendimage{mark=triangle*,only marks,mark size=2,Dark2-C}
        
        % Plot data series 1 with error bars
        \addplot+[mark=square*, error bars/.cd,y dir=both, y explicit] table[x=delta, y=p_logical, y error=p_logical_error_bar, col sep=comma] {data/results_LP_1.csv};
        \addlegendentry{$\mathrm{LP}(A_1) : [\![260,58,6]\!]$}
        
        % Plot data series 2 with error bars
        \addplot+[mark=*, error bars/.cd,y dir=both, y explicit] table[x=delta, y=p_logical, y error=p_logical_error_bar, col sep=comma] {data/results_LP_3.csv};
        \addlegendentry{$\mathrm{LP}(\tilde{A_1}) : [\![2132,58,14]\!]$}
        
        % Plot data series 3 with error bars
        \addplot+[mark=triangle*, error bars/.cd,y dir=both, y explicit] table[x=delta, y=p_logical, y error=p_logical_error_bar, col sep=comma] {data/results_LP_2.csv};
        \addlegendentry{$\mathrm{LP}(\tilde{A_1}^{(c)}) : [\![676,58,8]\!]$}
        
        \end{axis}
    \end{tikzpicture}}
    \quad
    \subfloat[]{
    \begin{tikzpicture}
        % Define axes
        \begin{axis}[
            xlabel={$\sigma^2/\sigma_{\mathrm{vac}}^2$ [dB]}, % X-axis label with LaTeX expression
            ylabel={$p_{\mathrm{fail}}$}, % Y-axis label with LaTeX expression
            axis lines=box,
            ymode=log, % Set y-axis to log scale
            % ymin=0.1, ymax=100, % Set appropriate y-axis limits for log scale
            grid=both, % Add grids to both major and minor ticks
            legend style={at={(0.03,0.02)},anchor=south west}, % Place legend inside the plot
            cycle list/Dark2,
            legend cell align={left},
            legend style={nodes={scale=0.9, transform shape}},
            % legend columns=2
            ]
        \addlegendimage{mark=square*,only marks,mark size=2,Dark2-A}
        \addlegendimage{mark=*,only marks,mark size=2,Dark2-B}
        \addlegendimage{mark=triangle*,only marks,mark size=2,Dark2-C}
        \addlegendimage{mark=diamond*,only marks,mark size=2,Dark2-D}
        
        % Plot data series 1 with error bars
        \addplot+[mark=square*, error bars/.cd,y dir=both, y explicit] table[x=delta, y=p_logical, y error=p_logical_error_bar, col sep=comma] {data/results_LP_13.csv};
        \addlegendentry{$\mathrm{LP}(A_2) : [\![175,19,10]\!]$}
        
        % Plot data series 2 with error bars
        \addplot+[mark=*, error bars/.cd,y dir=both, y explicit] table[x=delta, y=p_logical, y error=p_logical_error_bar, col sep=comma] {data/results_LP_14.csv};
        \addlegendentry{$\mathrm{LP}(\tilde{A}_2^{(c)}) : [\![595,19,18]\!]$}
        
        % Plot data series 3 with error bars
        \addplot+[mark=triangle*, error bars/.cd,y dir=both, y explicit] table[x=delta, y=p_logical, y error=p_logical_error_bar, col sep=comma] {data/results_LP_16.csv};
        \addlegendentry{$\mathrm{LP}(A_3) : [\![225,21,12]\!]$}
        
        % Plot data series 4 with error bars
        \addplot+[mark=diamond*, error bars/.cd,y dir=both, y explicit] table[x=delta, y=p_logical, y error=p_logical_error_bar, col sep=comma] {data/results_LP_17.csv};
        \addlegendentry{$\mathrm{LP}(\tilde{A}_3^{(c)}) : [\![765,21,22]\!]$}
        
        \end{axis}
    \end{tikzpicture}}
    \caption{Simulation results for lifted product codes. (a) We compare the code $\mathrm{LP}(A_1)$ with its weight reduced counterparts, $\mathrm{LP}(\tilde{A}_1)$ and $\mathrm{LP}(\tilde{A}_1^{(c)})$; see the first row of \cref{tab:lp-examples}. We again observe improved performance for the weight-reduced codes, with the uncompressed variant performing best. (b) Comparison of the hypergraph product codes from rows 2 and 3 of \cref{tab:hgp-joschka}. As in the hypergraph product case, the waterfall region begins at a smaller value of $\sigma^2/\sigma_{\mathrm{vac}}^2$ for the weight-reduced codes.}
    \label{fig:sim-lp}
\end{figure}
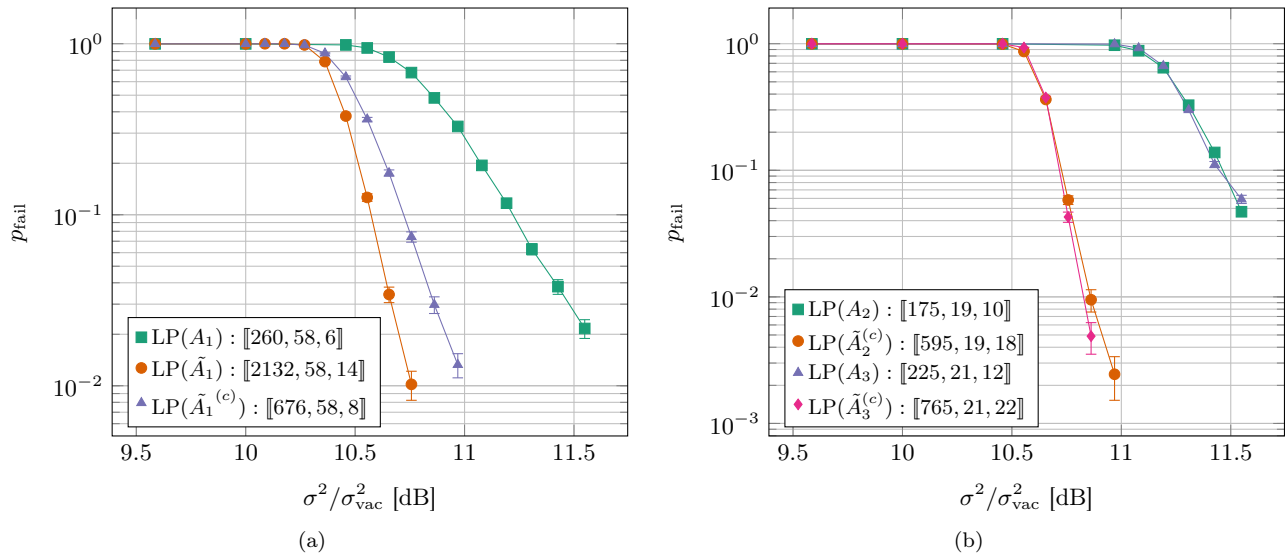

We first focus on a single hypergraph product code: the $[\![45, 9, 3]\!]$ code from \cref{tab:hgp-examples} constructed from the parity-check $H$ of a $[6, 3, 3]$ linear code. We compare the performance of $\mathrm{HGP}(H)$ with the performance of:
\begin{itemize}
    \item The $[\![2892, 9, 5]\!]$ code, $\widetilde{HGP}(H)$, formed by applying quantum weight reduction to $\mathrm{HGP}(H)$.
    \item The hypergraph product code $\mathrm{HGP}(\tilde{H})$ with parameters $[\![117, 9, 4]\!]$.
    \item The hypergraph product code $\mathrm{HGP}(\tilde{H}^{(c)})$ with parameters $[\![65, 9, 4]\!]$.
\end{itemize}
The results are shown in \cref{fig:633-sim}. We find that the codes obtained using classical weight reduction outperform the original code and the code obtained using quantum weight reduction. Even though the distance of the quantum weight-reduced code is the highest, the waterfall region (where the logical error rate starts to decrease) for this code begins at much higher values of $\sigma^2/\sigma^2_{\mathrm{vac}}$ when compared with the other codes. This is likely due to the increased row and columns weights of the parity-check matrices.

We also compare the performance of two codes from \cref{tab:hgp-joschka} with their (compressed) weight-reduced counterparts. The weight-reduced codes again have superior performance, and we note that the waterfall region for the weight-reduced codes begins at $\sigma^2/\sigma^2_{\mathrm{vac}}\approx 10.5$dB compared with $\sigma^2/\sigma^2_{\mathrm{vac}}\approx 11$dB for the baseline codes, indicating that weight reduction not only improves logical error rates but also the break-even point.

\section{Discussion \& Conclusion}\label{sec:conc}
Constructing qLDPC codes with good performance under realistic noise models is an essential step towards designing more efficient fault-tolerant quantum computing architectures. Since the noise associated with measuring a stabilizer scales with both the row and column weight of the stabilizer matrix in many proposed architectures, qLDPC codes with lower weights are likely to have superior performance. Hastings provided a method to reduce the weights of CSS codes \cite{hastings2016, hastings2021quantum}, but its use has so far been restricted to the asymptotic setting. This work provides the first application of quantum weight reduction to codes constructed with near-term hardware in mind. Our examples show the method leads to a large increase in the number of physical qubits and that it is often difficult to find cellulations that produce small weights. Additionaly, the Decongestion Lemma adds randomness to finding a cycle basis and the optimal choice of cellulation of large cycles is unclear. As a result, the majority of runs on inputs with weights already close to the theoretical minimum values produced outputs with weights with equal or worse to the initial weights. An implementation of the method is provided at \cite{Sabo_2021}.

In addition to examples, we provided an accessible review of the method with fewer mathematical prerequisites. Explicit examples accompany each step, including a completely diagrammatic description. Several statements made in \cite{hastings2021quantum} are clarified and subtleties are discussed. We also proposed modifications to reduce the overhead and analysed the method's effect on iterative decoding. Finally, we showed that three of the four steps of the method may be viewed in the same mathematical framework, without algebraic topology.

In response to some of the drawbacks of quantum weight reduction, we introduced classical and compressed classical weight reduction. Applying this technique to the inputs of the hypergraph product (for example) guarantees row weights of at most six and column weights at most three. We showed that classical weight reduction can increase the distance of the classical code and hence the distance for the hypergraph product code. The compressed variant reduces the overhead while at least maintaining the distance of the input. The two approaches represent a trade-off between achieving the maximum reduction in overhead versus the maximum increase in distance. We also showed that permutations in weight reduction can increase both the distance and the girth of the Tanner graph. We emphasize that classical weight reduction is applicable to quasi-cyclic codes defined by matrices with entries in a polynomial quotient ring, allowing us to construct examples of  weight-reduced lifted product codes with superior parameters to our hypergraph product code examples.

Both the quantum and classical weight reductions have the same input and output code dimensions, but the classical method uses far fewer qubits than the quantum method. An analysis of the cycle structure of both approaches shows that quantum weight reduction may strongly degrade the performance of iterative decoders, while the classical case may actually improve it. We benchmarked the performance of our weight-reduced codes in a photonic quantum computing architecture based on GKP qubits and passive linear optics. We used Monte Carlo simulations to estimate the logical error rate of the foliated cluster state corresponding to a logical identity channel (quantum memory). In every case we simulated, we observed improved performance for the codes produced using classical weight reduction.

There is another approach to weight reduction that we have so far neglected, where the check weights of a QEC code are reduced by transforming the QEC code into a subsystem code~\cite{kribs2005unified, poulin2005stabilizer,bacon2006operator} whose gauge operators are lower weight. This method has been successfully applied to topological codes defined on Euclidean and hyperbolic tilings~\cite{bombin2010topological,suchara2011constructions,bravyi2013subsystem,bombin2015gauge,bravyi2015doubled,oconnor2016stacked,jones2016gauge,higgott2021subsystem,kubica2022single}. A general method was proposed for qLPDC codes in Ref.~\cite{higgott2021subsystem}, however the Tanner graph of the code must obey certain conditions and the method can dramatically reduce the code distance. We leave it to future work to compare this method with the results of \cref{sec:num}. 

In future work, we plan to investigate whether we can improve the performance of our codes by using tailored decoders take advantage of the structure introduced by weight reduction. We note that one could already use our codes as a quantum memory in an architecture with a high-rate qLDPC memory and a surface code processor, as has been proposed recently for superconducting qubits~\cite{bravyi2023high} and neutral atom arrays~\cite{xu2023constant}. However, we also plan to explore techniques for performing logical operations (e.g.~\cite{krishna2021fault, cohen2022, quintavalle2023, breuckmann2022fold, huang2022homomorphic}) on our codes directly, as this could enable a fully qLDPC architecture with large savings in the number of physical qubits required to execute useful quantum algorithms at scale.

Our results illustrate that weight reduction techniques can be a useful tool for constructing useful qLDPC codes by transforming codes with good parameters but relatively high-weight checks into codes with low-weight checks, comparable parameters, and improved performance. It may be preferable to construct stabilizer codes with very low row/column weights and high encoding rate and distance directly, but the lack of known code constructions with these properties suggest that this is a challenging task. We found that optimizing the classical inputs to quantum product constructions is a superior strategy to optimizing the output quantum code itself, which is perhaps not surprising given that the orthogonality constraints that restrict quantum codes do not apply in the classical case. Here, we only scratched the surface of the space of possible quantum product codes that can be constructed with carefully designed classical inputs, and we suspect that exceptional examples are waiting to be discovered. 

% \subsubsection*{Statement of Contributions}
% E.S.\ and B.I.\ are responsible for Sections II and III. L.G., E.S., and B.I.\ are responsible for Section IV. M.V.\ is responsible for Section V. E.S.\ was the primary writer.

\subsection*{Acknowledgements}
We gratefully acknowledge work done by Priya Nadkarni on simulating the Xanadu architecture. We thank Rafael Alexander for feedback throughout the writing process. Computations were performed on the Niagara supercomputer at the SciNet HPC Consortium. SciNet is funded by Innovation, Science and Economic Development Canada; the Digital Research Alliance of Canada; the Ontario Research Fund: Research Excellence; and the University of Toronto. Research at Perimeter Institute is supported in part by the Government of Canada through the Department of Innovation, Science and Economic Development Canada and by the Province of Ontario through the Ministry of Colleges and Universities.

\appendix
\section{Review Of Homological Algebra}\label{app:homological}
\subsection{The Tensor Product Of Chain Complexes} \label{app:tenprodch}
Let
\begin{align*}
    &\mathcal{A} \,: \, \cdots \xrightarrow{} A_{i + 1} \xrightarrow{\partial^\mathcal{A}_{i + 1}} A_i \xrightarrow{\partial^\mathcal{A}_i} A_{i - 1} \xrightarrow{} \cdots\\
    &\mathcal{B} \,: \, \cdots \xrightarrow{} B_{i + 1} \xrightarrow{\partial^\mathcal{B}_{i + 1}} B_i \xrightarrow{\partial^\mathcal{B}_i} B_{i - 1} \xrightarrow{} \cdots
\end{align*}
be chain complexes with vector spaces over the same field. Then $\mathcal{A} \otimes \mathcal{B}$ is defined to have vector spaces $\displaystyle (\mathcal{A} \otimes \mathcal{B})_n = \bigoplus_{i + j = n} A_i \otimes B_j$ and maps
\begin{equation}\label{totalcomplexboundary}
    \partial^{\mathcal{A} \otimes \mathcal{B}}_{n} = \bigoplus_{i + j = n - 1} \partial^\mathcal{A}_{i + 1} \otimes I_{B_{j}} + I_{A_{i}} \otimes \partial^\mathcal{B}_{j + 1},
\end{equation}
where $I$ is the identity map on the appropriate space.

Concrete examples relevant to this work are the product of a 3-term and a 2-term chain complex and the product of two 3-term chain complexes. For the former, let
\begin{align*}
    &\mathcal{A} \,: \,A_2 \xrightarrow{\partial^\mathcal{A}_2} A_1 \xrightarrow{\partial^\mathcal{A}_1} A_0\\
    &\mathcal{B} \,: \, B_1 \xrightarrow{\partial^\mathcal{B}_1} B_0
\end{align*} 
The product $\mathcal{A} \otimes \mathcal{B}$ is often drawn as
\begin{equation}\label{doublecomplex2-1ph}
    \begin{tikzcd}
        {A_2 \otimes B_1} && {A_2 \otimes B_0} && {A_1 \otimes B_0} && {A_0 \otimes B_0} \\
        && {A_1 \otimes B_1} && {A_0 \otimes B_1}
        \arrow["{\partial^\mathcal{A}_2 \otimes I_{B_1}}"', from=1-1, to=2-3]
        \arrow["{I_{A_2} \otimes \partial^\mathcal{B}_1}", from=1-1, to=1-3]
        \arrow["{\partial^\mathcal{A}_2 \otimes I_{B_0}}", from=1-3, to=1-5]
        \arrow["{\partial^\mathcal{A}_1 \otimes I_{B_1}}"', from=2-3, to=2-5]
        \arrow["{I_{A_0} \otimes \partial^\mathcal{B}_1}"', from=2-5, to=1-7]
        \arrow["{\partial^\mathcal{A}_1 \otimes I_{B_0}}", from=1-5, to=1-7]
        \arrow["{I_{A_1} \otimes \partial^\mathcal{B}_1}"{description}, from=2-3, to=1-5]
    \end{tikzcd}
\end{equation}
or
\begin{equation}\label{doublecomplex2-1}
    \begin{tikzcd}
	   {A_2 \otimes B_1} && {A_1 \otimes B_1} && {A_0 \otimes B_1} \\
	\\
	   {A_2 \otimes B_0} && {A_1 \otimes B_0} && {A_0 \otimes B_0}
	   \arrow["{\partial^\mathcal{A}_2 \otimes I_{B_1}}"', from=1-1, to=1-3]
	   \arrow["{I_{A_2} \otimes \partial^\mathcal{B}_1}"{description}, from=1-1, to=3-1]
	   \arrow["{\partial^\mathcal{A}_2 \otimes I_{B_0}}", from=3-1, to=3-3]
	   \arrow["{\partial^\mathcal{A}_1 \otimes I_{B_1}}"', from=1-3, to=1-5]
	   \arrow["{I_{A_0} \otimes \partial^\mathcal{B}_1}"'{description}, from=1-5, to=3-5]
	   \arrow["{\partial^\mathcal{A}_1 \otimes I_{B_0}}", from=3-3, to=3-5]
	   \arrow["{I_{A_1} \otimes \partial^\mathcal{B}_1}"{description}, from=1-3, to=3-3]
    \end{tikzcd}.
\end{equation}
Using the second diagram, we may define vertical and horizontal maps, $\partial^v_i$ and $\partial^h_i$, respectively such that $\partial^v_i \circ \partial^v_{i + 1}= \partial^h_i \circ \partial^h_{i + 1} = 0$ and $\partial^v_i$ and $\partial^h_i$ commute. Then every purely vertical chain and every purely horizontal chain form a valid chain complex.

Diagram \eqref{doublecomplex2-1} is called a double complex. \Cref{totalcomplexboundary} describes the chain complex formed from collapsing the double complex into the 4-term sequence
\begin{equation}\label{3by2chain}
    \mathcal{C} = \mathcal{A} \otimes \mathcal{B} \, : \, C_3 \xrightarrow{\partial^\mathcal{C}_3} C_2 \xrightarrow{\partial^\mathcal{C}_2} C_1 \xrightarrow{\partial^\mathcal{C}_1} C_0
\end{equation}
with
\begin{align*}
    C_3 &= A_2 \otimes B_1\\
    C_2 &= (A_2 \otimes B_0) \oplus (A_1 \otimes B_1)\\
    C_1 &= (A_1 \otimes B_0) \oplus (A_0 \otimes B_1)\\
    C_0 &= A_0 \otimes B_0.
\end{align*}
Notice that these can be read off of \cref{doublecomplex2-1ph} by collapsing each vertically aligned piece with a direct sum or by taking diagonal lines through \cref{doublecomplex2-1}. This is called the total complex.

The maps of the total complex (\cref{totalcomplexboundary}) are
\begin{align*}
    \partial^\mathcal{C}_3 &= \left(I_{A_2} \otimes \partial^\mathcal{B}_1\right) \oplus \left(\partial^\mathcal{A}_2 \otimes I_{B_1}\right)\\
    \partial^\mathcal{C}_2 &= \left(\partial^\mathcal{A}_2 \otimes I_{B_0} + I_{A_1} \otimes \partial^\mathcal{B}_1\right) \oplus \left(\partial^\mathcal{A}_1 \otimes I_{B_1}\right)\\
    \partial^\mathcal{C}_1 &= \partial^\mathcal{A}_1 \otimes I_{B_0} + I_{A_0} \otimes \partial^\mathcal{B}_1.
\end{align*}
To derive explicit matrix representations of these maps, ignore the $+$ and $\oplus$ and start from first principles:  $\partial^\mathcal{C}_3$ takes basis vectors of $C_3$ to basis vectors of $C_2$. We can group the basis elements of $C_2$ by those spanning the space $A_2 \otimes B_0$ then those spanning $A_1 \otimes B_1$. Assuming we want $\partial^\mathcal{C}_3$ to act by left multiplication, a matrix representation can be organized as
\begin{equation}\label{thirdmap}
    \mathrm{Mat}\left(\partial^\mathcal{C}_3\right) = \, \begin{pNiceArray}{c}[first-row, first-col]
        & A_2 \otimes B_1\\
        A_2 \otimes B_0 & I_{A_2} \otimes \partial^\mathcal{B}_1\\
        A_1 \otimes B_1 & \partial^\mathcal{A}_2 \otimes I_{B_1}
    \end{pNiceArray},
\end{equation}
where the row and column labels are added for convenience. Similarly,\footnote{We can check that the boundary of the boundary is zero:
\begin{align*}
    \mathrm{Mat}\left(\partial^\mathcal{C}_2\right) \mathrm{Mat}\left(\partial^\mathcal{C}_3\right) &= \begin{pmatrix}
        \partial^\mathcal{A}_2 \otimes \partial^\mathcal{B}_1 + \partial^\mathcal{A}_2 \otimes \partial^\mathcal{B}_1\\
        \partial^\mathcal{A}_2 \partial^\mathcal{A}_1 \otimes I_{B_1}
    \end{pmatrix} = 0\\
    \mathrm{Mat}\left(\partial^\mathcal{C}_1\right) \mathrm{Mat}\left(\partial^\mathcal{C}_2\right) &= \begin{pmatrix}
        \partial^\mathcal{A}_2 \partial^\mathcal{A}_1 \otimes I_{B_0} & \partial^\mathcal{A}_1 \otimes \partial^\mathcal{B}_1 + \partial^\mathcal{A}_1 \otimes \partial^\mathcal{B}_1
    \end{pmatrix} = 0,
\end{align*}
where we have used the fact that $2 \equiv 0$ in $\mathbb{F}_2$.}
\begin{align}
    \mathrm{Mat}\left(\partial^\mathcal{C}_2\right) &= \, \begin{pNiceArray}{cc}[first-row, first-col]
        & A_2 \otimes B_0 & A_1 \otimes B_1\\
        A_1 \otimes B_0 & \partial^\mathcal{A}_2 \otimes I_{B_0} & I_{A_1} \otimes \partial^\mathcal{B}_1\\
        A_0 \otimes B_1 & 0 & \partial^\mathcal{A}_1 \otimes I_{B_1}
    \end{pNiceArray}\label{3by2chainmatsa}\\
    & \nonumber \\
    \mathrm{Mat}\left(\partial^\mathcal{C}_1\right) &= \, \begin{pNiceArray}{cc}[first-row, first-col]
        & A_1 \otimes B_0 & A_0 \otimes B_1\\
        A_0 \otimes B_0 & \partial^\mathcal{A}_1 \otimes I_{B_0} & I_{A_0} \otimes \partial^\mathcal{B}_1
    \end{pNiceArray}.\label{3by2chainmatsb}
\end{align}
Note that the ordering of the basis vectors in the rows of $\mathrm{Mat}\left(\partial^\mathcal{C}_i\right)$ must be consistent with the ordering of the columns of $\mathrm{Mat}\left(\partial^\mathcal{C}_{i - 1}\right)$. If we were over a different field, it would have been necessary to define the maps to be
\begin{equation*}
    \partial^{\mathcal{A} \otimes \mathcal{B}}_{n} = \bigoplus_{i + j = n - 1} \partial^A_{i + 1} \otimes I_{B_{j}} + (-1)^i I_{A_{i}} \otimes \partial^B_{j + 1}
\end{equation*}
to get the necessary cancellation.

The procedure is the same for the tensor product of two 3-term chain complexes:
\begin{align*}
    &\mathcal{A} \,: \,A_2 \xrightarrow{\partial^\mathcal{A}_2} A_1 \xrightarrow{\partial^\mathcal{A}_1} A_0\\
    &\mathcal{B} \,: \, B_2 \xrightarrow{\partial^\mathcal{B}_2} B_1 \xrightarrow{\partial^\mathcal{B}_1} B_0\\
    &\mathcal{C} = \mathcal{A} \otimes \mathcal{B} \, : \, C_4 \xrightarrow{\partial^\mathcal{C}_4} C_3 \xrightarrow{\partial^\mathcal{C}_3} C_2 \xrightarrow{\partial^\mathcal{C}_2} C_1 \xrightarrow{\partial^\mathcal{C}_1} C_0
\end{align*}
with
\begin{align*}
    C_4 &= A_2 \otimes B_2\\
    C_3 &= (A_2 \otimes B_1) \oplus (A_1 \otimes B_2)\\
    C_2 &= (A_2 \otimes B_0) \oplus (A_1 \otimes B_1) \oplus (A_0 \otimes B_2)\\
    C_1 &= (A_1 \otimes B_0) \oplus (A_0 \otimes B_1)\\
    C_0 &= A_0 \otimes B_0
\end{align*}
and
\begin{align*}
    \mathrm{Mat}\left(\partial^\mathcal{C}_4\right) &= \begin{pNiceArray}{c}[first-row, first-col]
        & A_2 \otimes B_2 \\
        A_2 \otimes B_1 & I_{A_2} \otimes \partial^\mathcal{B}_2\\
        A_1 \otimes B_2 & \partial^\mathcal{A}_2 \otimes I_{B_2}
    \end{pNiceArray}\\
    & \\
    \mathrm{Mat}\left(\partial^\mathcal{C}_3\right) &= \begin{pNiceArray}{cc}[first-row, first-col]
        & A_2 \otimes B_1 & A_1 \otimes B_2\\
        A_2 \otimes B_0 & I_{A_2} \otimes \partial^\mathcal{B}_1 & 0\\
        A_1 \otimes B_1 & \partial^\mathcal{A}_1 \otimes I_{B_1} & I_{A_1} \otimes \partial^\mathcal{B}_2\\
        A_0 \otimes B_2 & 0 & \partial^\mathcal{A}_1 \otimes I_{A_2}
    \end{pNiceArray}\\
    & \\
    \mathrm{Mat}\left(\partial^\mathcal{C}_2\right) &= \begin{pNiceArray}{ccc}[first-row, first-col]
        & A_2 \otimes B_0 & A_1 \otimes B_1 & A_0 \otimes B_2\\
        A_1 \otimes B_0 & \partial^\mathcal{A}_1 \otimes I_{B_0} & I_{A_1} \otimes \partial^\mathcal{B}_1 & 0\\
        A_0 \otimes B_1 & 0 & \partial^\mathcal{A}_1 \otimes I_{B_1} & I_{A_0} \otimes \partial^\mathcal{B}_2
    \end{pNiceArray}\\
    & \\
    \mathrm{Mat}\left(\partial^\mathcal{C}_1\right) &= \, \begin{pNiceArray}{cc}[first-row, first-col]
        & A_1 \otimes B_0 & A_0 \otimes B_1\\
        A_0 \otimes B_0 & \partial^\mathcal{A}_1 \otimes I_{B_0} & I_{A_0} \otimes \partial^\mathcal{B}_1
    \end{pNiceArray}.
\end{align*}

The homology of the total complex follows a similar form\footnote{Equations of this type are called K\"{u}nneth formulas.}
\begin{equation}\label{Kunneth}
    H_k(\mathcal{A} \otimes \mathcal{B}) \cong \bigoplus_{i + j = k} \left(H_i(\mathcal{A}) \otimes H_j(\mathcal{B})\right).
\end{equation}
Consider the total complex (\cref{3by2chain}) where the chain $\mathcal{A}$ is the CSS code (\cref{CSSchain}) and $\mathcal{B}$ is the classical repetition code (\cref{classicalchain}), \cref{repcode}). Take the CSS code determined by the right two maps. By \cref{Kunneth}, the $Z$ logical operators of this code are
\begin{equation*}
    H_1(\mathcal{C}) = (H_1(\mathcal{A}) \otimes H_0(\mathcal{B})) \oplus (H_0(\mathcal{A}) \otimes H_1(\mathcal{B})).
\end{equation*}
To compute the homology of $\mathcal{B}$, extend it by zero on both sides:
\begin{equation}\label{repchainextended}
    0 \to \mathbb{F}^{\ell - 1}_2 \xrightarrow{H^{\mathrm T}} \mathbb{F}^\ell_2 \to 0.
\end{equation}
Then $H_1(\mathcal{B}) = \ker H^{\mathrm T} = 0$ and $H_0(\mathcal{B}) = \mathbb{F}^\ell_2 / \mathrm{im} \, H^{\mathrm T}$ is the space of all vectors modulo even-weight vectors. This has two cosets: the coset of all even-weight vectors and the coset of all odd-weight vectors. The latter is generated by any weight-one vector. The logical operators are therefore of the form $a \otimes b$, where $a \in H_1(\mathcal{A})$ is a $Z$ logical operator of the original code and $b \in H_0(\mathcal{B})$, which has minimum weight $d_Z \cdot 1$. For the $X$ logical operators, we take the dual of \cref{repchainextended} and apply the K\"{u}nneth formula to cohomology. Now $H^1(\mathcal{B}) = \mathbb{F}^{\ell - 1}_2 / \mathrm{im} \, H = 0$ as $\mathrm{rank} H = \ell - 1$ and $H^0(\mathcal{B}) = \ker H$ is the length $\ell$ all-ones vector. Hence, $X$ logical operators are of the form $a \otimes b$, where $a \in H^1(\mathcal{A})$ is an $X$ logical operator of the original code and $b \in H^0(\mathcal{B})$. The $X$ distance therefore increases to $d_X \cdot \ell$. This technique first appeared in~\cite{hastings2016} and is called distance balancing. It was generalized to use other parity-check matrices in~\cite{evra2022decodable}; see also~\cite{cross2022quantum, wills2023general}.

The hypergraph product (\cref{hgp}) is the tensor product
\begin{align*}
    &\mathcal{A} \,: \mathbb{F}^{n_1}_2 \xrightarrow{H_1} \mathbb{F}^{m_1}_2\\
    &\mathcal{B} \,: \, \mathbb{F}^{m_2}_2 \xrightarrow{H^{\mathrm{T}}_2} \mathbb{F}^{n_2}_2\\
    &\mathcal{C} = \mathcal{A} \otimes \mathcal{B} \,: \mathbb{F}^{n_1}_2 \otimes \mathbb{F}^{m_2}_2 \xrightarrow{\partial_2} (\mathbb{F}^{n_1}_2 \otimes \mathbb{F}^{n_2}_2) \oplus (\mathbb{F}^{m_2}_2 \otimes \mathbb{F}^{m_1}_2) \xrightarrow{\partial_1} \mathbb{F}^{m_1}_2 \otimes \mathbb{F}^{n_2}_2,
\end{align*}
where $H_1$ and $H_2$ are parity-check matrices of classical linear codes and
\begin{equation*}
    \partial_2 = \begin{pmatrix}
        I_{n_1} \otimes H^{\mathrm{T}}_2\\
        H_1 \otimes I_{m_2}
    \end{pmatrix}
    \quad , \quad
    \partial_1 = \begin{pmatrix}
        H_1 \otimes I_{n_2} & I_{m_1} \otimes H^{\mathrm{T}}_2    
    \end{pmatrix}.
\end{equation*}
Extending the chains on both sides by zero and applying \cref{Kunneth}, the $Z$ and $X$ logicals are of the form
\begin{equation*}
    (\ker H_1 \otimes \mathbb{F}^{n_2}_2 / \IM H^{\mathrm{T}}_2) \oplus (\mathbb{F}^{n_1}_2 / \IM H_1 \otimes \ker H^{\mathrm{T}}_2)
\end{equation*}
and
\begin{equation*}
    (\mathbb{F}^{n_1}_2 / \IM H^{\mathrm{T}}_1 \otimes \ker H_2) \oplus (\ker H^{\mathrm{T}}_1 \otimes \mathbb{F}^{n_2}_2 / \IM H^{\mathrm{T}}_2),
\end{equation*}
respectively.

\subsection{The Mapping Cone}\label{app:mappingcone}
Closely related to the above is the concept of the mapping cone. Consider the two chain complexes $\mathcal{A}$ and $\mathcal{B}$ below
\begin{equation}
    \begin{tikzcd}
        {\mathcal{A}: \cdots} && {A_{i + 1}} && {A_i} && {A_{i - 1}} && \cdots \\
        \\
        {\mathcal{B}: \cdots} && {B_{i + 1}} && {B_i} && {B_{i - 1}} && \cdots
        \arrow[from=1-1, to=1-3]
        \arrow[from=3-1, to=3-3]
        \arrow["{\partial^\mathcal{A}_{i + 1}}", from=1-3, to=1-5]
        \arrow["{\partial^\mathcal{A}_i}", from=1-5, to=1-7]
        \arrow["{\partial^\mathcal{B}_{i + 1}}", from=3-3, to=3-5]
        \arrow["{\partial^\mathcal{B}_i}", from=3-5, to=3-7]
        \arrow["{f_{i + 1}}", from=1-3, to=3-3]
        \arrow["{f_i}", from=1-5, to=3-5]
        \arrow["{f_{i - 1}}", from=1-7, to=3-7]
        \arrow[from=1-7, to=1-9]
        \arrow[from=3-7, to=3-9]
    \end{tikzcd}.
\end{equation}
The maps $f_i$ are called chain maps and we require that they are homomorphisms that commute with the other maps, i.e. $\partial^\mathcal{B}_{i+1}(f_{i + 1}(a)) = f_i\left(\partial^\mathcal{A}_{i + 1}(a)\right)$ for $a \in A_{i + 1}$. The mapping cone is defined to be the chain complex with spaces $\mathrm{cone}(f)_i = A_i \oplus B_{i + 1}$. Graphically,
\begin{equation}
    \begin{tikzcd}
        {A_i} && {A_{i - 1}} \\
        \oplus && \oplus \\
        {B_{i + 1}} && {B_i}
        \arrow["{f_i}", from=1-1, to=3-3]
        \arrow["{\partial^\mathcal{B}_{i + 1}}", from=3-1, to=3-3]
        \arrow["{\partial^\mathcal{A}_i}", from=1-1, to=1-3]
    \end{tikzcd}.
\end{equation}
Similar to the previous section, the maps $\partial_i: \mathrm{cone}(f)_i \to \mathrm{cone}(f)_{i - 1}$ are
\begin{equation*}
    \mathrm{Mat}\left(\partial_i\right) = \, \begin{pNiceArray}{cc}[first-row, first-col]
        & A_i & B_{i + 1}\\
        A_{i - 1} & \partial^\mathcal{A}_i & 0\\
        B_i & f_i & \partial^\mathcal{B}_{i + 1}
    \end{pNiceArray}.
\end{equation*}
Note that in non-binary fields, the first column should receive a minus sign, or equivalently, $f_i$ should be replaced with $(-1)^if_i$. From the mapping cone we have the short exact (split) sequence $0 \to A_{i - 1} \to \mathrm{cone}(f)_i \to B_i \to 0$. This induces the long exact sequence on homology (via the Snake Lemma) \cite{rotman2009introduction}.
\begin{equation}\label{LES}
    H_{k + 1}(\mathrm{cone}(f)) \to H_k(\mathcal{A}) \to H_k(\mathcal{B}) \to H_k(\mathrm{cone}(f)).
\end{equation}

\section{Examples}
\subsection{Hypergraph Product Codes}\label{app:examples-hgp}

{\renewcommand{\arraystretch}{1.4}
\begin{center}
\begin{table}[ht!]
    \begin{tabular}{||c|c|c|c|c|c|c||}
        \hline
        $\mathcal{C}(H)$ & $\mathrm{HGP}(H)$ & $R$ & $\mathrm{HGP}(\tilde{H})$ & $R$ & $\mathrm{HGP}(\tilde{H}^{(c)})$ & $R$ \\
        \hline
        $[11, 2, 7]$ & $[\![202, 4, 7]\!]$ & 0.020 & $[\![884, 4, 12]\!]$ & 0.005 & $[\![580, 4, 10]\!]$ & 0.007 \\
        \hline
        $[11, 3, 6]$ & $[\![185, 9, 6]\!]$ & 0.049 & $[\![1745, 9, 12\rightarrow15]\!]$ & 0.005 & $[\![765, 9, 8\rightarrow10]\!]$ & 0.012\\
        \hline
        $[11, 4, 5]$ & $[\![170, 16, 5]\!]$ & 0.094 & $[\![1930, 16, 12\rightarrow13]\!]$ & 0.008 & $[\![586, 16, 8]\!]$ & 0.027 \\
        \hline
        $[11, 5, 4]$ & $[\![157, 25, 4]\!]$ & 0.159 & $[\![557, 25, 5\rightarrow6]\!]$ & 0.045 & $[\![325, 25, 4\rightarrow5]\!]$ & 0.077 \\
        \hline
        $[11, 6, 4]$ & $[\![146, 36, 4]\!]$ & 0.247 & $[\![1170, 36, 7\rightarrow9]\!]$ & 0.031 & $[\![530, 36, 4\rightarrow7]\!]$ & 0.068 \\
        \hline
        $[11, 7, 3]$ & $[\![137, 49, 3]\!]$ & 0.358  & $[\![1885, 49, 7\rightarrow10]\!]$ & 0.026 & $[\![865, 49, 4\rightarrow8]\!]$ & 0.057 \\
        \hline
        $[12, 2, 8]$ & $[\![244, 4, 8]\!]$ & 0.016  & $[\![1060, 4, 14]\!]$ & 0.004 & $[\![724, 4, 12]\!]$ & 0.006\\
        \hline
        $[12, 3, 6]$ & $[\![225, 9, 6]\!]$ & 0.040 & $[\![1865, 9, 14\rightarrow15]\!]$ & 0.005 & $[\![845, 9, 10]\!]$ & 0.011 \\
        \hline
        $[12, 4, 6]$ & $[\![208, 16, 6]\!]$ & 0.077  & $[\![3880, 16, 14\rightarrow19]\!]$ & 0.004 & $[\![1360, 16, 8\rightarrow12]\!]$ & 0.012 \\
        \hline
        $[12, 5, 4]$ & $[\![193, 25, 4]\!]$ & 0.130 & $[\![697, 25, 5]\!]$ & 0.036 & $[\![325, 25, 4]\!]$ & 0.077  \\
        \hline
        $[12, 6, 4]$ & $[\![180, 36, 4]\!]$ & 0.200 & $[\![900, 36, 6\rightarrow7]\!]$ & 0.040 & $[\![468, 36, 4\rightarrow6]\!]$ & 0.077 \\
        \hline
        $[12, 7, 4]$ & $[\![169, 49, 4]\!]$ & 0.290 & $[\![1765, 49, 7\rightarrow10]\!]$ & 0.028 & $[\![785, 49, 4\rightarrow7]\!]$ & 0.062  \\
        \hline
        $[12, 8, 3]$ & $[\![160, 64, 3]\!]$ & 0.400  & $[\![2210, 64, 7\rightarrow9]\!]$ & 0.029 & $[\![1090, 64, 4\rightarrow7]\!]$ & 0.059 \\
        \hline
        $[13, 2, 8]$ & $[\![290, 4, 8]\!]$ & 0.014 & $[\![1252, 4, 14]\!]$ & 0.003 & $[\![884, 4, 12]\!]$ & 0.005  \\
        \hline
        $[13, 3, 7]$ & $[\![269, 9, 7]\!]$ & 0.033 & $[\![2385, 9, 16\rightarrow17]\!]$ & 0.004 & $[\![1205, 9, 12]\!]$ & 0.007 \\
        \hline
        $[13, 4, 6]$ & $[\![250, 16, 6]\!]$ & 0.064 & $[\![4058, 16, 17\rightarrow19]\!]$ & 0.004 & $[\![1466, 16, 11\rightarrow12]\!]$ & 0.011 \\
        \hline
        $[13, 5, 5]$ & $[\![233, 25, 5]\!]$ & 0.107& $[\![4717, 25, 13\rightarrow18]\!]$ & 0.005  & $[\![1637, 25, 8\rightarrow12]\!]$ & 0.015 \\
        \hline
        $[13, 6, 4]$ & $[\![218, 36, 4]\!]$ & 0.165  & $[\![900, 36, 5]\!]$ & 0.040 & $[\![468, 36, 4]\!]$ & 0.077\\
        \hline
        $[13, 7, 4]$ & $[\![205, 49, 4]\!]$ & 0.239 & $[\![1225, 49, 6\rightarrow8]\!]$ & 0.040 & $[\![709, 49, 4\rightarrow7]\!]$ & 0.069 \\
        \hline
        $[13, 8, 4]$ & $[\![194, 64, 4]\!]$ & 0.330 & $[\![2210, 64, 7\rightarrow10]\!]$ & 0.029 & $[\![1090, 64, 4\rightarrow8]\!]$ & 0.059  \\\hline
        $[13, 9, 3]$ & $[\![185, 81, 3]\!]$ & 0.438 & $[\![2561, 81, 6\rightarrow9]\!]$ & 0.032 & $[\![1341, 81, 4\rightarrow7]\!]$ & 0.060  \\
        \hline
        \end{tabular}
        \caption{Classical weight reduction applied to hypergraph product codes with some of the best known linear codes with $11 \leq n \leq 13$. $\mathcal{C}(H)$ is the linear code with parity-check matrix $H$ obtained from GAP (see main text). For each hypergraph product code, we give its encoding rate $R=k/n$. For each weight-reduction method, we apply the relevant algorithm 10,000 times using different permutations of the input parity-check matrix. In cases where permutations improved the distance, we use the notation $d_1 \rightarrow d_2$, where $d_1$ indicates the distance without permutations, and $d_2$ indicates the highest obtained distance.}
        \label{tab:hgp-examples-app-A}
\end{table}
\end{center}
}

{\renewcommand{\arraystretch}{1.4}
\begin{center}
\begin{table}[ht!]
    \begin{tabular}{||c|c|c|c|c|c|c||}
        \hline
        $\mathcal{C}(H)$ & $\mathrm{HGP}(H)$ & $R$ & $\mathrm{HGP}(\tilde H)$ & $R$ & $\mathrm{HGP}(\tilde H^{(c)})$ & $R$ \\
        \hline
        $[14, 2, 9]$ & $[\![340, 4, 9]\!]$ & 0.012 & $[\![1570, 4, 16]\!]$ & 0.003 & $[\![1154, 4, 14]\!]$ & 0.003 \\
        \hline
        $[14, 3, 8]$ & $[\![317, 9, 8]\!]$ & 0.028 & $[\![2669, 9, 16\rightarrow18]\!]$ & 0.003 & $[\![1409, 9, 12\rightarrow13]\!]$ & 0.006 \\
        \hline
        $[14, 4, 7]$ & $[\![296, 16, 7]\!]$ & 0.054 & $[\![5008, 16, 18\rightarrow21]\!]$ & 0.003 & $[\![2056, 16, 12\rightarrow14]\!]$ & 0.008 \\
        \hline
        $[14, 5, 6]$ & $[\![277, 25, 6]\!]$ & 0.090 & $[\![6173, 25, 17\rightarrow20]\!]$ & 0.004 & $[\![2257, 25, 11\rightarrow13]\!]$ & 0.011 \\
        \hline
        $[14, 6, 5]$ & $[\![260, 36, 5]\!]$ & 0.138 & $[\![7460, 36, 13\rightarrow18]\!]$ & 0.005 & $[\![2468, 36, 8\rightarrow12]\!]$ & 0.015 \\
        \hline
        $[14, 7, 4]$ & $[\![245, 49, 4]\!]$ & 0.200 & $[\![1429, 49, 6\rightarrow8]\!]$ & 0.034 & $[\![709, 49, 4\rightarrow5]\!]$ & 0.069 \\
        \hline
        $[14, 8, 4]$ & $[\![232, 64, 4]\!]$ & 0.276 & $[\![1832, 64, 6\rightarrow9]\!]$ & 0.035 & $[\![1000, 64, 4\rightarrow7]\!]$ & 0.064 \\
        \hline
        $[14, 9, 4]$ & $[\![221, 81, 4]\!]$ & 0.367 & $[\![2705, 81, 7\rightarrow11]\!]$ & 0.030 & $[\![1445, 81, 4\rightarrow8]\!]$ & 0.056 \\
        \hline
        $[14, 10, 3]$ & $[\![212, 100, 3]\!]$ & 0.472 & $[\![2938, 100, 6\rightarrow9]\!]$ & 0.034 & $[\![1618, 100, 4\rightarrow7]\!]$ & 0.062 \\
        \hline
        $[15, 2, 10]$ & $[\![394, 4, 10]\!]$ & 0.010 & $[\![1802, 4, 18]\!]$ & 0.002 & $[\![1354, 4, 16]\!]$ & 0.003\\
        \hline
        $[15, 3, 8]$ & $[\![369, 9, 8]\!]$ & 0.024 & $[\![2817, 9, 16\rightarrow18]\!]$ & 0.003 & $[\![1517, 9, 12\rightarrow13]\!]$ & 0.006\\
        \hline
        $[15, 4, 8]$ & $[\![346, 16, 8]\!]$ & 0.046 & $[\![5626, 16, 18\rightarrow22]\!]$ & 0.003 & $[\![2458, 16, 12\rightarrow15]\!]$ & 0.007\\
        \hline
        $[15, 5, 7]$ & $[\![325, 25, 7]\!]$ & 0.077 & $[\![10237, 25, 21\rightarrow25]\!]$ & 0.002 & $[\![3457, 25, 12\rightarrow16]\!]$ & 0.007\\
        \hline
        $[15, 6, 6]$ & $[\![306, 36, 6]\!]$ & 0.118 & $[\![9266, 36, 17\rightarrow20]\!]$ & 0.004 & $[\![3218, 36, 11\rightarrow14]\!]$ & 0.011\\
        \hline
        $[15, 7, 5]$ & $[\![289, 49, 5]\!]$ & 0.170 & $[\![9965, 49, 14\rightarrow19]\!]$ & 0.005 & $[\![3305, 49, 8\rightarrow13]\!]$ & 0.015\\
        \hline
        $[15, 8, 4]$ & $[\![274, 64, 4]\!]$ & 0.234 & $[\![2920, 64, 6\rightarrow9]\!]$ & 0.022 & $[\![1384, 64, 4\rightarrow7]\!]$ & 0.046\\
        \hline
        $[15, 9, 4]$ & $[\![261, 81, 4]\!]$ & 0.310 & $[\![2561, 81, 7\rightarrow10]\!]$ & 0.032 & $[\![1341, 81, 4\rightarrow8]\!]$ & 0.060\\
        \hline
        $[15, 10, 4]$ & $[\![250, 100, 4]\!]$ & 0.400 & $[\![3250, 100, 7\rightarrow11]\!]$ & 0.031 & $[\![1850, 100, 4\rightarrow9]\!]$ & 0.054\\
        \hline
        $[15, 11, 3]$ & $[\![241, 121, 3]\!]$ & 0.502 & $[\![3341, 121, 6\rightarrow9]\!]$ & 0.036 & $[\![1921, 121, 3\rightarrow7]\!]$ & 0.063\\
        \hline
    \end{tabular}
    \caption{Classical weight reduction applied to hypergraph product codes with some of the best known linear codes with $14 \leq n \leq 15$. $\mathcal{C}(H)$ is the linear code with parity-check matrix $H$ obtained from GAP (see main text). For each hypergraph product code, we give its encoding rate $R=k/n$. For each weight-reduction method, we apply the relevant algorithm 10,000 times using different permutations of the input parity-check matrix. In cases where permutations improved the distance, we use the notation $d_1 \rightarrow d_2$, where $d_1$ indicates the distance without permutations, and $d_2$ indicates the highest obtained distance.}
    \label{tab:hgp-examples-app-B}
\end{table}
\end{center}
}

\newpage
\subsection{Quasi-Cyclic Codes}\label{app:examples-lp}
Here we give the base matrices used to construct the lifted product code examples in \cref{tab:lp-examples}.
\begin{enumerate}
    \item From~\cite[Table 2]{bocharova2009}, a $[52, 27, 6]$ quasi-cyclic code with lift size $\ell = 13$ and base matrix 
    \begin{equation*}
        A = 
        \begin{pmatrix}
            1 & 1 & 1 & 1\\
            1 & x & x^3 & x^9
        \end{pmatrix}.
    \end{equation*}
    \item From~\cite[Example 11]{smarandache2012quasi}, a $[124, 33, 24]$ quasi-cyclic code with lift size $\ell = 31$ and base matrix
    \begin{equation*}
        A = 
        \begin{pmatrix}
            x & x^2 & x^4 & x^8\\
            x^5 & x^{10} & x^{20} & x^9\\
            x^{25} & x^{19} & x^7 & x^{14}
        \end{pmatrix}.
    \end{equation*}
    \item From~\cite[Table 1]{raveendran2022finite}, a $[28, 9, 10]$ quasi-cyclic code with lift size $\ell = 7$ and base matrix
    \begin{equation*}
        A = 
        \begin{pmatrix}
            1 & 1 & 1 & 1\\
            1 & x & x^2 & x^5\\
            1 & x^6 & x^3 & x
        \end{pmatrix}.
    \end{equation*}
    \item From~\cite[Table 1]{raveendran2022finite}, a $[36, 11, 12]$ quasi-cyclic code with lift size $\ell = 9$ and base matrix
    \begin{equation*}
        A = 
        \begin{pmatrix}
            1 & 1 & 1 & 1\\
            1 & x & x^6 & x^7\\
            1 & x^4 & x^5 & x^2
        \end{pmatrix}.
    \end{equation*}
    \item From~\cite[Table 1]{raveendran2022finite}, a $[68, 19, 18]$ quasi-cyclic code with lift size $\ell = 17$ and base matrix
    \begin{equation*}
        A = 
        \begin{pmatrix}
            1 & 1 & 1 & 1\\
            1 & x & x^2 & x^{11}\\
            1 & x^8 & x^{12} & x^{13}
        \end{pmatrix}
    \end{equation*}
    \item From~\cite[Table 1]{raveendran2022finite}, a $[76, 21, 20]$ quasi-cyclic code with lift size $\ell = 9$ and base matrix
    \begin{equation*}
        A = 
        \begin{pmatrix}
            1 & 1 & 1 & 1\\
            1 & x & x^6 & x^7\\
            1 & x^4 & x^5 & x^2
        \end{pmatrix}.
    \end{equation*}
\end{enumerate}

We can also weight reduce base matrices with higher weight entries in certain special cases.
Consider the matrix
\begin{equation*}
    A = 
    \begin{pmatrix}
        x + x^2 &  & x^4 & x^8\\
        x^5 & x^9 & x^{10} + x^{20} & \\
         & x^{25} + x^{19} &  & x^7 + x^{14}
    \end{pmatrix},
\end{equation*}
which for lift size $\ell = 46$ has an associated quasi-cyclic code with parameters $[184, 47, 32]$~\cite[Example 13]{smarandache2012quasi}.
The corresponding lifted product code $\mathrm{LP}(A)$ has parameters $[\![1150,50,\leq 21]\!]$.
If we choose the correct permutation for each row then we can obtain a weight reduced matrix where each column and row have weight at most three:
\begin{equation*}
    \tilde A = 
    \begin{pmatrix}
        x + x^2 &   &   &   & 1 &   &   &   &  \\
          &   & x^4 &   & 1 & 1 &   &   &  \\
          &   &   & x^8 &   & 1 &   &   &  \\
          &   & x^{1 } + x^{2 } &   &   &   & 1 &   &  \\
        x^5 &   &   &   &   &   & 1 & 1 &  \\
          & x^9 &   &   &   &   &   & 1 &  \\
          &   &   & x^7 + x^{14} &   &   &   &   & 1\\
          & x^{25} + x^{19} &   &   &   &   &   &   & 1
    \end{pmatrix}.
\end{equation*}
The corresponding quasi-cyclic code has parameters $[414,47,81]$ and the lifted product code $\mathrm{LP}(\tilde A)$ has parameters $[\![6670,50,\leq70]\!]$.

\bibliography{main}

\begin{thebibliography}{90}%
\makeatletter
\providecommand \@ifxundefined [1]{%
 \@ifx{#1\undefined}
}%
\providecommand \@ifnum [1]{%
 \ifnum #1\expandafter \@firstoftwo
 \else \expandafter \@secondoftwo
 \fi
}%
\providecommand \@ifx [1]{%
 \ifx #1\expandafter \@firstoftwo
 \else \expandafter \@secondoftwo
 \fi
}%
\providecommand \natexlab [1]{#1}%
\providecommand \enquote  [1]{``#1''}%
\providecommand \bibnamefont  [1]{#1}%
\providecommand \bibfnamefont [1]{#1}%
\providecommand \citenamefont [1]{#1}%
\providecommand \href@noop [0]{\@secondoftwo}%
\providecommand \href [0]{\begingroup \@sanitize@url \@href}%
\providecommand \@href[1]{\@@startlink{#1}\@@href}%
\providecommand \@@href[1]{\endgroup#1\@@endlink}%
\providecommand \@sanitize@url [0]{\catcode `\\12\catcode `\$12\catcode `\&12\catcode `\#12\catcode `\^12\catcode `\_12\catcode `\%12\relax}%
\providecommand \@@startlink[1]{}%
\providecommand \@@endlink[0]{}%
\providecommand \url  [0]{\begingroup\@sanitize@url \@url }%
\providecommand \@url [1]{\endgroup\@href {#1}{\urlprefix }}%
\providecommand \urlprefix  [0]{URL }%
\providecommand \Eprint [0]{\href }%
\providecommand \doibase [0]{https://doi.org/}%
\providecommand \selectlanguage [0]{\@gobble}%
\providecommand \bibinfo  [0]{\@secondoftwo}%
\providecommand \bibfield  [0]{\@secondoftwo}%
\providecommand \translation [1]{[#1]}%
\providecommand \BibitemOpen [0]{}%
\providecommand \bibitemStop [0]{}%
\providecommand \bibitemNoStop [0]{.\EOS\space}%
\providecommand \EOS [0]{\spacefactor3000\relax}%
\providecommand \BibitemShut  [1]{\csname bibitem#1\endcsname}%
\let\auto@bib@innerbib\@empty
%</preamble>
\bibitem [{\citenamefont {Campbell}\ \emph {et~al.}(2017)\citenamefont {Campbell}, \citenamefont {Terhal},\ and\ \citenamefont {Vuillot}}]{campbell2017}%
  \BibitemOpen
  \bibfield  {author} {\bibinfo {author} {\bibfnamefont {E.~T.}\ \bibnamefont {Campbell}}, \bibinfo {author} {\bibfnamefont {B.~M.}\ \bibnamefont {Terhal}},\ and\ \bibinfo {author} {\bibfnamefont {C.}~\bibnamefont {Vuillot}},\ }\bibfield  {title} {\bibinfo {title} {Roads towards fault-tolerant universal quantum computation},\ }\href {https://doi.org/10.1038/nature23460} {\bibfield  {journal} {\bibinfo  {journal} {Nature}\ }\textbf {\bibinfo {volume} {549}},\ \bibinfo {pages} {172} (\bibinfo {year} {2017})}\BibitemShut {NoStop}%
\bibitem [{\citenamefont {Dalzell}\ \emph {et~al.}(2023)\citenamefont {Dalzell}, \citenamefont {McArdle}, \citenamefont {Berta}, \citenamefont {Bienias}, \citenamefont {Chen}, \citenamefont {Gily{\'e}n}, \citenamefont {Hann}, \citenamefont {Kastoryano}, \citenamefont {Khabiboulline}, \citenamefont {Kubica}, \citenamefont {Salton}, \citenamefont {Wang},\ and\ \citenamefont {Brand{\~a}o}}]{dalzell2023}%
  \BibitemOpen
  \bibfield  {author} {\bibinfo {author} {\bibfnamefont {A.~M.}\ \bibnamefont {Dalzell}}, \bibinfo {author} {\bibfnamefont {S.}~\bibnamefont {McArdle}}, \bibinfo {author} {\bibfnamefont {M.}~\bibnamefont {Berta}}, \bibinfo {author} {\bibfnamefont {P.}~\bibnamefont {Bienias}}, \bibinfo {author} {\bibfnamefont {C.-F.}\ \bibnamefont {Chen}}, \bibinfo {author} {\bibfnamefont {A.}~\bibnamefont {Gily{\'e}n}}, \bibinfo {author} {\bibfnamefont {C.~T.}\ \bibnamefont {Hann}}, \bibinfo {author} {\bibfnamefont {M.~J.}\ \bibnamefont {Kastoryano}}, \bibinfo {author} {\bibfnamefont {E.~T.}\ \bibnamefont {Khabiboulline}}, \bibinfo {author} {\bibfnamefont {A.}~\bibnamefont {Kubica}}, \bibinfo {author} {\bibfnamefont {G.}~\bibnamefont {Salton}}, \bibinfo {author} {\bibfnamefont {S.}~\bibnamefont {Wang}},\ and\ \bibinfo {author} {\bibfnamefont {F.~G. S.~L.}\ \bibnamefont {Brand{\~a}o}},\ }\href@noop {} {\bibinfo {title} {Quantum algorithms: {{A}} survey of applications and end-to-end complexities}} (\bibinfo {year} {2023}),\
  \Eprint {https://arxiv.org/abs/2310.03011} {arxiv:2310.03011} \BibitemShut {NoStop}%
\bibitem [{\citenamefont {{Ryan-Anderson}}\ \emph {et~al.}(2021)\citenamefont {{Ryan-Anderson}}, \citenamefont {Bohnet}, \citenamefont {Lee}, \citenamefont {Gresh}, \citenamefont {Hankin}, \citenamefont {Gaebler}, \citenamefont {Francois}, \citenamefont {Chernoguzov}, \citenamefont {Lucchetti}, \citenamefont {Brown}, \citenamefont {Gatterman}, \citenamefont {Halit}, \citenamefont {Gilmore}, \citenamefont {Gerber}, \citenamefont {Neyenhuis}, \citenamefont {Hayes},\ and\ \citenamefont {Stutz}}]{ryan-anderson2021}%
  \BibitemOpen
  \bibfield  {author} {\bibinfo {author} {\bibfnamefont {C.}~\bibnamefont {{Ryan-Anderson}}}, \bibinfo {author} {\bibfnamefont {J.~G.}\ \bibnamefont {Bohnet}}, \bibinfo {author} {\bibfnamefont {K.}~\bibnamefont {Lee}}, \bibinfo {author} {\bibfnamefont {D.}~\bibnamefont {Gresh}}, \bibinfo {author} {\bibfnamefont {A.}~\bibnamefont {Hankin}}, \bibinfo {author} {\bibfnamefont {J.~P.}\ \bibnamefont {Gaebler}}, \bibinfo {author} {\bibfnamefont {D.}~\bibnamefont {Francois}}, \bibinfo {author} {\bibfnamefont {A.}~\bibnamefont {Chernoguzov}}, \bibinfo {author} {\bibfnamefont {D.}~\bibnamefont {Lucchetti}}, \bibinfo {author} {\bibfnamefont {N.~C.}\ \bibnamefont {Brown}}, \bibinfo {author} {\bibfnamefont {T.~M.}\ \bibnamefont {Gatterman}}, \bibinfo {author} {\bibfnamefont {S.~K.}\ \bibnamefont {Halit}}, \bibinfo {author} {\bibfnamefont {K.}~\bibnamefont {Gilmore}}, \bibinfo {author} {\bibfnamefont {J.~A.}\ \bibnamefont {Gerber}}, \bibinfo {author} {\bibfnamefont {B.}~\bibnamefont {Neyenhuis}}, \bibinfo {author}
  {\bibfnamefont {D.}~\bibnamefont {Hayes}},\ and\ \bibinfo {author} {\bibfnamefont {R.~P.}\ \bibnamefont {Stutz}},\ }\bibfield  {title} {\bibinfo {title} {Realization of {{Real-Time Fault-Tolerant Quantum Error Correction}}},\ }\href {https://doi.org/10.1103/PhysRevX.11.041058} {\bibfield  {journal} {\bibinfo  {journal} {Phys. Rev. X}\ }\textbf {\bibinfo {volume} {11}},\ \bibinfo {pages} {041058} (\bibinfo {year} {2021})}\BibitemShut {NoStop}%
\bibitem [{\citenamefont {Krinner}\ \emph {et~al.}(2022)\citenamefont {Krinner}, \citenamefont {Lacroix}, \citenamefont {Remm}, \citenamefont {Di~Paolo}, \citenamefont {Genois}, \citenamefont {Leroux}, \citenamefont {Hellings}, \citenamefont {Lazar}, \citenamefont {Swiadek}, \citenamefont {Herrmann}, \citenamefont {Norris}, \citenamefont {Andersen}, \citenamefont {M{\"u}ller}, \citenamefont {Blais}, \citenamefont {Eichler},\ and\ \citenamefont {Wallraff}}]{krinner2022}%
  \BibitemOpen
  \bibfield  {author} {\bibinfo {author} {\bibfnamefont {S.}~\bibnamefont {Krinner}}, \bibinfo {author} {\bibfnamefont {N.}~\bibnamefont {Lacroix}}, \bibinfo {author} {\bibfnamefont {A.}~\bibnamefont {Remm}}, \bibinfo {author} {\bibfnamefont {A.}~\bibnamefont {Di~Paolo}}, \bibinfo {author} {\bibfnamefont {E.}~\bibnamefont {Genois}}, \bibinfo {author} {\bibfnamefont {C.}~\bibnamefont {Leroux}}, \bibinfo {author} {\bibfnamefont {C.}~\bibnamefont {Hellings}}, \bibinfo {author} {\bibfnamefont {S.}~\bibnamefont {Lazar}}, \bibinfo {author} {\bibfnamefont {F.}~\bibnamefont {Swiadek}}, \bibinfo {author} {\bibfnamefont {J.}~\bibnamefont {Herrmann}}, \bibinfo {author} {\bibfnamefont {G.~J.}\ \bibnamefont {Norris}}, \bibinfo {author} {\bibfnamefont {C.~K.}\ \bibnamefont {Andersen}}, \bibinfo {author} {\bibfnamefont {M.}~\bibnamefont {M{\"u}ller}}, \bibinfo {author} {\bibfnamefont {A.}~\bibnamefont {Blais}}, \bibinfo {author} {\bibfnamefont {C.}~\bibnamefont {Eichler}},\ and\ \bibinfo {author} {\bibfnamefont
  {A.}~\bibnamefont {Wallraff}},\ }\bibfield  {title} {\bibinfo {title} {Realizing repeated quantum error correction in a distance-three surface code},\ }\href {https://doi.org/10.1038/s41586-022-04566-8} {\bibfield  {journal} {\bibinfo  {journal} {Nature}\ }\textbf {\bibinfo {volume} {605}},\ \bibinfo {pages} {669} (\bibinfo {year} {2022})}\BibitemShut {NoStop}%
\bibitem [{\citenamefont {Sundaresan}\ \emph {et~al.}(2023)\citenamefont {Sundaresan}, \citenamefont {Yoder}, \citenamefont {Kim}, \citenamefont {Li}, \citenamefont {Chen}, \citenamefont {Harper}, \citenamefont {Thorbeck}, \citenamefont {Cross}, \citenamefont {C{\'o}rcoles},\ and\ \citenamefont {Takita}}]{sundaresan2023}%
  \BibitemOpen
  \bibfield  {author} {\bibinfo {author} {\bibfnamefont {N.}~\bibnamefont {Sundaresan}}, \bibinfo {author} {\bibfnamefont {T.~J.}\ \bibnamefont {Yoder}}, \bibinfo {author} {\bibfnamefont {Y.}~\bibnamefont {Kim}}, \bibinfo {author} {\bibfnamefont {M.}~\bibnamefont {Li}}, \bibinfo {author} {\bibfnamefont {E.~H.}\ \bibnamefont {Chen}}, \bibinfo {author} {\bibfnamefont {G.}~\bibnamefont {Harper}}, \bibinfo {author} {\bibfnamefont {T.}~\bibnamefont {Thorbeck}}, \bibinfo {author} {\bibfnamefont {A.~W.}\ \bibnamefont {Cross}}, \bibinfo {author} {\bibfnamefont {A.~D.}\ \bibnamefont {C{\'o}rcoles}},\ and\ \bibinfo {author} {\bibfnamefont {M.}~\bibnamefont {Takita}},\ }\bibfield  {title} {\bibinfo {title} {Demonstrating multi-round subsystem quantum error correction using matching and maximum likelihood decoders},\ }\href {https://doi.org/10.1038/s41467-023-38247-5} {\bibfield  {journal} {\bibinfo  {journal} {Nat Commun}\ }\textbf {\bibinfo {volume} {14}},\ \bibinfo {pages} {2852} (\bibinfo {year} {2023})}\BibitemShut
  {NoStop}%
\bibitem [{\citenamefont {{Google Quantum AI}}\ \emph {et~al.}(2023)\citenamefont {{Google Quantum AI}}, \citenamefont {Acharya}, \citenamefont {Aleiner}, \citenamefont {Allen}, \citenamefont {Andersen}, \citenamefont {Ansmann}, \citenamefont {Arute}, \citenamefont {Arya}, \citenamefont {Asfaw}, \citenamefont {Atalaya}, \citenamefont {Babbush}, \citenamefont {Bacon}, \citenamefont {Bardin}, \citenamefont {Basso}, \citenamefont {Bengtsson}, \citenamefont {Boixo}, \citenamefont {Bortoli}, \citenamefont {Bourassa}, \citenamefont {Bovaird}, \citenamefont {Brill}, \citenamefont {Broughton}, \citenamefont {Buckley}, \citenamefont {Buell}, \citenamefont {Burger}, \citenamefont {Burkett}, \citenamefont {Bushnell}, \citenamefont {Chen}, \citenamefont {Chen}, \citenamefont {Chiaro}, \citenamefont {Cogan}, \citenamefont {Collins}, \citenamefont {Conner}, \citenamefont {Courtney}, \citenamefont {Crook}, \citenamefont {Curtin}, \citenamefont {Debroy}, \citenamefont {Del Toro~Barba}, \citenamefont {Demura}, \citenamefont
  {Dunsworth}, \citenamefont {Eppens}, \citenamefont {Erickson}, \citenamefont {Faoro}, \citenamefont {Farhi}, \citenamefont {Fatemi}, \citenamefont {Flores~Burgos}, \citenamefont {Forati}, \citenamefont {Fowler}, \citenamefont {Foxen}, \citenamefont {Giang}, \citenamefont {Gidney}, \citenamefont {Gilboa}, \citenamefont {Giustina}, \citenamefont {Grajales~Dau}, \citenamefont {Gross}, \citenamefont {Habegger}, \citenamefont {Hamilton}, \citenamefont {Harrigan}, \citenamefont {Harrington}, \citenamefont {Higgott}, \citenamefont {Hilton}, \citenamefont {Hoffmann}, \citenamefont {Hong}, \citenamefont {Huang}, \citenamefont {Huff}, \citenamefont {Huggins}, \citenamefont {Ioffe}, \citenamefont {Isakov}, \citenamefont {Iveland}, \citenamefont {Jeffrey}, \citenamefont {Jiang}, \citenamefont {Jones}, \citenamefont {Juhas}, \citenamefont {Kafri}, \citenamefont {Kechedzhi}, \citenamefont {Kelly}, \citenamefont {Khattar}, \citenamefont {Khezri}, \citenamefont {Kieferov{\'a}}, \citenamefont {Kim}, \citenamefont {Kitaev},
  \citenamefont {Klimov}, \citenamefont {Klots}, \citenamefont {Korotkov}, \citenamefont {Kostritsa}, \citenamefont {Kreikebaum}, \citenamefont {Landhuis}, \citenamefont {Laptev}, \citenamefont {Lau}, \citenamefont {Laws}, \citenamefont {Lee}, \citenamefont {Lee}, \citenamefont {Lester}, \citenamefont {Lill}, \citenamefont {Liu}, \citenamefont {Locharla}, \citenamefont {Lucero}, \citenamefont {Malone}, \citenamefont {Marshall}, \citenamefont {Martin}, \citenamefont {McClean}, \citenamefont {McCourt}, \citenamefont {McEwen}, \citenamefont {Megrant}, \citenamefont {Meurer~Costa}, \citenamefont {Mi}, \citenamefont {Miao}, \citenamefont {Mohseni}, \citenamefont {Montazeri}, \citenamefont {Morvan}, \citenamefont {Mount}, \citenamefont {Mruczkiewicz}, \citenamefont {Naaman}, \citenamefont {Neeley}, \citenamefont {Neill}, \citenamefont {Nersisyan}, \citenamefont {Neven}, \citenamefont {Newman}, \citenamefont {Ng}, \citenamefont {Nguyen}, \citenamefont {Nguyen}, \citenamefont {Niu}, \citenamefont {O'Brien},
  \citenamefont {Opremcak}, \citenamefont {Platt}, \citenamefont {Petukhov}, \citenamefont {Potter}, \citenamefont {Pryadko}, \citenamefont {Quintana}, \citenamefont {Roushan}, \citenamefont {Rubin}, \citenamefont {Saei}, \citenamefont {Sank}, \citenamefont {Sankaragomathi}, \citenamefont {Satzinger}, \citenamefont {Schurkus}, \citenamefont {Schuster}, \citenamefont {Shearn}, \citenamefont {Shorter}, \citenamefont {Shvarts}, \citenamefont {Skruzny}, \citenamefont {Smelyanskiy}, \citenamefont {Smith}, \citenamefont {Sterling}, \citenamefont {Strain}, \citenamefont {Szalay}, \citenamefont {Torres}, \citenamefont {Vidal}, \citenamefont {Villalonga}, \citenamefont {Vollgraff~Heidweiller}, \citenamefont {White}, \citenamefont {Xing}, \citenamefont {Yao}, \citenamefont {Yeh}, \citenamefont {Yoo}, \citenamefont {Young}, \citenamefont {Zalcman}, \citenamefont {Zhang},\ and\ \citenamefont {Zhu}}]{googlequantumai2023}%
  \BibitemOpen
  \bibfield  {author} {\bibinfo {author} {\bibnamefont {{Google Quantum AI}}}, \bibinfo {author} {\bibfnamefont {R.}~\bibnamefont {Acharya}}, \bibinfo {author} {\bibfnamefont {I.}~\bibnamefont {Aleiner}}, \bibinfo {author} {\bibfnamefont {R.}~\bibnamefont {Allen}}, \bibinfo {author} {\bibfnamefont {T.~I.}\ \bibnamefont {Andersen}}, \bibinfo {author} {\bibfnamefont {M.}~\bibnamefont {Ansmann}}, \bibinfo {author} {\bibfnamefont {F.}~\bibnamefont {Arute}}, \bibinfo {author} {\bibfnamefont {K.}~\bibnamefont {Arya}}, \bibinfo {author} {\bibfnamefont {A.}~\bibnamefont {Asfaw}}, \bibinfo {author} {\bibfnamefont {J.}~\bibnamefont {Atalaya}}, \bibinfo {author} {\bibfnamefont {R.}~\bibnamefont {Babbush}}, \bibinfo {author} {\bibfnamefont {D.}~\bibnamefont {Bacon}}, \bibinfo {author} {\bibfnamefont {J.~C.}\ \bibnamefont {Bardin}}, \bibinfo {author} {\bibfnamefont {J.}~\bibnamefont {Basso}}, \bibinfo {author} {\bibfnamefont {A.}~\bibnamefont {Bengtsson}}, \bibinfo {author} {\bibfnamefont {S.}~\bibnamefont {Boixo}},
  \bibinfo {author} {\bibfnamefont {G.}~\bibnamefont {Bortoli}}, \bibinfo {author} {\bibfnamefont {A.}~\bibnamefont {Bourassa}}, \bibinfo {author} {\bibfnamefont {J.}~\bibnamefont {Bovaird}}, \bibinfo {author} {\bibfnamefont {L.}~\bibnamefont {Brill}}, \bibinfo {author} {\bibfnamefont {M.}~\bibnamefont {Broughton}}, \bibinfo {author} {\bibfnamefont {B.~B.}\ \bibnamefont {Buckley}}, \bibinfo {author} {\bibfnamefont {D.~A.}\ \bibnamefont {Buell}}, \bibinfo {author} {\bibfnamefont {T.}~\bibnamefont {Burger}}, \bibinfo {author} {\bibfnamefont {B.}~\bibnamefont {Burkett}}, \bibinfo {author} {\bibfnamefont {N.}~\bibnamefont {Bushnell}}, \bibinfo {author} {\bibfnamefont {Y.}~\bibnamefont {Chen}}, \bibinfo {author} {\bibfnamefont {Z.}~\bibnamefont {Chen}}, \bibinfo {author} {\bibfnamefont {B.}~\bibnamefont {Chiaro}}, \bibinfo {author} {\bibfnamefont {J.}~\bibnamefont {Cogan}}, \bibinfo {author} {\bibfnamefont {R.}~\bibnamefont {Collins}}, \bibinfo {author} {\bibfnamefont {P.}~\bibnamefont {Conner}}, \bibinfo {author}
  {\bibfnamefont {W.}~\bibnamefont {Courtney}}, \bibinfo {author} {\bibfnamefont {A.~L.}\ \bibnamefont {Crook}}, \bibinfo {author} {\bibfnamefont {B.}~\bibnamefont {Curtin}}, \bibinfo {author} {\bibfnamefont {D.~M.}\ \bibnamefont {Debroy}}, \bibinfo {author} {\bibfnamefont {A.}~\bibnamefont {Del Toro~Barba}}, \bibinfo {author} {\bibfnamefont {S.}~\bibnamefont {Demura}}, \bibinfo {author} {\bibfnamefont {A.}~\bibnamefont {Dunsworth}}, \bibinfo {author} {\bibfnamefont {D.}~\bibnamefont {Eppens}}, \bibinfo {author} {\bibfnamefont {C.}~\bibnamefont {Erickson}}, \bibinfo {author} {\bibfnamefont {L.}~\bibnamefont {Faoro}}, \bibinfo {author} {\bibfnamefont {E.}~\bibnamefont {Farhi}}, \bibinfo {author} {\bibfnamefont {R.}~\bibnamefont {Fatemi}}, \bibinfo {author} {\bibfnamefont {L.}~\bibnamefont {Flores~Burgos}}, \bibinfo {author} {\bibfnamefont {E.}~\bibnamefont {Forati}}, \bibinfo {author} {\bibfnamefont {A.~G.}\ \bibnamefont {Fowler}}, \bibinfo {author} {\bibfnamefont {B.}~\bibnamefont {Foxen}}, \bibinfo {author}
  {\bibfnamefont {W.}~\bibnamefont {Giang}}, \bibinfo {author} {\bibfnamefont {C.}~\bibnamefont {Gidney}}, \bibinfo {author} {\bibfnamefont {D.}~\bibnamefont {Gilboa}}, \bibinfo {author} {\bibfnamefont {M.}~\bibnamefont {Giustina}}, \bibinfo {author} {\bibfnamefont {A.}~\bibnamefont {Grajales~Dau}}, \bibinfo {author} {\bibfnamefont {J.~A.}\ \bibnamefont {Gross}}, \bibinfo {author} {\bibfnamefont {S.}~\bibnamefont {Habegger}}, \bibinfo {author} {\bibfnamefont {M.~C.}\ \bibnamefont {Hamilton}}, \bibinfo {author} {\bibfnamefont {M.~P.}\ \bibnamefont {Harrigan}}, \bibinfo {author} {\bibfnamefont {S.~D.}\ \bibnamefont {Harrington}}, \bibinfo {author} {\bibfnamefont {O.}~\bibnamefont {Higgott}}, \bibinfo {author} {\bibfnamefont {J.}~\bibnamefont {Hilton}}, \bibinfo {author} {\bibfnamefont {M.}~\bibnamefont {Hoffmann}}, \bibinfo {author} {\bibfnamefont {S.}~\bibnamefont {Hong}}, \bibinfo {author} {\bibfnamefont {T.}~\bibnamefont {Huang}}, \bibinfo {author} {\bibfnamefont {A.}~\bibnamefont {Huff}}, \bibinfo {author}
  {\bibfnamefont {W.~J.}\ \bibnamefont {Huggins}}, \bibinfo {author} {\bibfnamefont {L.~B.}\ \bibnamefont {Ioffe}}, \bibinfo {author} {\bibfnamefont {S.~V.}\ \bibnamefont {Isakov}}, \bibinfo {author} {\bibfnamefont {J.}~\bibnamefont {Iveland}}, \bibinfo {author} {\bibfnamefont {E.}~\bibnamefont {Jeffrey}}, \bibinfo {author} {\bibfnamefont {Z.}~\bibnamefont {Jiang}}, \bibinfo {author} {\bibfnamefont {C.}~\bibnamefont {Jones}}, \bibinfo {author} {\bibfnamefont {P.}~\bibnamefont {Juhas}}, \bibinfo {author} {\bibfnamefont {D.}~\bibnamefont {Kafri}}, \bibinfo {author} {\bibfnamefont {K.}~\bibnamefont {Kechedzhi}}, \bibinfo {author} {\bibfnamefont {J.}~\bibnamefont {Kelly}}, \bibinfo {author} {\bibfnamefont {T.}~\bibnamefont {Khattar}}, \bibinfo {author} {\bibfnamefont {M.}~\bibnamefont {Khezri}}, \bibinfo {author} {\bibfnamefont {M.}~\bibnamefont {Kieferov{\'a}}}, \bibinfo {author} {\bibfnamefont {S.}~\bibnamefont {Kim}}, \bibinfo {author} {\bibfnamefont {A.}~\bibnamefont {Kitaev}}, \bibinfo {author}
  {\bibfnamefont {P.~V.}\ \bibnamefont {Klimov}}, \bibinfo {author} {\bibfnamefont {A.~R.}\ \bibnamefont {Klots}}, \bibinfo {author} {\bibfnamefont {A.~N.}\ \bibnamefont {Korotkov}}, \bibinfo {author} {\bibfnamefont {F.}~\bibnamefont {Kostritsa}}, \bibinfo {author} {\bibfnamefont {J.~M.}\ \bibnamefont {Kreikebaum}}, \bibinfo {author} {\bibfnamefont {D.}~\bibnamefont {Landhuis}}, \bibinfo {author} {\bibfnamefont {P.}~\bibnamefont {Laptev}}, \bibinfo {author} {\bibfnamefont {K.-M.}\ \bibnamefont {Lau}}, \bibinfo {author} {\bibfnamefont {L.}~\bibnamefont {Laws}}, \bibinfo {author} {\bibfnamefont {J.}~\bibnamefont {Lee}}, \bibinfo {author} {\bibfnamefont {K.}~\bibnamefont {Lee}}, \bibinfo {author} {\bibfnamefont {B.~J.}\ \bibnamefont {Lester}}, \bibinfo {author} {\bibfnamefont {A.}~\bibnamefont {Lill}}, \bibinfo {author} {\bibfnamefont {W.}~\bibnamefont {Liu}}, \bibinfo {author} {\bibfnamefont {A.}~\bibnamefont {Locharla}}, \bibinfo {author} {\bibfnamefont {E.}~\bibnamefont {Lucero}}, \bibinfo {author}
  {\bibfnamefont {F.~D.}\ \bibnamefont {Malone}}, \bibinfo {author} {\bibfnamefont {J.}~\bibnamefont {Marshall}}, \bibinfo {author} {\bibfnamefont {O.}~\bibnamefont {Martin}}, \bibinfo {author} {\bibfnamefont {J.~R.}\ \bibnamefont {McClean}}, \bibinfo {author} {\bibfnamefont {T.}~\bibnamefont {McCourt}}, \bibinfo {author} {\bibfnamefont {M.}~\bibnamefont {McEwen}}, \bibinfo {author} {\bibfnamefont {A.}~\bibnamefont {Megrant}}, \bibinfo {author} {\bibfnamefont {B.}~\bibnamefont {Meurer~Costa}}, \bibinfo {author} {\bibfnamefont {X.}~\bibnamefont {Mi}}, \bibinfo {author} {\bibfnamefont {K.~C.}\ \bibnamefont {Miao}}, \bibinfo {author} {\bibfnamefont {M.}~\bibnamefont {Mohseni}}, \bibinfo {author} {\bibfnamefont {S.}~\bibnamefont {Montazeri}}, \bibinfo {author} {\bibfnamefont {A.}~\bibnamefont {Morvan}}, \bibinfo {author} {\bibfnamefont {E.}~\bibnamefont {Mount}}, \bibinfo {author} {\bibfnamefont {W.}~\bibnamefont {Mruczkiewicz}}, \bibinfo {author} {\bibfnamefont {O.}~\bibnamefont {Naaman}}, \bibinfo {author}
  {\bibfnamefont {M.}~\bibnamefont {Neeley}}, \bibinfo {author} {\bibfnamefont {C.}~\bibnamefont {Neill}}, \bibinfo {author} {\bibfnamefont {A.}~\bibnamefont {Nersisyan}}, \bibinfo {author} {\bibfnamefont {H.}~\bibnamefont {Neven}}, \bibinfo {author} {\bibfnamefont {M.}~\bibnamefont {Newman}}, \bibinfo {author} {\bibfnamefont {J.~H.}\ \bibnamefont {Ng}}, \bibinfo {author} {\bibfnamefont {A.}~\bibnamefont {Nguyen}}, \bibinfo {author} {\bibfnamefont {M.}~\bibnamefont {Nguyen}}, \bibinfo {author} {\bibfnamefont {M.~Y.}\ \bibnamefont {Niu}}, \bibinfo {author} {\bibfnamefont {T.~E.}\ \bibnamefont {O'Brien}}, \bibinfo {author} {\bibfnamefont {A.}~\bibnamefont {Opremcak}}, \bibinfo {author} {\bibfnamefont {J.}~\bibnamefont {Platt}}, \bibinfo {author} {\bibfnamefont {A.}~\bibnamefont {Petukhov}}, \bibinfo {author} {\bibfnamefont {R.}~\bibnamefont {Potter}}, \bibinfo {author} {\bibfnamefont {L.~P.}\ \bibnamefont {Pryadko}}, \bibinfo {author} {\bibfnamefont {C.}~\bibnamefont {Quintana}}, \bibinfo {author}
  {\bibfnamefont {P.}~\bibnamefont {Roushan}}, \bibinfo {author} {\bibfnamefont {N.~C.}\ \bibnamefont {Rubin}}, \bibinfo {author} {\bibfnamefont {N.}~\bibnamefont {Saei}}, \bibinfo {author} {\bibfnamefont {D.}~\bibnamefont {Sank}}, \bibinfo {author} {\bibfnamefont {K.}~\bibnamefont {Sankaragomathi}}, \bibinfo {author} {\bibfnamefont {K.~J.}\ \bibnamefont {Satzinger}}, \bibinfo {author} {\bibfnamefont {H.~F.}\ \bibnamefont {Schurkus}}, \bibinfo {author} {\bibfnamefont {C.}~\bibnamefont {Schuster}}, \bibinfo {author} {\bibfnamefont {M.~J.}\ \bibnamefont {Shearn}}, \bibinfo {author} {\bibfnamefont {A.}~\bibnamefont {Shorter}}, \bibinfo {author} {\bibfnamefont {V.}~\bibnamefont {Shvarts}}, \bibinfo {author} {\bibfnamefont {J.}~\bibnamefont {Skruzny}}, \bibinfo {author} {\bibfnamefont {V.}~\bibnamefont {Smelyanskiy}}, \bibinfo {author} {\bibfnamefont {W.~C.}\ \bibnamefont {Smith}}, \bibinfo {author} {\bibfnamefont {G.}~\bibnamefont {Sterling}}, \bibinfo {author} {\bibfnamefont {D.}~\bibnamefont {Strain}}, \bibinfo
  {author} {\bibfnamefont {M.}~\bibnamefont {Szalay}}, \bibinfo {author} {\bibfnamefont {A.}~\bibnamefont {Torres}}, \bibinfo {author} {\bibfnamefont {G.}~\bibnamefont {Vidal}}, \bibinfo {author} {\bibfnamefont {B.}~\bibnamefont {Villalonga}}, \bibinfo {author} {\bibfnamefont {C.}~\bibnamefont {Vollgraff~Heidweiller}}, \bibinfo {author} {\bibfnamefont {T.}~\bibnamefont {White}}, \bibinfo {author} {\bibfnamefont {C.}~\bibnamefont {Xing}}, \bibinfo {author} {\bibfnamefont {Z.~J.}\ \bibnamefont {Yao}}, \bibinfo {author} {\bibfnamefont {P.}~\bibnamefont {Yeh}}, \bibinfo {author} {\bibfnamefont {J.}~\bibnamefont {Yoo}}, \bibinfo {author} {\bibfnamefont {G.}~\bibnamefont {Young}}, \bibinfo {author} {\bibfnamefont {A.}~\bibnamefont {Zalcman}}, \bibinfo {author} {\bibfnamefont {Y.}~\bibnamefont {Zhang}},\ and\ \bibinfo {author} {\bibfnamefont {N.}~\bibnamefont {Zhu}},\ }\bibfield  {title} {\bibinfo {title} {Suppressing quantum errors by scaling a surface code logical qubit},\ }\href
  {https://doi.org/10.1038/s41586-022-05434-1} {\bibfield  {journal} {\bibinfo  {journal} {Nature}\ }\textbf {\bibinfo {volume} {614}},\ \bibinfo {pages} {676} (\bibinfo {year} {2023})}\BibitemShut {NoStop}%
\bibitem [{\citenamefont {Kitaev}(2003)}]{kitaev2003}%
  \BibitemOpen
  \bibfield  {author} {\bibinfo {author} {\bibfnamefont {{\relax A.Yu}.}~\bibnamefont {Kitaev}},\ }\bibfield  {title} {\bibinfo {title} {Fault-tolerant quantum computation by anyons},\ }\href {https://doi.org/10.1016/S0003-4916(02)00018-0} {\bibfield  {journal} {\bibinfo  {journal} {Annals of Physics}\ }\textbf {\bibinfo {volume} {303}},\ \bibinfo {pages} {2} (\bibinfo {year} {2003})}\BibitemShut {NoStop}%
\bibitem [{\citenamefont {Dennis}\ \emph {et~al.}(2002)\citenamefont {Dennis}, \citenamefont {Kitaev}, \citenamefont {Landahl},\ and\ \citenamefont {Preskill}}]{dennis2002}%
  \BibitemOpen
  \bibfield  {author} {\bibinfo {author} {\bibfnamefont {E.}~\bibnamefont {Dennis}}, \bibinfo {author} {\bibfnamefont {A.}~\bibnamefont {Kitaev}}, \bibinfo {author} {\bibfnamefont {A.}~\bibnamefont {Landahl}},\ and\ \bibinfo {author} {\bibfnamefont {J.}~\bibnamefont {Preskill}},\ }\bibfield  {title} {\bibinfo {title} {Topological quantum memory},\ }\href {https://doi.org/10.1063/1.1499754} {\bibfield  {journal} {\bibinfo  {journal} {Journal of Mathematical Physics}\ }\textbf {\bibinfo {volume} {43}},\ \bibinfo {pages} {4452} (\bibinfo {year} {2002})}\BibitemShut {NoStop}%
\bibitem [{\citenamefont {Fowler}\ \emph {et~al.}(2012)\citenamefont {Fowler}, \citenamefont {Mariantoni}, \citenamefont {Martinis},\ and\ \citenamefont {Cleland}}]{fowler2012}%
  \BibitemOpen
  \bibfield  {author} {\bibinfo {author} {\bibfnamefont {A.~G.}\ \bibnamefont {Fowler}}, \bibinfo {author} {\bibfnamefont {M.}~\bibnamefont {Mariantoni}}, \bibinfo {author} {\bibfnamefont {J.~M.}\ \bibnamefont {Martinis}},\ and\ \bibinfo {author} {\bibfnamefont {A.~N.}\ \bibnamefont {Cleland}},\ }\bibfield  {title} {\bibinfo {title} {Surface codes: {{Towards}} practical large-scale quantum computation},\ }\href {https://doi.org/10.1103/PhysRevA.86.032324} {\bibfield  {journal} {\bibinfo  {journal} {Phys. Rev. A}\ }\textbf {\bibinfo {volume} {86}},\ \bibinfo {pages} {032324} (\bibinfo {year} {2012})}\BibitemShut {NoStop}%
\bibitem [{\citenamefont {Breuckmann}\ and\ \citenamefont {Eberhardt}(2021{\natexlab{a}})}]{breuckmann2021a}%
  \BibitemOpen
  \bibfield  {author} {\bibinfo {author} {\bibfnamefont {N.~P.}\ \bibnamefont {Breuckmann}}\ and\ \bibinfo {author} {\bibfnamefont {J.~N.}\ \bibnamefont {Eberhardt}},\ }\bibfield  {title} {\bibinfo {title} {Quantum {{Low-Density Parity-Check Codes}}},\ }\href {https://doi.org/10.1103/PRXQuantum.2.040101} {\bibfield  {journal} {\bibinfo  {journal} {PRX Quantum}\ }\textbf {\bibinfo {volume} {2}},\ \bibinfo {pages} {040101} (\bibinfo {year} {2021}{\natexlab{a}})}\BibitemShut {NoStop}%
\bibitem [{\citenamefont {Gidney}\ and\ \citenamefont {Eker{\aa}}(2021)}]{gidney2021}%
  \BibitemOpen
  \bibfield  {author} {\bibinfo {author} {\bibfnamefont {C.}~\bibnamefont {Gidney}}\ and\ \bibinfo {author} {\bibfnamefont {M.}~\bibnamefont {Eker{\aa}}},\ }\bibfield  {title} {\bibinfo {title} {How to factor 2048 bit {{RSA}} integers in 8 hours using 20 million noisy qubits},\ }\href {https://doi.org/10.22331/q-2021-04-15-433} {\bibfield  {journal} {\bibinfo  {journal} {Quantum}\ }\textbf {\bibinfo {volume} {5}},\ \bibinfo {pages} {433} (\bibinfo {year} {2021})}\BibitemShut {NoStop}%
\bibitem [{\citenamefont {Kim}\ \emph {et~al.}(2022)\citenamefont {Kim}, \citenamefont {Liu}, \citenamefont {Pallister}, \citenamefont {Pol}, \citenamefont {Roberts},\ and\ \citenamefont {Lee}}]{kim2022}%
  \BibitemOpen
  \bibfield  {author} {\bibinfo {author} {\bibfnamefont {I.~H.}\ \bibnamefont {Kim}}, \bibinfo {author} {\bibfnamefont {Y.-H.}\ \bibnamefont {Liu}}, \bibinfo {author} {\bibfnamefont {S.}~\bibnamefont {Pallister}}, \bibinfo {author} {\bibfnamefont {W.}~\bibnamefont {Pol}}, \bibinfo {author} {\bibfnamefont {S.}~\bibnamefont {Roberts}},\ and\ \bibinfo {author} {\bibfnamefont {E.}~\bibnamefont {Lee}},\ }\bibfield  {title} {\bibinfo {title} {Fault-tolerant resource estimate for quantum chemical simulations: {{Case}} study on {{Li-ion}} battery electrolyte molecules},\ }\href {https://doi.org/10.1103/PhysRevResearch.4.023019} {\bibfield  {journal} {\bibinfo  {journal} {Phys. Rev. Research}\ }\textbf {\bibinfo {volume} {4}},\ \bibinfo {pages} {023019} (\bibinfo {year} {2022})}\BibitemShut {NoStop}%
\bibitem [{\citenamefont {Baspin}\ and\ \citenamefont {Krishna}(2022{\natexlab{a}})}]{baspin2022}%
  \BibitemOpen
  \bibfield  {author} {\bibinfo {author} {\bibfnamefont {N.}~\bibnamefont {Baspin}}\ and\ \bibinfo {author} {\bibfnamefont {A.}~\bibnamefont {Krishna}},\ }\bibfield  {title} {\bibinfo {title} {Connectivity constrains quantum codes},\ }\href {https://doi.org/10.22331/q-2022-05-13-711} {\bibfield  {journal} {\bibinfo  {journal} {Quantum}\ }\textbf {\bibinfo {volume} {6}},\ \bibinfo {pages} {711} (\bibinfo {year} {2022}{\natexlab{a}})}\BibitemShut {NoStop}%
\bibitem [{\citenamefont {Baspin}\ and\ \citenamefont {Krishna}(2022{\natexlab{b}})}]{baspin2022a}%
  \BibitemOpen
  \bibfield  {author} {\bibinfo {author} {\bibfnamefont {N.}~\bibnamefont {Baspin}}\ and\ \bibinfo {author} {\bibfnamefont {A.}~\bibnamefont {Krishna}},\ }\bibfield  {title} {\bibinfo {title} {Quantifying {{Nonlocality}}: {{How Outperforming Local Quantum Codes Is Expensive}}},\ }\href {https://doi.org/10.1103/PhysRevLett.129.050505} {\bibfield  {journal} {\bibinfo  {journal} {Phys. Rev. Lett.}\ }\textbf {\bibinfo {volume} {129}},\ \bibinfo {pages} {050505} (\bibinfo {year} {2022}{\natexlab{b}})}\BibitemShut {NoStop}%
\bibitem [{\citenamefont {Baspin}\ \emph {et~al.}(2023)\citenamefont {Baspin}, \citenamefont {Guruswami}, \citenamefont {Krishna},\ and\ \citenamefont {Li}}]{baspin2023}%
  \BibitemOpen
  \bibfield  {author} {\bibinfo {author} {\bibfnamefont {N.}~\bibnamefont {Baspin}}, \bibinfo {author} {\bibfnamefont {V.}~\bibnamefont {Guruswami}}, \bibinfo {author} {\bibfnamefont {A.}~\bibnamefont {Krishna}},\ and\ \bibinfo {author} {\bibfnamefont {R.}~\bibnamefont {Li}},\ }\href@noop {} {\bibinfo {title} {Improved rate-distance trade-offs for quantum codes with restricted connectivity}} (\bibinfo {year} {2023}),\ \Eprint {https://arxiv.org/abs/2307.03283} {arxiv:2307.03283} \BibitemShut {NoStop}%
\bibitem [{\citenamefont {Bravyi}\ \emph {et~al.}(2023)\citenamefont {Bravyi}, \citenamefont {Cross}, \citenamefont {Gambetta}, \citenamefont {Maslov}, \citenamefont {Rall},\ and\ \citenamefont {Yoder}}]{bravyi2023high}%
  \BibitemOpen
  \bibfield  {author} {\bibinfo {author} {\bibfnamefont {S.}~\bibnamefont {Bravyi}}, \bibinfo {author} {\bibfnamefont {A.~W.}\ \bibnamefont {Cross}}, \bibinfo {author} {\bibfnamefont {J.~M.}\ \bibnamefont {Gambetta}}, \bibinfo {author} {\bibfnamefont {D.}~\bibnamefont {Maslov}}, \bibinfo {author} {\bibfnamefont {P.}~\bibnamefont {Rall}},\ and\ \bibinfo {author} {\bibfnamefont {T.~J.}\ \bibnamefont {Yoder}},\ }\href@noop {} {\bibinfo {title} {High-threshold and low-overhead fault-tolerant quantum memory}} (\bibinfo {year} {2023}),\ \Eprint {https://arxiv.org/abs/2308.07915} {arXiv:2308.07915} \BibitemShut {NoStop}%
\bibitem [{\citenamefont {Xu}\ \emph {et~al.}(2023)\citenamefont {Xu}, \citenamefont {Ataides}, \citenamefont {Pattison}, \citenamefont {Raveendran}, \citenamefont {Bluvstein}, \citenamefont {Wurtz}, \citenamefont {Vasic}, \citenamefont {Lukin}, \citenamefont {Jiang},\ and\ \citenamefont {Zhou}}]{xu2023constant}%
  \BibitemOpen
  \bibfield  {author} {\bibinfo {author} {\bibfnamefont {Q.}~\bibnamefont {Xu}}, \bibinfo {author} {\bibfnamefont {J.}~\bibnamefont {Ataides}}, \bibinfo {author} {\bibfnamefont {C.~A.}\ \bibnamefont {Pattison}}, \bibinfo {author} {\bibfnamefont {N.}~\bibnamefont {Raveendran}}, \bibinfo {author} {\bibfnamefont {D.}~\bibnamefont {Bluvstein}}, \bibinfo {author} {\bibfnamefont {J.}~\bibnamefont {Wurtz}}, \bibinfo {author} {\bibfnamefont {B.}~\bibnamefont {Vasic}}, \bibinfo {author} {\bibfnamefont {M.~D.}\ \bibnamefont {Lukin}}, \bibinfo {author} {\bibfnamefont {L.}~\bibnamefont {Jiang}},\ and\ \bibinfo {author} {\bibfnamefont {H.}~\bibnamefont {Zhou}},\ }\href@noop {} {\bibinfo {title} {Constant-overhead fault-tolerant quantum computation with reconfigurable atom arrays}} (\bibinfo {year} {2023}),\ \Eprint {https://arxiv.org/abs/2308.08648} {arXiv:2308.08648} \BibitemShut {NoStop}%
\bibitem [{\citenamefont {Bourassa}\ \emph {et~al.}(2021)\citenamefont {Bourassa}, \citenamefont {Alexander}, \citenamefont {Vasmer}, \citenamefont {Patil}, \citenamefont {Tzitrin}, \citenamefont {Matsuura}, \citenamefont {Su}, \citenamefont {Baragiola}, \citenamefont {Guha}, \citenamefont {Dauphinais}, \citenamefont {Sabapathy}, \citenamefont {Menicucci},\ and\ \citenamefont {Dhand}}]{bourassa2021Blueprint}%
  \BibitemOpen
  \bibfield  {author} {\bibinfo {author} {\bibfnamefont {J.~E.}\ \bibnamefont {Bourassa}}, \bibinfo {author} {\bibfnamefont {R.~N.}\ \bibnamefont {Alexander}}, \bibinfo {author} {\bibfnamefont {M.}~\bibnamefont {Vasmer}}, \bibinfo {author} {\bibfnamefont {A.}~\bibnamefont {Patil}}, \bibinfo {author} {\bibfnamefont {I.}~\bibnamefont {Tzitrin}}, \bibinfo {author} {\bibfnamefont {T.}~\bibnamefont {Matsuura}}, \bibinfo {author} {\bibfnamefont {D.}~\bibnamefont {Su}}, \bibinfo {author} {\bibfnamefont {B.~Q.}\ \bibnamefont {Baragiola}}, \bibinfo {author} {\bibfnamefont {S.}~\bibnamefont {Guha}}, \bibinfo {author} {\bibfnamefont {G.}~\bibnamefont {Dauphinais}}, \bibinfo {author} {\bibfnamefont {K.~K.}\ \bibnamefont {Sabapathy}}, \bibinfo {author} {\bibfnamefont {N.~C.}\ \bibnamefont {Menicucci}},\ and\ \bibinfo {author} {\bibfnamefont {I.}~\bibnamefont {Dhand}},\ }\bibfield  {title} {\bibinfo {title} {Blueprint for a {S}calable {P}hotonic {F}ault-{T}olerant {Q}uantum {C}omputer},\ }\href
  {https://doi.org/10.22331/q-2021-02-04-392} {\bibfield  {journal} {\bibinfo  {journal} {{Quantum}}\ }\textbf {\bibinfo {volume} {5}},\ \bibinfo {pages} {392} (\bibinfo {year} {2021})}\BibitemShut {NoStop}%
\bibitem [{\citenamefont {Tzitrin}\ \emph {et~al.}(2021)\citenamefont {Tzitrin}, \citenamefont {Matsuura}, \citenamefont {Alexander}, \citenamefont {Dauphinais}, \citenamefont {Bourassa}, \citenamefont {Sabapathy}, \citenamefont {Menicucci},\ and\ \citenamefont {Dhand}}]{tzitrin2021Fault}%
  \BibitemOpen
  \bibfield  {author} {\bibinfo {author} {\bibfnamefont {I.}~\bibnamefont {Tzitrin}}, \bibinfo {author} {\bibfnamefont {T.}~\bibnamefont {Matsuura}}, \bibinfo {author} {\bibfnamefont {R.~N.}\ \bibnamefont {Alexander}}, \bibinfo {author} {\bibfnamefont {G.}~\bibnamefont {Dauphinais}}, \bibinfo {author} {\bibfnamefont {J.~E.}\ \bibnamefont {Bourassa}}, \bibinfo {author} {\bibfnamefont {K.~K.}\ \bibnamefont {Sabapathy}}, \bibinfo {author} {\bibfnamefont {N.~C.}\ \bibnamefont {Menicucci}},\ and\ \bibinfo {author} {\bibfnamefont {I.}~\bibnamefont {Dhand}},\ }\bibfield  {title} {\bibinfo {title} {Fault-tolerant quantum computation with static linear optics},\ }\href {https://doi.org/10.1103/PRXQuantum.2.040353} {\bibfield  {journal} {\bibinfo  {journal} {PRX Quantum}\ }\textbf {\bibinfo {volume} {2}},\ \bibinfo {pages} {040353} (\bibinfo {year} {2021})}\BibitemShut {NoStop}%
\bibitem [{\citenamefont {Hastings}(2016)}]{hastings2016}%
  \BibitemOpen
  \bibfield  {author} {\bibinfo {author} {\bibfnamefont {M.~B.}\ \bibnamefont {Hastings}},\ }\href@noop {} {\bibinfo {title} {Weight {{Reduction}} for {{Quantum Codes}}}} (\bibinfo {year} {2016}),\ \Eprint {https://arxiv.org/abs/1611.03790} {arxiv:1611.03790} \BibitemShut {NoStop}%
\bibitem [{\citenamefont {Hastings}(2021)}]{hastings2021quantum}%
  \BibitemOpen
  \bibfield  {author} {\bibinfo {author} {\bibfnamefont {M.~B.}\ \bibnamefont {Hastings}},\ }\href@noop {} {\bibinfo {title} {On quantum weight reduction}} (\bibinfo {year} {2021}),\ \Eprint {https://arxiv.org/abs/2102.10030} {arXiv:2102.10030} \BibitemShut {NoStop}%
\bibitem [{\citenamefont {Wills}\ \emph {et~al.}(2023{\natexlab{a}})\citenamefont {Wills}, \citenamefont {Lin},\ and\ \citenamefont {Hsieh}}]{wills2023tradeoff}%
  \BibitemOpen
  \bibfield  {author} {\bibinfo {author} {\bibfnamefont {A.}~\bibnamefont {Wills}}, \bibinfo {author} {\bibfnamefont {T.-C.}\ \bibnamefont {Lin}},\ and\ \bibinfo {author} {\bibfnamefont {M.-H.}\ \bibnamefont {Hsieh}},\ }\href@noop {} {\bibinfo {title} {Tradeoff constructions for quantum locally testable codes}} (\bibinfo {year} {2023}{\natexlab{a}}),\ \Eprint {https://arxiv.org/abs/2309.05541} {arXiv:2309.05541} \BibitemShut {NoStop}%
\bibitem [{\citenamefont {Tillich}\ and\ \citenamefont {Zemor}(2014)}]{tillich2014quantum}%
  \BibitemOpen
  \bibfield  {author} {\bibinfo {author} {\bibfnamefont {J.-P.}\ \bibnamefont {Tillich}}\ and\ \bibinfo {author} {\bibfnamefont {G.}~\bibnamefont {Zemor}},\ }\bibfield  {title} {\bibinfo {title} {Quantum {{LDPC Codes With Positive Rate}} and {{Minimum Distance Proportional}} to the {{Square Root}} of the {{Blocklength}}},\ }\href {https://doi.org/10.1109/TIT.2013.2292061} {\bibfield  {journal} {\bibinfo  {journal} {IEEE Trans. Inform. Theory}\ }\textbf {\bibinfo {volume} {60}},\ \bibinfo {pages} {1193} (\bibinfo {year} {2014})}\BibitemShut {NoStop}%
\bibitem [{\citenamefont {Panteleev}\ and\ \citenamefont {Kalachev}(2022{\natexlab{a}})}]{panteleev2022}%
  \BibitemOpen
  \bibfield  {author} {\bibinfo {author} {\bibfnamefont {P.}~\bibnamefont {Panteleev}}\ and\ \bibinfo {author} {\bibfnamefont {G.}~\bibnamefont {Kalachev}},\ }\bibfield  {title} {\bibinfo {title} {Quantum {{LDPC Codes With Almost Linear Minimum Distance}}},\ }\href {https://doi.org/10.1109/TIT.2021.3119384} {\bibfield  {journal} {\bibinfo  {journal} {IEEE Trans. Inform. Theory}\ }\textbf {\bibinfo {volume} {68}},\ \bibinfo {pages} {213} (\bibinfo {year} {2022}{\natexlab{a}})}\BibitemShut {NoStop}%
\bibitem [{\citenamefont {Breuckmann}\ and\ \citenamefont {Eberhardt}(2021{\natexlab{b}})}]{breuckmann2021}%
  \BibitemOpen
  \bibfield  {author} {\bibinfo {author} {\bibfnamefont {N.~P.}\ \bibnamefont {Breuckmann}}\ and\ \bibinfo {author} {\bibfnamefont {J.~N.}\ \bibnamefont {Eberhardt}},\ }\bibfield  {title} {\bibinfo {title} {Balanced {{Product Quantum Codes}}},\ }\href {https://doi.org/10.1109/TIT.2021.3097347} {\bibfield  {journal} {\bibinfo  {journal} {IEEE Trans. Inform. Theory}\ }\textbf {\bibinfo {volume} {67}},\ \bibinfo {pages} {6653} (\bibinfo {year} {2021}{\natexlab{b}})}\BibitemShut {NoStop}%
\bibitem [{\citenamefont {Panteleev}\ and\ \citenamefont {Kalachev}(2022{\natexlab{b}})}]{panteleev2022asymptotically}%
  \BibitemOpen
  \bibfield  {author} {\bibinfo {author} {\bibfnamefont {P.}~\bibnamefont {Panteleev}}\ and\ \bibinfo {author} {\bibfnamefont {G.}~\bibnamefont {Kalachev}},\ }\bibfield  {title} {\bibinfo {title} {Asymptotically good quantum and locally testable classical {LDPC} codes},\ }in\ \href {https://doi.org/10.1145/3519935.3520017} {\emph {\bibinfo {booktitle} {Proceedings of the 54th Annual ACM SIGACT Symposium on Theory of Computing}}}\ (\bibinfo {year} {2022})\ pp.\ \bibinfo {pages} {375--388}\BibitemShut {NoStop}%
\bibitem [{\citenamefont {Nielsen}\ and\ \citenamefont {Chuang}(2000)}]{nielsen00}%
  \BibitemOpen
  \bibfield  {author} {\bibinfo {author} {\bibfnamefont {M.~A.}\ \bibnamefont {Nielsen}}\ and\ \bibinfo {author} {\bibfnamefont {I.~L.}\ \bibnamefont {Chuang}},\ }\href@noop {} {\emph {\bibinfo {title} {Quantum Computation and Quantum Information}}}\ (\bibinfo  {publisher} {Cambridge University Press},\ \bibinfo {year} {2000})\BibitemShut {NoStop}%
\bibitem [{\citenamefont {Raussendorf}\ \emph {et~al.}(2006)\citenamefont {Raussendorf}, \citenamefont {Harrington},\ and\ \citenamefont {Goyal}}]{raussendorf2006Afault}%
  \BibitemOpen
  \bibfield  {author} {\bibinfo {author} {\bibfnamefont {R.}~\bibnamefont {Raussendorf}}, \bibinfo {author} {\bibfnamefont {J.}~\bibnamefont {Harrington}},\ and\ \bibinfo {author} {\bibfnamefont {K.}~\bibnamefont {Goyal}},\ }\bibfield  {title} {\bibinfo {title} {A fault-tolerant one-way quantum computer},\ }\href {https://doi.org/https://doi.org/10.1016/j.aop.2006.01.012} {\bibfield  {journal} {\bibinfo  {journal} {Annals of Physics}\ }\textbf {\bibinfo {volume} {321}},\ \bibinfo {pages} {2242} (\bibinfo {year} {2006})}\BibitemShut {NoStop}%
\bibitem [{\citenamefont {Bolt}\ \emph {et~al.}(2016)\citenamefont {Bolt}, \citenamefont {Duclos-Cianci}, \citenamefont {Poulin},\ and\ \citenamefont {Stace}}]{bolt2016Foliation}%
  \BibitemOpen
  \bibfield  {author} {\bibinfo {author} {\bibfnamefont {A.}~\bibnamefont {Bolt}}, \bibinfo {author} {\bibfnamefont {G.}~\bibnamefont {Duclos-Cianci}}, \bibinfo {author} {\bibfnamefont {D.}~\bibnamefont {Poulin}},\ and\ \bibinfo {author} {\bibfnamefont {T.~M.}\ \bibnamefont {Stace}},\ }\bibfield  {title} {\bibinfo {title} {Foliated quantum error-correcting codes},\ }\href {https://doi.org/10.1103/PhysRevLett.117.070501} {\bibfield  {journal} {\bibinfo  {journal} {Phys. Rev. Lett.}\ }\textbf {\bibinfo {volume} {117}},\ \bibinfo {pages} {070501} (\bibinfo {year} {2016})}\BibitemShut {NoStop}%
\bibitem [{\citenamefont {Brown}\ and\ \citenamefont {Roberts}(2020)}]{roberts2020Universal}%
  \BibitemOpen
  \bibfield  {author} {\bibinfo {author} {\bibfnamefont {B.~J.}\ \bibnamefont {Brown}}\ and\ \bibinfo {author} {\bibfnamefont {S.}~\bibnamefont {Roberts}},\ }\bibfield  {title} {\bibinfo {title} {Universal fault-tolerant measurement-based quantum computation},\ }\href {https://doi.org/10.1103/PhysRevResearch.2.033305} {\bibfield  {journal} {\bibinfo  {journal} {Phys. Rev. Res.}\ }\textbf {\bibinfo {volume} {2}},\ \bibinfo {pages} {033305} (\bibinfo {year} {2020})}\BibitemShut {NoStop}%
\bibitem [{\citenamefont {Gottesman}\ \emph {et~al.}(2001)\citenamefont {Gottesman}, \citenamefont {Kitaev},\ and\ \citenamefont {Preskill}}]{gottesman2001Encoding}%
  \BibitemOpen
  \bibfield  {author} {\bibinfo {author} {\bibfnamefont {D.}~\bibnamefont {Gottesman}}, \bibinfo {author} {\bibfnamefont {A.}~\bibnamefont {Kitaev}},\ and\ \bibinfo {author} {\bibfnamefont {J.}~\bibnamefont {Preskill}},\ }\bibfield  {title} {\bibinfo {title} {Encoding a qubit in an oscillator},\ }\href {https://doi.org/10.1103/PhysRevA.64.012310} {\bibfield  {journal} {\bibinfo  {journal} {Phys. Rev. A}\ }\textbf {\bibinfo {volume} {64}},\ \bibinfo {pages} {012310} (\bibinfo {year} {2001})}\BibitemShut {NoStop}%
\bibitem [{\citenamefont {Xanadu}(2024)}]{xanadu2024inprep}%
  \BibitemOpen
  \bibfield  {author} {\bibinfo {author} {\bibnamefont {Xanadu}}} (\bibinfo {year} {2024}),\ \bibinfo {note} {in preparation}\BibitemShut {NoStop}%
\bibitem [{\citenamefont {Gottesman}(1997)}]{gottesman1997Stabilizer}%
  \BibitemOpen
  \bibfield  {author} {\bibinfo {author} {\bibfnamefont {D.}~\bibnamefont {Gottesman}},\ }\emph {\bibinfo {title} {Stabilizer Codes and Quantum Error Correction}},\ \href {https://doi.org/https://doi.org/10.48550/arXiv.quant-ph/9705052} {Ph.D. thesis},\ \bibinfo  {school} {Caltech} (\bibinfo {year} {1997})\BibitemShut {NoStop}%
\bibitem [{\citenamefont {MacWilliams}\ and\ \citenamefont {Sloane}(1978)}]{macwilliams1978Theory}%
  \BibitemOpen
  \bibfield  {author} {\bibinfo {author} {\bibfnamefont {F.}~\bibnamefont {MacWilliams}}\ and\ \bibinfo {author} {\bibfnamefont {N.}~\bibnamefont {Sloane}},\ }\href@noop {} {\emph {\bibinfo {title} {The Theory of Error-Correcting Codes}}},\ \bibinfo {edition} {2nd}\ ed.\ (\bibinfo  {publisher} {North-holland Publishing Company},\ \bibinfo {year} {1978})\BibitemShut {NoStop}%
\bibitem [{\citenamefont {Calderbank}\ and\ \citenamefont {Shor}(1996)}]{calderbank1996Good}%
  \BibitemOpen
  \bibfield  {author} {\bibinfo {author} {\bibfnamefont {A.~R.}\ \bibnamefont {Calderbank}}\ and\ \bibinfo {author} {\bibfnamefont {P.~W.}\ \bibnamefont {Shor}},\ }\bibfield  {title} {\bibinfo {title} {Good quantum error-correcting codes exist},\ }\href {https://doi.org/10.1103/PhysRevA.54.1098} {\bibfield  {journal} {\bibinfo  {journal} {Phys. Rev. A}\ }\textbf {\bibinfo {volume} {54}},\ \bibinfo {pages} {1098} (\bibinfo {year} {1996})}\BibitemShut {NoStop}%
\bibitem [{\citenamefont {Steane}(1996)}]{steane1996Multiple}%
  \BibitemOpen
  \bibfield  {author} {\bibinfo {author} {\bibfnamefont {A.}~\bibnamefont {Steane}},\ }\bibfield  {title} {\bibinfo {title} {Multiple-particle interference and quantum error correction},\ }\href {https://doi.org/10.1098/rspa.1996.0136} {\bibfield  {journal} {\bibinfo  {journal} {Proc. R. Soc. Lond. A.}\ }\textbf {\bibinfo {volume} {452}},\ \bibinfo {pages} {2551–} (\bibinfo {year} {1996})}\BibitemShut {NoStop}%
\bibitem [{\citenamefont {Hastings}\ \emph {et~al.}(2021)\citenamefont {Hastings}, \citenamefont {Haah},\ and\ \citenamefont {O'Donnell}}]{hastings2021fiber}%
  \BibitemOpen
  \bibfield  {author} {\bibinfo {author} {\bibfnamefont {M.~B.}\ \bibnamefont {Hastings}}, \bibinfo {author} {\bibfnamefont {J.}~\bibnamefont {Haah}},\ and\ \bibinfo {author} {\bibfnamefont {R.}~\bibnamefont {O'Donnell}},\ }\bibfield  {title} {\bibinfo {title} {Fiber bundle codes: breaking the n1/2 polylog(n) barrier for quantum {LDPC} codes},\ }in\ \href {https://doi.org/10.1145/3406325.3451005} {\emph {\bibinfo {booktitle} {Proceedings of the 53rd Annual ACM SIGACT Symposium on Theory of Computing}}}\ (\bibinfo {year} {2021})\ pp.\ \bibinfo {pages} {1276--1288}\BibitemShut {NoStop}%
\bibitem [{\citenamefont {Chen}\ \emph {et~al.}(1969)\citenamefont {Chen}, \citenamefont {Peterson},\ and\ \citenamefont {Weldon}}]{chen1969some}%
  \BibitemOpen
  \bibfield  {author} {\bibinfo {author} {\bibfnamefont {C.}~\bibnamefont {Chen}}, \bibinfo {author} {\bibfnamefont {W.}~\bibnamefont {Peterson}},\ and\ \bibinfo {author} {\bibfnamefont {E.}~\bibnamefont {Weldon}},\ }\bibfield  {title} {\bibinfo {title} {Some results on quasi-cyclic codes},\ }\href {https://doi.org/https://doi.org/10.1016/S0019-9958(69)90497-5} {\bibfield  {journal} {\bibinfo  {journal} {Information and Control}\ }\textbf {\bibinfo {volume} {15}},\ \bibinfo {pages} {407} (\bibinfo {year} {1969})}\BibitemShut {NoStop}%
\bibitem [{\citenamefont {Panteleev}\ and\ \citenamefont {Kalachev}(2021)}]{panteleev2021quantum}%
  \BibitemOpen
  \bibfield  {author} {\bibinfo {author} {\bibfnamefont {P.}~\bibnamefont {Panteleev}}\ and\ \bibinfo {author} {\bibfnamefont {G.}~\bibnamefont {Kalachev}},\ }\bibfield  {title} {\bibinfo {title} {Quantum {LDPC} codes with almost linear minimum distance},\ }\href {https://doi.org/10.1109/TIT.2021.3119384} {\bibfield  {journal} {\bibinfo  {journal} {IEEE Transactions on Information Theory}\ }\textbf {\bibinfo {volume} {68}},\ \bibinfo {pages} {213} (\bibinfo {year} {2021})}\BibitemShut {NoStop}%
\bibitem [{\citenamefont {Raveendran}\ \emph {et~al.}(2022)\citenamefont {Raveendran}, \citenamefont {Rengaswamy}, \citenamefont {Rozp{\k{e}}dek}, \citenamefont {Raina}, \citenamefont {Jiang},\ and\ \citenamefont {Vasi{\'{c}}}}]{raveendran2022finite}%
  \BibitemOpen
  \bibfield  {author} {\bibinfo {author} {\bibfnamefont {N.}~\bibnamefont {Raveendran}}, \bibinfo {author} {\bibfnamefont {N.}~\bibnamefont {Rengaswamy}}, \bibinfo {author} {\bibfnamefont {F.}~\bibnamefont {Rozp{\k{e}}dek}}, \bibinfo {author} {\bibfnamefont {A.}~\bibnamefont {Raina}}, \bibinfo {author} {\bibfnamefont {L.}~\bibnamefont {Jiang}},\ and\ \bibinfo {author} {\bibfnamefont {B.}~\bibnamefont {Vasi{\'{c}}}},\ }\bibfield  {title} {\bibinfo {title} {Finite {R}ate {QLDPC}-{GKP} {C}oding {S}cheme that {S}urpasses the {CSS} {H}amming {B}ound},\ }\href {https://doi.org/10.22331/q-2022-07-20-767} {\bibfield  {journal} {\bibinfo  {journal} {{Quantum}}\ }\textbf {\bibinfo {volume} {6}},\ \bibinfo {pages} {767} (\bibinfo {year} {2022})}\BibitemShut {NoStop}%
\bibitem [{\citenamefont {Roffe}\ \emph {et~al.}(2023)\citenamefont {Roffe}, \citenamefont {Cohen}, \citenamefont {Quintavalle}, \citenamefont {Chandra},\ and\ \citenamefont {Campbell}}]{roffe2023bias}%
  \BibitemOpen
  \bibfield  {author} {\bibinfo {author} {\bibfnamefont {J.}~\bibnamefont {Roffe}}, \bibinfo {author} {\bibfnamefont {L.~Z.}\ \bibnamefont {Cohen}}, \bibinfo {author} {\bibfnamefont {A.~O.}\ \bibnamefont {Quintavalle}}, \bibinfo {author} {\bibfnamefont {D.}~\bibnamefont {Chandra}},\ and\ \bibinfo {author} {\bibfnamefont {E.~T.}\ \bibnamefont {Campbell}},\ }\bibfield  {title} {\bibinfo {title} {Bias-tailored quantum {LDPC} codes},\ }\href {https://doi.org/10.22331/q-2023-05-15-1005} {\bibfield  {journal} {\bibinfo  {journal} {{Quantum}}\ }\textbf {\bibinfo {volume} {7}},\ \bibinfo {pages} {1005} (\bibinfo {year} {2023})}\BibitemShut {NoStop}%
\bibitem [{\citenamefont {Sabo}(2021)}]{Sabo_2021}%
  \BibitemOpen
  \bibfield  {author} {\bibinfo {author} {\bibfnamefont {E.}~\bibnamefont {Sabo}},\ }\href {https://github.com/esabo/CodingTheory} {\bibinfo {title} {esabo/codingtheory: A basic coding theory library for julia.}} (\bibinfo {year} {2021})\BibitemShut {NoStop}%
\bibitem [{\citenamefont {Bravyi}\ and\ \citenamefont {Terhal}(2009)}]{bravyi2009no}%
  \BibitemOpen
  \bibfield  {author} {\bibinfo {author} {\bibfnamefont {S.}~\bibnamefont {Bravyi}}\ and\ \bibinfo {author} {\bibfnamefont {B.}~\bibnamefont {Terhal}},\ }\bibfield  {title} {\bibinfo {title} {A no-go theorem for a two-dimensional self-correcting quantum memory based on stabilizer codes},\ }\href {https://doi.org/10.1088/1367-2630/11/4/043029} {\bibfield  {journal} {\bibinfo  {journal} {New Journal of Physics}\ }\textbf {\bibinfo {volume} {11}},\ \bibinfo {pages} {043029} (\bibinfo {year} {2009})}\BibitemShut {NoStop}%
\bibitem [{\citenamefont {Campbell}(2019)}]{campbell2019theory}%
  \BibitemOpen
  \bibfield  {author} {\bibinfo {author} {\bibfnamefont {E.~T.}\ \bibnamefont {Campbell}},\ }\bibfield  {title} {\bibinfo {title} {A theory of single-shot error correction for adversarial noise},\ }\href {https://doi.org/10.1088/2058-9565/aafc8f} {\bibfield  {journal} {\bibinfo  {journal} {Quantum Science and Technology}\ }\textbf {\bibinfo {volume} {4}},\ \bibinfo {pages} {025006} (\bibinfo {year} {2019})}\BibitemShut {NoStop}%
\bibitem [{\citenamefont {Zeng}\ and\ \citenamefont {Pryadko}(2019)}]{pryadko2019higher}%
  \BibitemOpen
  \bibfield  {author} {\bibinfo {author} {\bibfnamefont {W.}~\bibnamefont {Zeng}}\ and\ \bibinfo {author} {\bibfnamefont {L.~P.}\ \bibnamefont {Pryadko}},\ }\bibfield  {title} {\bibinfo {title} {Higher-dimensional quantum hypergraph-product codes with finite rates},\ }\href {https://doi.org/10.1103/PhysRevLett.122.230501} {\bibfield  {journal} {\bibinfo  {journal} {Phys. Rev. Lett.}\ }\textbf {\bibinfo {volume} {122}},\ \bibinfo {pages} {230501} (\bibinfo {year} {2019})}\BibitemShut {NoStop}%
\bibitem [{\citenamefont {Quintavalle}\ \emph {et~al.}(2021)\citenamefont {Quintavalle}, \citenamefont {Vasmer}, \citenamefont {Roffe},\ and\ \citenamefont {Campbell}}]{quintavalle2021single}%
  \BibitemOpen
  \bibfield  {author} {\bibinfo {author} {\bibfnamefont {A.~O.}\ \bibnamefont {Quintavalle}}, \bibinfo {author} {\bibfnamefont {M.}~\bibnamefont {Vasmer}}, \bibinfo {author} {\bibfnamefont {J.}~\bibnamefont {Roffe}},\ and\ \bibinfo {author} {\bibfnamefont {E.~T.}\ \bibnamefont {Campbell}},\ }\bibfield  {title} {\bibinfo {title} {Single-shot error correction of three-dimensional homological product codes},\ }\href {https://doi.org/10.1103/PRXQuantum.2.020340} {\bibfield  {journal} {\bibinfo  {journal} {PRX Quantum}\ }\textbf {\bibinfo {volume} {2}},\ \bibinfo {pages} {020340} (\bibinfo {year} {2021})}\BibitemShut {NoStop}%
\bibitem [{\citenamefont {Knill}\ \emph {et~al.}(1996)\citenamefont {Knill}, \citenamefont {Laflamme},\ and\ \citenamefont {Zurek}}]{knill1996threshold}%
  \BibitemOpen
  \bibfield  {author} {\bibinfo {author} {\bibfnamefont {E.}~\bibnamefont {Knill}}, \bibinfo {author} {\bibfnamefont {R.}~\bibnamefont {Laflamme}},\ and\ \bibinfo {author} {\bibfnamefont {W.}~\bibnamefont {Zurek}},\ }\href@noop {} {\bibinfo {title} {Threshold accuracy for quantum computation}} (\bibinfo {year} {1996}),\ \Eprint {https://arxiv.org/abs/quant-ph/9610011} {arXiv:quant-ph/9610011} \BibitemShut {NoStop}%
\bibitem [{\citenamefont {Anderson}\ \emph {et~al.}(2014)\citenamefont {Anderson}, \citenamefont {Duclos-Cianci},\ and\ \citenamefont {Poulin}}]{anderson2014fault}%
  \BibitemOpen
  \bibfield  {author} {\bibinfo {author} {\bibfnamefont {J.~T.}\ \bibnamefont {Anderson}}, \bibinfo {author} {\bibfnamefont {G.}~\bibnamefont {Duclos-Cianci}},\ and\ \bibinfo {author} {\bibfnamefont {D.}~\bibnamefont {Poulin}},\ }\bibfield  {title} {\bibinfo {title} {Fault-tolerant conversion between the {Steane} and {Reed-Muller} quantum codes},\ }\href {https://doi.org/10.1103/PhysRevLett.113.080501} {\bibfield  {journal} {\bibinfo  {journal} {Phys. Rev. Lett.}\ }\textbf {\bibinfo {volume} {113}},\ \bibinfo {pages} {080501} (\bibinfo {year} {2014})}\BibitemShut {NoStop}%
\bibitem [{\citenamefont {Freedman}\ and\ \citenamefont {Hastings}(2021)}]{freedman2021building}%
  \BibitemOpen
  \bibfield  {author} {\bibinfo {author} {\bibfnamefont {M.}~\bibnamefont {Freedman}}\ and\ \bibinfo {author} {\bibfnamefont {M.}~\bibnamefont {Hastings}},\ }\bibfield  {title} {\bibinfo {title} {Building manifolds from quantum codes},\ }\href {https://doi.org/10.1007/s00039-021-00567-3} {\bibfield  {journal} {\bibinfo  {journal} {Geometric and Functional Analysis}\ }\textbf {\bibinfo {volume} {31}},\ \bibinfo {pages} {855} (\bibinfo {year} {2021})}\BibitemShut {NoStop}%
\bibitem [{\citenamefont {Horton}(1987)}]{horton1987polynomial}%
  \BibitemOpen
  \bibfield  {author} {\bibinfo {author} {\bibfnamefont {J.~D.}\ \bibnamefont {Horton}},\ }\bibfield  {title} {\bibinfo {title} {A polynomial-time algorithm to find the shortest cycle basis of a graph},\ }\href@noop {} {\bibfield  {journal} {\bibinfo  {journal} {SIAM Journal on Computing}\ }\textbf {\bibinfo {volume} {16}},\ \bibinfo {pages} {358} (\bibinfo {year} {1987})}\BibitemShut {NoStop}%
\bibitem [{\citenamefont {Wills}\ \emph {et~al.}(2023{\natexlab{b}})\citenamefont {Wills}, \citenamefont {Lin},\ and\ \citenamefont {Hsieh}}]{wills2023}%
  \BibitemOpen
  \bibfield  {author} {\bibinfo {author} {\bibfnamefont {A.}~\bibnamefont {Wills}}, \bibinfo {author} {\bibfnamefont {T.-C.}\ \bibnamefont {Lin}},\ and\ \bibinfo {author} {\bibfnamefont {M.-H.}\ \bibnamefont {Hsieh}},\ }\href@noop {} {\bibinfo {title} {Tradeoff {{Constructions}} for {{Quantum Locally Testable Codes}}}} (\bibinfo {year} {2023}{\natexlab{b}}),\ \Eprint {https://arxiv.org/abs/2309.05541} {arxiv:2309.05541} \BibitemShut {NoStop}%
\bibitem [{\citenamefont {Ryan}\ and\ \citenamefont {Lin}(2009)}]{ryan2009channel}%
  \BibitemOpen
  \bibfield  {author} {\bibinfo {author} {\bibfnamefont {W.}~\bibnamefont {Ryan}}\ and\ \bibinfo {author} {\bibfnamefont {S.}~\bibnamefont {Lin}},\ }\href@noop {} {\emph {\bibinfo {title} {{Channel Codes: Classical and Modern}}}}\ (\bibinfo  {publisher} {Cambridge University Press},\ \bibinfo {year} {2009})\BibitemShut {NoStop}%
\bibitem [{\citenamefont {Shankar}(2005)}]{shankar2005expander}%
  \BibitemOpen
  \bibfield  {author} {\bibinfo {author} {\bibfnamefont {P.}~\bibnamefont {Shankar}},\ }\bibfield  {title} {\bibinfo {title} {{Expander Codes: The Sipser-Spielman Construction}},\ }\href {https://doi.org/10.1007/BF02835890} {\bibfield  {journal} {\bibinfo  {journal} {Resonance}\ }\textbf {\bibinfo {volume} {10}},\ \bibinfo {pages} {25} (\bibinfo {year} {2005})}\BibitemShut {NoStop}%
\bibitem [{\citenamefont {Tian}\ \emph {et~al.}(2004)\citenamefont {Tian}, \citenamefont {Jones}, \citenamefont {Villasenor},\ and\ \citenamefont {Wesel}}]{tian2004selective}%
  \BibitemOpen
  \bibfield  {author} {\bibinfo {author} {\bibfnamefont {T.}~\bibnamefont {Tian}}, \bibinfo {author} {\bibfnamefont {C.~R.}\ \bibnamefont {Jones}}, \bibinfo {author} {\bibfnamefont {J.~D.}\ \bibnamefont {Villasenor}},\ and\ \bibinfo {author} {\bibfnamefont {R.~D.}\ \bibnamefont {Wesel}},\ }\bibfield  {title} {\bibinfo {title} {Selective avoidance of cycles in irregular {LDPC} code construction},\ }\href@noop {} {\bibfield  {journal} {\bibinfo  {journal} {IEEE Transactions on Communications}\ }\textbf {\bibinfo {volume} {52}},\ \bibinfo {pages} {1242} (\bibinfo {year} {2004})}\BibitemShut {NoStop}%
\bibitem [{\citenamefont {Bravyi}\ \emph {et~al.}(2010)\citenamefont {Bravyi}, \citenamefont {Poulin},\ and\ \citenamefont {Terhal}}]{bravyi2010tradeoffs}%
  \BibitemOpen
  \bibfield  {author} {\bibinfo {author} {\bibfnamefont {S.}~\bibnamefont {Bravyi}}, \bibinfo {author} {\bibfnamefont {D.}~\bibnamefont {Poulin}},\ and\ \bibinfo {author} {\bibfnamefont {B.}~\bibnamefont {Terhal}},\ }\bibfield  {title} {\bibinfo {title} {Tradeoffs for reliable quantum information storage in 2d systems},\ }\href {https://doi.org/10.1103/PhysRevLett.104.050503} {\bibfield  {journal} {\bibinfo  {journal} {Physical review letters}\ }\textbf {\bibinfo {volume} {104}},\ \bibinfo {pages} {050503} (\bibinfo {year} {2010})}\BibitemShut {NoStop}%
\bibitem [{\citenamefont {Cramwinckel}\ \emph {et~al.}(2023)\citenamefont {Cramwinckel}, \citenamefont {Roijackers}, \citenamefont {Baart}, \citenamefont {Minkes}, \citenamefont {Ruscio}, \citenamefont {Miller}, \citenamefont {Boothby}, \citenamefont {Tjhai}, \citenamefont {Joyner},\ and\ \citenamefont {Fields}}]{guava}%
  \BibitemOpen
  \bibfield  {author} {\bibinfo {author} {\bibfnamefont {J.}~\bibnamefont {Cramwinckel}}, \bibinfo {author} {\bibfnamefont {E.}~\bibnamefont {Roijackers}}, \bibinfo {author} {\bibfnamefont {R.}~\bibnamefont {Baart}}, \bibinfo {author} {\bibfnamefont {E.}~\bibnamefont {Minkes}}, \bibinfo {author} {\bibfnamefont {L.}~\bibnamefont {Ruscio}}, \bibinfo {author} {\bibfnamefont {R.~L.}\ \bibnamefont {Miller}}, \bibinfo {author} {\bibfnamefont {T.}~\bibnamefont {Boothby}}, \bibinfo {author} {\bibfnamefont {C.}~\bibnamefont {Tjhai}}, \bibinfo {author} {\bibfnamefont {D.}~\bibnamefont {Joyner}},\ and\ \bibinfo {author} {\bibfnamefont {J.}~\bibnamefont {Fields}},\ }\href@noop {} {\bibinfo {title} {Guava}} (\bibinfo {year} {2023}),\ \bibinfo {note} {\url{https://gap-packages.github.io/guava/}}\BibitemShut {NoStop}%
\bibitem [{\citenamefont {Roffe}\ \emph {et~al.}(2020)\citenamefont {Roffe}, \citenamefont {White}, \citenamefont {Burton},\ and\ \citenamefont {Campbell}}]{roffe2020decoding}%
  \BibitemOpen
  \bibfield  {author} {\bibinfo {author} {\bibfnamefont {J.}~\bibnamefont {Roffe}}, \bibinfo {author} {\bibfnamefont {D.~R.}\ \bibnamefont {White}}, \bibinfo {author} {\bibfnamefont {S.}~\bibnamefont {Burton}},\ and\ \bibinfo {author} {\bibfnamefont {E.}~\bibnamefont {Campbell}},\ }\bibfield  {title} {\bibinfo {title} {Decoding across the quantum low-density parity-check code landscape},\ }\href {https://doi.org/10.1103/PhysRevResearch.2.043423} {\bibfield  {journal} {\bibinfo  {journal} {Phys. Rev. Res.}\ }\textbf {\bibinfo {volume} {2}},\ \bibinfo {pages} {043423} (\bibinfo {year} {2020})}\BibitemShut {NoStop}%
\bibitem [{\citenamefont {Bombin}\ and\ \citenamefont {Martin-Delgado}(2007)}]{bombin2007optimal}%
  \BibitemOpen
  \bibfield  {author} {\bibinfo {author} {\bibfnamefont {H.}~\bibnamefont {Bombin}}\ and\ \bibinfo {author} {\bibfnamefont {M.~A.}\ \bibnamefont {Martin-Delgado}},\ }\bibfield  {title} {\bibinfo {title} {Optimal resources for topological two-dimensional stabilizer codes: Comparative study},\ }\href {https://doi.org/10.1103/PhysRevA.76.012305} {\bibfield  {journal} {\bibinfo  {journal} {Phys. Rev. A}\ }\textbf {\bibinfo {volume} {76}},\ \bibinfo {pages} {012305} (\bibinfo {year} {2007})}\BibitemShut {NoStop}%
\bibitem [{\citenamefont {Fukui}\ \emph {et~al.}(2018)\citenamefont {Fukui}, \citenamefont {Tomita}, \citenamefont {Okamoto},\ and\ \citenamefont {Fujii}}]{fukui2018high}%
  \BibitemOpen
  \bibfield  {author} {\bibinfo {author} {\bibfnamefont {K.}~\bibnamefont {Fukui}}, \bibinfo {author} {\bibfnamefont {A.}~\bibnamefont {Tomita}}, \bibinfo {author} {\bibfnamefont {A.}~\bibnamefont {Okamoto}},\ and\ \bibinfo {author} {\bibfnamefont {K.}~\bibnamefont {Fujii}},\ }\bibfield  {title} {\bibinfo {title} {High-threshold fault-tolerant quantum computation with analog quantum error correction},\ }\href {https://doi.org/10.1103/PhysRevX.8.021054} {\bibfield  {journal} {\bibinfo  {journal} {Phys. Rev. X}\ }\textbf {\bibinfo {volume} {8}},\ \bibinfo {pages} {021054} (\bibinfo {year} {2018})}\BibitemShut {NoStop}%
\bibitem [{\citenamefont {Vuillot}\ \emph {et~al.}(2019)\citenamefont {Vuillot}, \citenamefont {Asasi}, \citenamefont {Wang}, \citenamefont {Pryadko},\ and\ \citenamefont {Terhal}}]{vuillot2019quantum}%
  \BibitemOpen
  \bibfield  {author} {\bibinfo {author} {\bibfnamefont {C.}~\bibnamefont {Vuillot}}, \bibinfo {author} {\bibfnamefont {H.}~\bibnamefont {Asasi}}, \bibinfo {author} {\bibfnamefont {Y.}~\bibnamefont {Wang}}, \bibinfo {author} {\bibfnamefont {L.~P.}\ \bibnamefont {Pryadko}},\ and\ \bibinfo {author} {\bibfnamefont {B.~M.}\ \bibnamefont {Terhal}},\ }\bibfield  {title} {\bibinfo {title} {Quantum error correction with the toric gottesman-kitaev-preskill code},\ }\href {https://doi.org/10.1103/PhysRevA.99.032344} {\bibfield  {journal} {\bibinfo  {journal} {Phys. Rev. A}\ }\textbf {\bibinfo {volume} {99}},\ \bibinfo {pages} {032344} (\bibinfo {year} {2019})}\BibitemShut {NoStop}%
\bibitem [{\citenamefont {Noh}\ and\ \citenamefont {Chamberland}(2020)}]{noh2020fault}%
  \BibitemOpen
  \bibfield  {author} {\bibinfo {author} {\bibfnamefont {K.}~\bibnamefont {Noh}}\ and\ \bibinfo {author} {\bibfnamefont {C.}~\bibnamefont {Chamberland}},\ }\bibfield  {title} {\bibinfo {title} {Fault-tolerant bosonic quantum error correction with the surface--gottesman-kitaev-preskill code},\ }\href {https://doi.org/10.1103/PhysRevA.101.012316} {\bibfield  {journal} {\bibinfo  {journal} {Phys. Rev. A}\ }\textbf {\bibinfo {volume} {101}},\ \bibinfo {pages} {012316} (\bibinfo {year} {2020})}\BibitemShut {NoStop}%
\bibitem [{\citenamefont {H\"anggli}\ \emph {et~al.}(2020)\citenamefont {H\"anggli}, \citenamefont {Heinze},\ and\ \citenamefont {K\"onig}}]{hanggli2020enhanced}%
  \BibitemOpen
  \bibfield  {author} {\bibinfo {author} {\bibfnamefont {L.}~\bibnamefont {H\"anggli}}, \bibinfo {author} {\bibfnamefont {M.}~\bibnamefont {Heinze}},\ and\ \bibinfo {author} {\bibfnamefont {R.}~\bibnamefont {K\"onig}},\ }\bibfield  {title} {\bibinfo {title} {Enhanced noise resilience of the surface--gottesman-kitaev-preskill code via designed bias},\ }\href {https://doi.org/10.1103/PhysRevA.102.052408} {\bibfield  {journal} {\bibinfo  {journal} {Phys. Rev. A}\ }\textbf {\bibinfo {volume} {102}},\ \bibinfo {pages} {052408} (\bibinfo {year} {2020})}\BibitemShut {NoStop}%
\bibitem [{\citenamefont {Zhang}\ \emph {et~al.}(2021)\citenamefont {Zhang}, \citenamefont {Zhao}, \citenamefont {Wu},\ and\ \citenamefont {Guo}}]{zhang2021quantum}%
  \BibitemOpen
  \bibfield  {author} {\bibinfo {author} {\bibfnamefont {J.}~\bibnamefont {Zhang}}, \bibinfo {author} {\bibfnamefont {J.}~\bibnamefont {Zhao}}, \bibinfo {author} {\bibfnamefont {Y.-C.}\ \bibnamefont {Wu}},\ and\ \bibinfo {author} {\bibfnamefont {G.-P.}\ \bibnamefont {Guo}},\ }\bibfield  {title} {\bibinfo {title} {Quantum error correction with the color-gottesman-kitaev-preskill code},\ }\href {https://doi.org/10.1103/PhysRevA.104.062434} {\bibfield  {journal} {\bibinfo  {journal} {Phys. Rev. A}\ }\textbf {\bibinfo {volume} {104}},\ \bibinfo {pages} {062434} (\bibinfo {year} {2021})}\BibitemShut {NoStop}%
\bibitem [{\citenamefont {Noh}\ \emph {et~al.}(2022)\citenamefont {Noh}, \citenamefont {Chamberland},\ and\ \citenamefont {Brand\~ao}}]{noh2022low}%
  \BibitemOpen
  \bibfield  {author} {\bibinfo {author} {\bibfnamefont {K.}~\bibnamefont {Noh}}, \bibinfo {author} {\bibfnamefont {C.}~\bibnamefont {Chamberland}},\ and\ \bibinfo {author} {\bibfnamefont {F.~G.}\ \bibnamefont {Brand\~ao}},\ }\bibfield  {title} {\bibinfo {title} {Low-overhead fault-tolerant quantum error correction with the surface-gkp code},\ }\href {https://doi.org/10.1103/PRXQuantum.3.010315} {\bibfield  {journal} {\bibinfo  {journal} {PRX Quantum}\ }\textbf {\bibinfo {volume} {3}},\ \bibinfo {pages} {010315} (\bibinfo {year} {2022})}\BibitemShut {NoStop}%
\bibitem [{\citenamefont {Zhang}\ \emph {et~al.}(2023)\citenamefont {Zhang}, \citenamefont {Wu},\ and\ \citenamefont {Guo}}]{zhang2023concatenation}%
  \BibitemOpen
  \bibfield  {author} {\bibinfo {author} {\bibfnamefont {J.}~\bibnamefont {Zhang}}, \bibinfo {author} {\bibfnamefont {Y.-C.}\ \bibnamefont {Wu}},\ and\ \bibinfo {author} {\bibfnamefont {G.-P.}\ \bibnamefont {Guo}},\ }\bibfield  {title} {\bibinfo {title} {Concatenation of the gottesman-kitaev-preskill code with the xzzx surface code},\ }\href {https://doi.org/10.1103/PhysRevA.107.062408} {\bibfield  {journal} {\bibinfo  {journal} {Phys. Rev. A}\ }\textbf {\bibinfo {volume} {107}},\ \bibinfo {pages} {062408} (\bibinfo {year} {2023})}\BibitemShut {NoStop}%
\bibitem [{\citenamefont {Walshe}\ \emph {et~al.}(2020)\citenamefont {Walshe}, \citenamefont {Baragiola}, \citenamefont {Alexander},\ and\ \citenamefont {Menicucci}}]{walshe2020continuous}%
  \BibitemOpen
  \bibfield  {author} {\bibinfo {author} {\bibfnamefont {B.~W.}\ \bibnamefont {Walshe}}, \bibinfo {author} {\bibfnamefont {B.~Q.}\ \bibnamefont {Baragiola}}, \bibinfo {author} {\bibfnamefont {R.~N.}\ \bibnamefont {Alexander}},\ and\ \bibinfo {author} {\bibfnamefont {N.~C.}\ \bibnamefont {Menicucci}},\ }\bibfield  {title} {\bibinfo {title} {Continuous-variable gate teleportation and bosonic-code error correction},\ }\href {https://doi.org/10.1103/PhysRevA.102.062411} {\bibfield  {journal} {\bibinfo  {journal} {Physical Review A}\ }\textbf {\bibinfo {volume} {102}},\ \bibinfo {pages} {062411} (\bibinfo {year} {2020})}\BibitemShut {NoStop}%
\bibitem [{\citenamefont {Brown}\ \emph {et~al.}(2001)\citenamefont {Brown}, \citenamefont {Cai},\ and\ \citenamefont {DasGupta}}]{brown2001interval}%
  \BibitemOpen
  \bibfield  {author} {\bibinfo {author} {\bibfnamefont {L.~D.}\ \bibnamefont {Brown}}, \bibinfo {author} {\bibfnamefont {T.~T.}\ \bibnamefont {Cai}},\ and\ \bibinfo {author} {\bibfnamefont {A.}~\bibnamefont {DasGupta}},\ }\bibfield  {title} {\bibinfo {title} {{Interval Estimation for a Binomial Proportion}},\ }\href {https://doi.org/10.1214/ss/1009213286} {\bibfield  {journal} {\bibinfo  {journal} {Statistical Science}\ }\textbf {\bibinfo {volume} {16}},\ \bibinfo {pages} {101 } (\bibinfo {year} {2001})}\BibitemShut {NoStop}%
\bibitem [{\citenamefont {Kribs}\ \emph {et~al.}(2005)\citenamefont {Kribs}, \citenamefont {Laflamme},\ and\ \citenamefont {Poulin}}]{kribs2005unified}%
  \BibitemOpen
  \bibfield  {author} {\bibinfo {author} {\bibfnamefont {D.}~\bibnamefont {Kribs}}, \bibinfo {author} {\bibfnamefont {R.}~\bibnamefont {Laflamme}},\ and\ \bibinfo {author} {\bibfnamefont {D.}~\bibnamefont {Poulin}},\ }\bibfield  {title} {\bibinfo {title} {Unified and generalized approach to quantum error correction},\ }\href {https://doi.org/10.1103/PhysRevLett.94.180501} {\bibfield  {journal} {\bibinfo  {journal} {Phys. Rev. Lett.}\ }\textbf {\bibinfo {volume} {94}},\ \bibinfo {pages} {180501} (\bibinfo {year} {2005})}\BibitemShut {NoStop}%
\bibitem [{\citenamefont {Poulin}(2005)}]{poulin2005stabilizer}%
  \BibitemOpen
  \bibfield  {author} {\bibinfo {author} {\bibfnamefont {D.}~\bibnamefont {Poulin}},\ }\bibfield  {title} {\bibinfo {title} {Stabilizer formalism for operator quantum error correction},\ }\href {https://doi.org/10.1103/PhysRevLett.95.230504} {\bibfield  {journal} {\bibinfo  {journal} {Phys. Rev. Lett.}\ }\textbf {\bibinfo {volume} {95}},\ \bibinfo {pages} {230504} (\bibinfo {year} {2005})}\BibitemShut {NoStop}%
\bibitem [{\citenamefont {Bacon}(2006)}]{bacon2006operator}%
  \BibitemOpen
  \bibfield  {author} {\bibinfo {author} {\bibfnamefont {D.}~\bibnamefont {Bacon}},\ }\bibfield  {title} {\bibinfo {title} {Operator quantum error-correcting subsystems for self-correcting quantum memories},\ }\href {https://doi.org/10.1103/PhysRevA.73.012340} {\bibfield  {journal} {\bibinfo  {journal} {Phys. Rev. A}\ }\textbf {\bibinfo {volume} {73}},\ \bibinfo {pages} {012340} (\bibinfo {year} {2006})}\BibitemShut {NoStop}%
\bibitem [{\citenamefont {Bombin}(2010)}]{bombin2010topological}%
  \BibitemOpen
  \bibfield  {author} {\bibinfo {author} {\bibfnamefont {H.}~\bibnamefont {Bombin}},\ }\bibfield  {title} {\bibinfo {title} {Topological subsystem codes},\ }\href {https://doi.org/10.1103/PhysRevA.81.032301} {\bibfield  {journal} {\bibinfo  {journal} {Phys. Rev. A}\ }\textbf {\bibinfo {volume} {81}},\ \bibinfo {pages} {032301} (\bibinfo {year} {2010})}\BibitemShut {NoStop}%
\bibitem [{\citenamefont {Suchara}\ \emph {et~al.}(2011)\citenamefont {Suchara}, \citenamefont {Bravyi},\ and\ \citenamefont {Terhal}}]{suchara2011constructions}%
  \BibitemOpen
  \bibfield  {author} {\bibinfo {author} {\bibfnamefont {M.}~\bibnamefont {Suchara}}, \bibinfo {author} {\bibfnamefont {S.}~\bibnamefont {Bravyi}},\ and\ \bibinfo {author} {\bibfnamefont {B.}~\bibnamefont {Terhal}},\ }\bibfield  {title} {\bibinfo {title} {Constructions and noise threshold of topological subsystem codes},\ }\href {https://doi.org/10.1088/1751-8113/44/15/155301} {\bibfield  {journal} {\bibinfo  {journal} {Journal of Physics A: Mathematical and Theoretical}\ }\textbf {\bibinfo {volume} {44}},\ \bibinfo {pages} {155301} (\bibinfo {year} {2011})}\BibitemShut {NoStop}%
\bibitem [{\citenamefont {Bravyi}\ \emph {et~al.}(2013)\citenamefont {Bravyi}, \citenamefont {Duclos-Cianci}, \citenamefont {Poulin},\ and\ \citenamefont {Suchara}}]{bravyi2013subsystem}%
  \BibitemOpen
  \bibfield  {author} {\bibinfo {author} {\bibfnamefont {S.}~\bibnamefont {Bravyi}}, \bibinfo {author} {\bibfnamefont {G.}~\bibnamefont {Duclos-Cianci}}, \bibinfo {author} {\bibfnamefont {D.}~\bibnamefont {Poulin}},\ and\ \bibinfo {author} {\bibfnamefont {M.}~\bibnamefont {Suchara}},\ }\bibfield  {title} {\bibinfo {title} {Subsystem surface codes with three-qubit check operators},\ }\href {https://doi.org/10.26421/QIC13.11-12-4} {\bibfield  {journal} {\bibinfo  {journal} {Quantum Info. Comput.}\ }\textbf {\bibinfo {volume} {13}},\ \bibinfo {pages} {963–985} (\bibinfo {year} {2013})}\BibitemShut {NoStop}%
\bibitem [{\citenamefont {Bombín}(2015)}]{bombin2015gauge}%
  \BibitemOpen
  \bibfield  {author} {\bibinfo {author} {\bibfnamefont {H.}~\bibnamefont {Bombín}},\ }\bibfield  {title} {\bibinfo {title} {Gauge color codes: optimal transversal gates and gauge fixing in topological stabilizer codes},\ }\href {https://doi.org/10.1088/1367-2630/17/8/083002} {\bibfield  {journal} {\bibinfo  {journal} {New Journal of Physics}\ }\textbf {\bibinfo {volume} {17}},\ \bibinfo {pages} {083002} (\bibinfo {year} {2015})}\BibitemShut {NoStop}%
\bibitem [{\citenamefont {Bravyi}\ and\ \citenamefont {Cross}(2015)}]{bravyi2015doubled}%
  \BibitemOpen
  \bibfield  {author} {\bibinfo {author} {\bibfnamefont {S.}~\bibnamefont {Bravyi}}\ and\ \bibinfo {author} {\bibfnamefont {A.}~\bibnamefont {Cross}},\ }\href@noop {} {\bibinfo {title} {Doubled color codes}} (\bibinfo {year} {2015}),\ \Eprint {https://arxiv.org/abs/1509.03239} {arXiv:1509.03239} \BibitemShut {NoStop}%
\bibitem [{\citenamefont {Jochym-O'Connor}\ and\ \citenamefont {Bartlett}(2016)}]{oconnor2016stacked}%
  \BibitemOpen
  \bibfield  {author} {\bibinfo {author} {\bibfnamefont {T.}~\bibnamefont {Jochym-O'Connor}}\ and\ \bibinfo {author} {\bibfnamefont {S.~D.}\ \bibnamefont {Bartlett}},\ }\bibfield  {title} {\bibinfo {title} {Stacked codes: Universal fault-tolerant quantum computation in a two-dimensional layout},\ }\href {https://doi.org/10.1103/PhysRevA.93.022323} {\bibfield  {journal} {\bibinfo  {journal} {Phys. Rev. A}\ }\textbf {\bibinfo {volume} {93}},\ \bibinfo {pages} {022323} (\bibinfo {year} {2016})}\BibitemShut {NoStop}%
\bibitem [{\citenamefont {Jones}\ \emph {et~al.}(2016)\citenamefont {Jones}, \citenamefont {Brooks},\ and\ \citenamefont {Harrington}}]{jones2016gauge}%
  \BibitemOpen
  \bibfield  {author} {\bibinfo {author} {\bibfnamefont {C.}~\bibnamefont {Jones}}, \bibinfo {author} {\bibfnamefont {P.}~\bibnamefont {Brooks}},\ and\ \bibinfo {author} {\bibfnamefont {J.}~\bibnamefont {Harrington}},\ }\bibfield  {title} {\bibinfo {title} {Gauge color codes in two dimensions},\ }\href {https://doi.org/10.1103/PhysRevA.93.052332} {\bibfield  {journal} {\bibinfo  {journal} {Phys. Rev. A}\ }\textbf {\bibinfo {volume} {93}},\ \bibinfo {pages} {052332} (\bibinfo {year} {2016})}\BibitemShut {NoStop}%
\bibitem [{\citenamefont {Higgott}\ and\ \citenamefont {Breuckmann}(2021)}]{higgott2021subsystem}%
  \BibitemOpen
  \bibfield  {author} {\bibinfo {author} {\bibfnamefont {O.}~\bibnamefont {Higgott}}\ and\ \bibinfo {author} {\bibfnamefont {N.~P.}\ \bibnamefont {Breuckmann}},\ }\bibfield  {title} {\bibinfo {title} {Subsystem codes with high thresholds by gauge fixing and reduced qubit overhead},\ }\href {https://doi.org/10.1103/PhysRevX.11.031039} {\bibfield  {journal} {\bibinfo  {journal} {Phys. Rev. X}\ }\textbf {\bibinfo {volume} {11}},\ \bibinfo {pages} {031039} (\bibinfo {year} {2021})}\BibitemShut {NoStop}%
\bibitem [{\citenamefont {Kubica}\ and\ \citenamefont {Vasmer}(2022)}]{kubica2022single}%
  \BibitemOpen
  \bibfield  {author} {\bibinfo {author} {\bibfnamefont {A.}~\bibnamefont {Kubica}}\ and\ \bibinfo {author} {\bibfnamefont {M.}~\bibnamefont {Vasmer}},\ }\bibfield  {title} {\bibinfo {title} {Single-shot quantum error correction with the three-dimensional subsystem toric code},\ }\href {https://doi.org/10.1038/s41467-022-33923-4} {\bibfield  {journal} {\bibinfo  {journal} {Nature Communications}\ }\textbf {\bibinfo {volume} {13}},\ \bibinfo {pages} {6272} (\bibinfo {year} {2022})}\BibitemShut {NoStop}%
\bibitem [{\citenamefont {Krishna}\ and\ \citenamefont {Poulin}(2021)}]{krishna2021fault}%
  \BibitemOpen
  \bibfield  {author} {\bibinfo {author} {\bibfnamefont {A.}~\bibnamefont {Krishna}}\ and\ \bibinfo {author} {\bibfnamefont {D.}~\bibnamefont {Poulin}},\ }\bibfield  {title} {\bibinfo {title} {Fault-tolerant gates on hypergraph product codes},\ }\href {https://doi.org/10.1103/PhysRevX.11.011023} {\bibfield  {journal} {\bibinfo  {journal} {Physical Review X}\ }\textbf {\bibinfo {volume} {11}},\ \bibinfo {pages} {011023} (\bibinfo {year} {2021})}\BibitemShut {NoStop}%
\bibitem [{\citenamefont {Cohen}\ \emph {et~al.}(2022)\citenamefont {Cohen}, \citenamefont {Kim}, \citenamefont {Bartlett},\ and\ \citenamefont {Brown}}]{cohen2022}%
  \BibitemOpen
  \bibfield  {author} {\bibinfo {author} {\bibfnamefont {L.~Z.}\ \bibnamefont {Cohen}}, \bibinfo {author} {\bibfnamefont {I.~H.}\ \bibnamefont {Kim}}, \bibinfo {author} {\bibfnamefont {S.~D.}\ \bibnamefont {Bartlett}},\ and\ \bibinfo {author} {\bibfnamefont {B.~J.}\ \bibnamefont {Brown}},\ }\bibfield  {title} {\bibinfo {title} {Low-overhead fault-tolerant quantum computing using long-range connectivity},\ }\href {https://doi.org/10.1126/sciadv.abn1717} {\bibfield  {journal} {\bibinfo  {journal} {Sci. Adv.}\ }\textbf {\bibinfo {volume} {8}},\ \bibinfo {pages} {eabn1717} (\bibinfo {year} {2022})}\BibitemShut {NoStop}%
\bibitem [{\citenamefont {Quintavalle}\ \emph {et~al.}(2023)\citenamefont {Quintavalle}, \citenamefont {Webster},\ and\ \citenamefont {Vasmer}}]{quintavalle2023}%
  \BibitemOpen
  \bibfield  {author} {\bibinfo {author} {\bibfnamefont {A.~O.}\ \bibnamefont {Quintavalle}}, \bibinfo {author} {\bibfnamefont {P.}~\bibnamefont {Webster}},\ and\ \bibinfo {author} {\bibfnamefont {M.}~\bibnamefont {Vasmer}},\ }\bibfield  {title} {\bibinfo {title} {Partitioning qubits in hypergraph product codes to implement logical gates},\ }\href {https://doi.org/10.22331/q-2023-10-24-1153} {\bibfield  {journal} {\bibinfo  {journal} {Quantum}\ }\textbf {\bibinfo {volume} {7}},\ \bibinfo {pages} {1153} (\bibinfo {year} {2023})}\BibitemShut {NoStop}%
\bibitem [{\citenamefont {Breuckmann}\ and\ \citenamefont {Burton}(2022)}]{breuckmann2022fold}%
  \BibitemOpen
  \bibfield  {author} {\bibinfo {author} {\bibfnamefont {N.~P.}\ \bibnamefont {Breuckmann}}\ and\ \bibinfo {author} {\bibfnamefont {S.}~\bibnamefont {Burton}},\ }\href@noop {} {\bibinfo {title} {Fold-transversal clifford gates for quantum codes}} (\bibinfo {year} {2022}),\ \Eprint {https://arxiv.org/abs/2202.06647} {arXiv:2202.06647} \BibitemShut {NoStop}%
\bibitem [{\citenamefont {Huang}\ \emph {et~al.}(2022)\citenamefont {Huang}, \citenamefont {Jochym-O'Connor},\ and\ \citenamefont {Yoder}}]{huang2022homomorphic}%
  \BibitemOpen
  \bibfield  {author} {\bibinfo {author} {\bibfnamefont {S.}~\bibnamefont {Huang}}, \bibinfo {author} {\bibfnamefont {T.}~\bibnamefont {Jochym-O'Connor}},\ and\ \bibinfo {author} {\bibfnamefont {T.~J.}\ \bibnamefont {Yoder}},\ }\href@noop {} {\bibinfo {title} {Homomorphic logical measurements}} (\bibinfo {year} {2022}),\ \Eprint {https://arxiv.org/abs/2211.03625} {arXiv:2211.03625 [quant-ph]} \BibitemShut {NoStop}%
\bibitem [{\citenamefont {Evra}\ \emph {et~al.}(2022)\citenamefont {Evra}, \citenamefont {Kaufman},\ and\ \citenamefont {Z{\'e}mor}}]{evra2022decodable}%
  \BibitemOpen
  \bibfield  {author} {\bibinfo {author} {\bibfnamefont {S.}~\bibnamefont {Evra}}, \bibinfo {author} {\bibfnamefont {T.}~\bibnamefont {Kaufman}},\ and\ \bibinfo {author} {\bibfnamefont {G.}~\bibnamefont {Z{\'e}mor}},\ }\bibfield  {title} {\bibinfo {title} {Decodable quantum {LDPC} codes beyond the n distance barrier using high-dimensional expanders},\ }\href {https://doi.org/10.1137/20M1383689} {\bibfield  {journal} {\bibinfo  {journal} {SIAM Journal on Computing}\ ,\ \bibinfo {pages} {FOCS20}} (\bibinfo {year} {2022})}\BibitemShut {NoStop}%
\bibitem [{\citenamefont {Cross}\ \emph {et~al.}(2022)\citenamefont {Cross}, \citenamefont {He}, \citenamefont {Natarajan}, \citenamefont {Szegedy},\ and\ \citenamefont {Zhu}}]{cross2022quantum}%
  \BibitemOpen
  \bibfield  {author} {\bibinfo {author} {\bibfnamefont {A.}~\bibnamefont {Cross}}, \bibinfo {author} {\bibfnamefont {Z.}~\bibnamefont {He}}, \bibinfo {author} {\bibfnamefont {A.}~\bibnamefont {Natarajan}}, \bibinfo {author} {\bibfnamefont {M.}~\bibnamefont {Szegedy}},\ and\ \bibinfo {author} {\bibfnamefont {G.}~\bibnamefont {Zhu}},\ }\href@noop {} {\bibinfo {title} {Quantum locally testable code with exotic parameters}} (\bibinfo {year} {2022}),\ \Eprint {https://arxiv.org/abs/2209.11405} {2209.11405} \BibitemShut {NoStop}%
\bibitem [{\citenamefont {Wills}\ \emph {et~al.}(2023{\natexlab{c}})\citenamefont {Wills}, \citenamefont {Lin},\ and\ \citenamefont {Hsieh}}]{wills2023general}%
  \BibitemOpen
  \bibfield  {author} {\bibinfo {author} {\bibfnamefont {A.}~\bibnamefont {Wills}}, \bibinfo {author} {\bibfnamefont {T.-C.}\ \bibnamefont {Lin}},\ and\ \bibinfo {author} {\bibfnamefont {M.-H.}\ \bibnamefont {Hsieh}},\ }\href@noop {} {\bibinfo {title} {General distance balancing for quantum locally testable codes}} (\bibinfo {year} {2023}{\natexlab{c}}),\ \Eprint {https://arxiv.org/abs/2305.00689} {arXiv:2305.00689} \BibitemShut {NoStop}%
\bibitem [{\citenamefont {Rotman}(2009)}]{rotman2009introduction}%
  \BibitemOpen
  \bibfield  {author} {\bibinfo {author} {\bibfnamefont {J.}~\bibnamefont {Rotman}},\ }\href@noop {} {\emph {\bibinfo {title} {An introduction to homological algebra}}},\ Vol.~\bibinfo {volume} {2}\ (\bibinfo  {publisher} {Springer},\ \bibinfo {year} {2009})\BibitemShut {NoStop}%
\bibitem [{\citenamefont {Bocharova}\ \emph {et~al.}(2009)\citenamefont {Bocharova}, \citenamefont {Kudryashov},\ and\ \citenamefont {Satyukov}}]{bocharova2009}%
  \BibitemOpen
  \bibfield  {author} {\bibinfo {author} {\bibfnamefont {I.~E.}\ \bibnamefont {Bocharova}}, \bibinfo {author} {\bibfnamefont {B.~D.}\ \bibnamefont {Kudryashov}},\ and\ \bibinfo {author} {\bibfnamefont {R.~V.}\ \bibnamefont {Satyukov}},\ }\bibfield  {title} {\bibinfo {title} {Graph-based convolutional and block {{LDPC}} codes},\ }\href {https://doi.org/10.1134/S0032946009040061} {\bibfield  {journal} {\bibinfo  {journal} {Probl Inf Transm}\ }\textbf {\bibinfo {volume} {45}},\ \bibinfo {pages} {357} (\bibinfo {year} {2009})}\BibitemShut {NoStop}%
\bibitem [{\citenamefont {Smarandache}\ and\ \citenamefont {Vontobel}(2012)}]{smarandache2012quasi}%
  \BibitemOpen
  \bibfield  {author} {\bibinfo {author} {\bibfnamefont {R.}~\bibnamefont {Smarandache}}\ and\ \bibinfo {author} {\bibfnamefont {P.~O.}\ \bibnamefont {Vontobel}},\ }\bibfield  {title} {\bibinfo {title} {Quasi-{{Cyclic LDPC Codes}}: {{Influence}} of {{Proto-}} and {{Tanner-Graph Structure}} on {{Minimum Hamming Distance Upper Bounds}}},\ }\href {https://doi.org/10.1109/TIT.2011.2173244} {\bibfield  {journal} {\bibinfo  {journal} {IEEE Trans. Inform. Theory}\ }\textbf {\bibinfo {volume} {58}},\ \bibinfo {pages} {585} (\bibinfo {year} {2012})}\BibitemShut {NoStop}%
\end{thebibliography}%

\end{document}